\documentclass[11pt,a4paper, reqno]{amsart}

\usepackage[utf8]{inputenc}
\usepackage[T1]{fontenc} 
\usepackage[english]{babel} 
\usepackage{amsmath,amssymb,amsthm}
\usepackage{graphicx}
\usepackage{geometry}
\usepackage{xcolor}
\usepackage{enumitem}
\usepackage{hyperref}
\usepackage{braket,stix} 
\usepackage{dsfont}

\usepackage[normalem]{ulem}

\usepackage{amssymb,amsthm,amsmath,amsxtra,dsfont,bbm,stix,color}
\usepackage{hyperref,cleveref} 

\usepackage{environ} % for math environment
\usepackage{pgfkeys} % Environment parameter keys

\usepackage{mathtools} % enhance math
\usepackage{stmaryrd} % a font
\usepackage{braket} % Bracket for quantum mechanics
\usepackage{enumitem}
\usepackage{mathrsfs}

\allowdisplaybreaks

\numberwithin{equation}{section}
\makeatletter

\makeatletter
\def\@tocline#1#2#3#4#5#6#7{\relax
  \ifnum #1>\c@tocdepth % then omit
  \else
    \par \addpenalty\@secpenalty\addvspace{#2}%
    \begingroup \hyphenpenalty\@M
    \@ifempty{#4}{%
      \@tempdima\csname r@tocindent\number#1\endcsname\relax
    }{%
      \@tempdima#4\relax
    }%
    \parindent\z@ \leftskip#3\relax \advance\leftskip\@tempdima\relax
    \rightskip\@pnumwidth plus4em \parfillskip-\@pnumwidth
    #5\leavevmode\hskip-\@tempdima
      \ifcase #1
       \or\or \hskip 1em \or \hskip 2em \else \hskip 3em \fi%
      #6\nobreak\relax
      \dotfill
      \hbox to\@pnumwidth{\@tocpagenum{#7}}
    \par
    \nobreak
    \endgroup
  \fi}
\makeatother

\newcommand{\tr}[1]{\textnormal{Tr}\left[ #1 \right]}
\newcommand{\ps}[2]{\left\langle #1 , #2 \right\rangle}

\newcommand{\dt}{\frac{\textnormal{d}}{\textnormal{d}t}}
\newcommand{\com}[2]{\left[#1 , #2 \right]}
\newcommand{\Ree}{\textnormal{Re}}
\newcommand{\Imm}{\textnormal{Im}}
\newcommand{\norm}[1]{\left \lVert #1 \right \lVert}

\usepackage[
  backend=biber,
  style=numeric,
  sorting=nyt,
  maxbibnames=99,
  doi=false,
  url=false,
  isbn=false,
  eprint=true,
  giveninits=true
]{biblatex}

\AtEveryBibitem{%
  \ifentrytype{article}{%
    \iffieldundef{journaltitle}
      {}
      {\clearfield{eprint}%
       \clearfield{eprinttype}%
       \clearfield{eprintclass}}
  }{}
}

\DeclareFieldFormat{title}{\mkbibemph{#1\isdot}}
\DeclareFieldFormat{booktitle}{\mkbibemph{#1\isdot}}
\DeclareFieldFormat{maintitle}{\mkbibemph{#1\isdot}}
\DeclareFieldFormat{journaltitle}{#1\isdot}

\newtheorem{theorem}{Theorem}[section]
\newtheorem{corollary}[theorem]{Corollary}
\newtheorem{proposition}[theorem]{Proposition}
\newtheorem{remark}[theorem]{Remark}
\newtheorem{lemma}[theorem]{Lemma}
\newtheorem{definition}[theorem]{Definition}
\newtheorem{assumption}[theorem]{Assumption}
\labelformat{definition}{definition}

\calclayout

\usepackage[colorinlistoftodos]{todonotes}

\usepackage{ulem}

\title{Classical limit of the Pauli--Fierz dynamics}

\author[L. Jougla]{Lucas Jougla}
\address{Constructor university Bremen}
\email{ljougla@constructor.university}

\author[N. Leopold]{Nikolai Leopold}
\address{Constructor university Bremen}
\email{nleopold@constructor.university}

\begin{document}

\maketitle

\begin{abstract}
\noindent We study the Schrödinger evolution generated by the Pauli--Fierz Hamiltonian, a model for nonrelativistic quantum electrodynamics,  in the classical limit $\hbar \rightarrow 0$.  In this regime, we rigorously derive the Newton--Maxwell equations of classical electrodynamics as effective dynamics approximating the time evolution. Our result complements prior work by an alternative derivation that provides explicit estimates on the rate of convergence, justifying the validity of the approximation for a special class of initial data.
\end{abstract}

\tableofcontents

\section{Introduction}

The invention of quantum electrodynamics (QED) led to a new description of the electromagnetic interaction that captures effects beyond classical electromagnetism. In many situations, however, these effects are negligible and classical electromagnetism remains a good approximation.
The goal of this work is to prove with mathematical rigor that, in the classical limit and for suitably chosen initial data, classical electromagnetism emerges as an approximate macroscopic theory from QED.
For the classical model, we take the Newton--Maxwell equations, also known as the Abraham model.
For a mathematically rigorous quantum model, we consider the spinless Pauli--Fierz Hamiltonian, which is obtained by canonical quantization of the Abraham model (see \cite{Spohn_2004}).
Our aim is to rigorously establish the Bohr correspondence principle, which states that the classical theory is recovered in situations where Planck's constant is effectively small and quantum effects are therefore negligible.

The problem has already been investigated in \cite{ammari2022derivationclassicalelectrodynamicscharges,knowles2009limiting}. We believe that the merit of this article is twofold.  On the one hand, we provide, in comparison to known results, explicit error estimates for the validity of the approximation, albeit for a smaller class of initial data.  
On the other hand, our approach -- viewing the system as a collection of tracer particles immersed in a coherent state of photons and proving the result by adapting techniques from \cite{ CARDENAS2022109403,Deckert2014}  -- is very straightforward. This suggests that it might be a good starting point to approach the more challenging open problem of the classical limit of the renormalized Nelson model.  \\

\paragraph{\textbf{Newton--Maxwell equations.}} 

We consider a system of $N \geq 1$ Newtonian particles in $\mathbb{R}^3$ which interact with the electromagnetic field, in the \textit{Coulomb gauge}. For simplicity, we assume the particles masses to be identical and equal to $1/2$. Moreover, their density is described by an even charge distribution, denoted by $\kappa$, for which more precise assumptions will be given later. The dynamics of the $j^{\textnormal{th}}$ particle is described by its position $p_j \in \mathbb{R}^3$ and momentum $q_j \in \mathbb{R}^3$. We  note $p = (p_1, \dots, p_N) \in \mathbb{R}^{3N}$ and $q = (q_1, \dots, q_N) \in \mathbb{R}^{3N}$. The interaction between the particles and the electromagnetic field is described by the vector potential
\begin{equation} \label{eq : def classical field}
\mathbf{A}_{\alpha}(x) = \sum_{\lambda = 1,2} \int_{\mathbb{R}^3} \frac{(2 \pi)^{-3/2}}{\sqrt{2|k|}} \boldsymbol{\epsilon}_\lambda(k)  \left(e^{ik\cdot x} \alpha(k,\lambda) + e^{-ik\cdot x}\overline{\alpha(k,\lambda)} \right) \textnormal{d}k \ . \\
\end{equation}
We also define the electric field by :
\begin{equation} \label{eq : classical electric field}
    \mathbf{E}_{\alpha}(x) = i \sum_{\lambda= 1,2} \int_{\mathbb{R}^3} (2 \pi)^{-3/2}\sqrt{\frac{|k|}{2}} \boldsymbol{\epsilon}_\lambda(k) \left( e^{ik \cdot x} \alpha(k,\lambda) - e^{-ik \cdot x} \overline{\alpha(k,\lambda)} \right) \textnormal{d}k \ . 
\end{equation}
Here, $\alpha : \mathbb{R}^3 \times \mathbb{C}^2 \rightarrow \mathbb{C}$ is a complex field. The two real polarization vectors $\left\{ \boldsymbol{\varepsilon}_\lambda(k) \right\}_{\lambda = 1,2}$ implement the \textit{Coulomb gauge}, that is $\nabla \cdot \mathbf{A}_\alpha = 0$, by requiring that 
\begin{equation} \label{eq : polarization vectors ONB}
    \left\{ \boldsymbol{\varepsilon}_\lambda(k), \frac{k}{|k|} \right\}_{\lambda = 1,2}
\end{equation}
is an orthonormal basis of $\mathbb{R}^3$. In this gauge, the particles interact with each other via the (smeared) Coulomb potential :
\begin{equation} \label{eq : def coulomb potential}
    \begin{split}
        V(x) & = \kappa \ast \kappa \ast |\cdot|^{-1} (x) \ . \\
    \end{split}
\end{equation}
The Newton--Maxwell equations of motion are given by 
\begin{equation} \label{eq : newton motion eqs PF}
\begin{split}
        \partial_t q_j & = 2 \left( p_j - \kappa \ast \mathbf{A}_{\alpha_t}(q_j) \right) \\
        \partial_t p_j & = 2 \sum_{\ell=1}^3 \left( p^\ell_j - \kappa \ast \mathbf{A}^\ell_{\alpha_t}(q_j) \right) \nabla \left( \kappa \ast \mathbf{A}^\ell_{\alpha_t}\right) (q_j) - \sum_{k \neq j} \left( \nabla V \right) \left(q_j - q_k \right) 
        \\
        i \partial_t \alpha_t(k,\lambda) & = |k| \alpha_t(k,\lambda) -2 \sum_{j=1}^N \frac{\mathcal{F}[\kappa](k)}{\sqrt{2 |k|}}  e^{-i k \cdot q_j} \boldsymbol{\epsilon}_{\lambda}(k) \cdot\left( p_j - \kappa \ast \mathbf{A}_{\alpha_t}(q_j) \right) \ , \\
\end{split}
\end{equation}
which are the Hamilton equations associated with the energy 
\begin{equation} \label{eq : energy NM}
    \mathcal{H}(q,p,\alpha) = \sum_{j=1}^N (p_j - \kappa \ast \mathbf{A}_\alpha(q_j))^2 + \sum_{1 \leq j < k \leq N} V(q_j - q_k) + \sum_{\lambda=1,2} \int_{\mathbb{R}^d}|k| |\alpha(k,\lambda)|^2 \textnormal{d}k \ . \\
\end{equation}
In the above,   $\mathcal{F}$ denotes the Fourier transform in $\mathbb{R}^3$,   defined by
\begin{equation}
\mathcal{F}\left[f\right](k) := \frac{1}{(2\pi)^{3/2}} \int_{\mathbb{R}^3} f(x) e^{-ik \cdot x} \textnormal{d}x \\
\end{equation}
for $f \in L^1(\mathbb{R}^3)$, or any suitable space where the definition can be extended. For later purposes, let us define
\begin{equation} \label{eq : def G}
\mathbf{G}_x(k,\lambda) = \frac{\mathcal{F}[\kappa](k)}{\sqrt{2|k|}}\boldsymbol{\epsilon}_{\lambda}(k) e^{-i k \cdot x} \ . \\
\end{equation}
We will denote the classical data described by the Newton--Maxwell system as $u = (q,p,\alpha)$. The precise functional setting to study these equations is given by the spaces $X^\sigma$ and $\dot{X}^\sigma$ which are defined by the norms 
\begin{equation} \label{eq : def X et Xdot sigma}
    \begin{split}
        \norm{u}^2_{X^\sigma} & = \sum_{j=1}^N\left( |q_j|^2 + |p_j|^2 \right) + \norm{\alpha}^2_{\mathfrak{h}^\sigma} \\
        \norm{u}^2_{\dot{X}^\sigma} & = \sum_{j=1}^N\left( |q_j|^2 + |p_j|^2 \right) + \norm{\alpha}^2_{\dot{\mathfrak{h}}^\sigma} \ , \\
    \end{split}
\end{equation}
where $\mathfrak{h}^\sigma$ and $\dot{\mathfrak{h}}^\sigma$ are the weighted $L^2$ spaces defined respectively by the norms :
\begin{equation} \label{eq : def norm alpha}
    \begin{split}
        \norm{\alpha}^2_{\mathfrak{h}^\sigma} & = \sum_{\lambda=1,2} \int_{\mathbb{R}^3} (1+|k|^2)^\sigma |\alpha(k,\lambda)|^2 \textnormal{d}k \\
         \norm{\alpha}^2_{\dot{\mathfrak{h}}^\sigma} & = \sum_{\lambda=1,2} \int_{\mathbb{R}^3} |k|^{2\sigma} |\alpha(k,\lambda)|^2 \textnormal{d}k \ . \\
    \end{split}
\end{equation} \\

\noindent \textbf{Pauli--Fierz Hamiltonian.} The canonical quantization of the Abraham model (see for example \cite{Spohn_2004}) results in the Pauli--Fierz Hamiltonian, which we choose as our mathematical model of non-relativistic quantum electrodynamics. It describes a system of $N$ quantum particles interacting with the quantized electromagnetic field.  The Hilbert space for this setting is the following
\begin{equation}
    \mathfrak{H} = L^2(\mathbb{R}^{3N}) \otimes \mathfrak{F} \ , \\
\end{equation}
where
\begin{equation}
    \mathfrak{F} = \bigoplus_{n \geq 0} \mathfrak{h}^{\otimes_s^n} \ , \quad \mathfrak{h} = L^2(\mathbb{R}^{3}; \mathbb{C}^2) \ . \\
\end{equation}
Here $\mathfrak{F}$ is the \textit{bosonic Fock space} over $\mathfrak{h}$. Elements of $\mathfrak{H}$ can be seen as infinite sequences $\Psi = \left\{\Psi^{(0)}, \Psi^{(1)}, \dots  \right\}$, where for all $n \in \mathbb{N}$,  $\Psi^{(n)} \in  L^2(\mathbb{R}^{3N}) \otimes \mathfrak{h}^{\otimes_s^n}$. Consequently, one can endow $\mathfrak{H}$ with the following inner product :
\begin{equation} \label{eq : inner product on H}
\left \langle \Psi, \Phi \right \rangle_{\mathfrak{H}} = \sum_{n=0}^\infty \sum_{\lambda_1, \dots, \lambda_n = 1}^2 \int_{\mathbb{R}^{3N}} \int_{\mathbb{R}^{3n}} \overline{\Psi^{(n)}(\mathbf{x}, \mathbf{k}, \boldsymbol{\lambda})} \Phi^{(n)}(\mathbf{x}, \mathbf{k}, \boldsymbol{\lambda}) \textnormal{d} \mathbf{x} \textnormal{d} \mathbf{k} \\
\end{equation}
for all $\Psi, \Phi \in \mathfrak{H}$, where $\mathbf{x} = (x_1, \dots, x_N) \in \mathbb{R}^{3N}$, $\mathbf{k} = (k_1, \dots, k_n) \in \mathbb{R}^{3n}$ and $\boldsymbol{\lambda} = (\lambda_1, \dots, \lambda_n) \in \left\{1,2 \right\}^n$. \\

For any $f \in \mathfrak{h}$, we define the bosonic creation and annihilation operators on $\mathfrak{F}$ by
\begin{equation}
    \begin{split}
        a(f) & = \sum_{\lambda = 1,2} \int_{\mathbb{R}^3} \overline{f(k,\lambda)} a_{k,\lambda} \textnormal{d}k \\
        a^*(f) & = \sum_{\lambda = 1,2} \int_{\mathbb{R}^3} f(k,\lambda) a^*_{k,\lambda}\textnormal{d}k \ , \\
    \end{split}
\end{equation}
where the operator-valued distributions $a_{k,\lambda}$ and $a^*_{k,\lambda}$ satisfy the canonical commutation relations (CCR)
\begin{equation}
    \com{a_{k,\lambda}}{a^*_{\ell,\mu}} = \delta_{\lambda,\mu}\delta(k - \ell) \ , \ \com{a_{k,\lambda}}{a_{\ell,\mu}} = \com{a^*_{k,\lambda}}{a^*_{\ell,\mu}} = 0 \ . \\
\end{equation}
The number of photon operator is defined by 
\begin{equation}
    \mathcal{N} = \sum_{\lambda = 1,2} \int_{\mathbb{R}^{3}}a^{*}_{k,\lambda} a_{k,\lambda} \textnormal{d}k \ . \\
\end{equation}
The quantized transverse vector potential and electric field are defined by
\begin{equation} \label{eq : quantized magnetic field}
    \hat{\mathbf{A}}(x)   = \sum_{\lambda = 1,2} \int_{\mathbb{R}^3} \frac{(2\pi)^{-3/2}}{\sqrt{2|k|}} \boldsymbol{\epsilon}_\lambda(k)  \left(e^{ik\cdot x} a_{k,\lambda} + e^{-ik\cdot x} a^*_{k,\lambda} \right) \textnormal{d}k \ ,
\end{equation}
and
\begin{equation} \label{eq : quantized electric field}
    \hat{\mathbf{E}}(x) = i \sum_{\lambda= 1,2} \int_{\mathbb{R}^3} (2 \pi)^{-3/2}\sqrt{\frac{|k|}{2}} \boldsymbol{\epsilon}_\lambda(k) \left( e^{ik \cdot x} a_{k,\lambda} - e^{-ik \cdot x} a^*_{k,\lambda}\right) \textnormal{d}k \ . \\
\end{equation}
The Coulomb gauge is again implemented through the polarization vectors, namely
\begin{equation}
    \nabla \cdot \hat{\mathbf{A}} = 0 \ . \\
\end{equation}
Note that in the above, we denoted by $x$ the position operator on the one particle space $L^2(\mathbb{R}^3)$, namely the multiplication operator by $x$. We will keep this notation throughout the rest of the article. We also define the position operator on the $N$ particle space $L^2(\mathbb{R}^{3N})$ by $\mathbf{x} = (x_1, \dots, x_n)$, and the momentum operator on the same space by $-i\hbar \nabla$, where $\nabla = (\nabla_{1}, \dots, \nabla_{N})$, and for all $j = 1, \dots, N$, $\nabla_j = (\partial_{j1}, \partial_{j2}, \partial_{j3})$. With these notations, the spinless Pauli--Fierz Hamiltonian is defined by 
\begin{equation} \label{eq : PF Hamiltonian}
    \mathbb{H}_{\hbar} = \sum_{j=1}^N( -i \hbar \nabla_j - \hbar^{1/2} \kappa \ast\hat{\mathbf{A}}(x_j) )^2 + \sum_{1 \leq j < k \leq N} V(x_j - x_k) + \hbar \sum_{\lambda = 1,2} \int_{\mathbb{R}^3} |k| a^*_{k, \lambda}a_{k,\lambda} \textnormal{d}k \ , \\
\end{equation}
which can be seen as the formal quantization of the classical Hamiltonian \eqref{eq : energy NM}. 
Under suitable assumptions on $\kappa$,  $\mathbb{H}_{\hbar}$ can be shown  \cite{Hiroshima2002,Matte2017,Spohn_2004}
to be self-adjoint on the domain of the free (non-interacting) Hamiltonian
\begin{equation}
    \mathbb{H}^0_{\hbar} = -\sum_{j=1}^N\hbar^2\Delta_j +  \hbar \sum_{\lambda = 1,2} \int_{\mathbb{R}^3} |k| a^*_{k, \lambda}a_{k,\lambda} \textnormal{d}k \ .
\end{equation} 
For later purposes, we also define the \textit{Hamiltonian of the free photon field} :
\begin{equation}
    H_{\textnormal{f}} := \sum_{\lambda = 1,2} \int_{\mathbb{R}^3} |k| a^*_{k, \lambda}a_{k,\lambda} \textnormal{d}k \ . \\
\end{equation}
Note that, for ease of notation, the speed of light and the particle charge are set to one,  $c=e=1$,  and the mass of the particles to $m=1/2$. The proofs can, however, be easily adapted to treat particles with different masses and charges.
Our first result concerns the time evolution of pure states.  In this case, the system\footnote{Note that we will consider the limit $\hbar \rightarrow 0$ and and therefore consider  a sequence of quantum systems.  For ease of notation, we omit the dependence of the wave functions on $\hbar $ throughout this article.} is described by  \textit{wave functions}, i.e elements of the Hilbert space $\mathfrak{H}$,  and the time evolution is governed by the Schr\"odinger equation 
\begin{equation} \label{eq : sch eq}
    \begin{cases}
        i \hbar \partial_t \Psi_t = \mathbb{H}_{\hbar}\Psi_t \\
     \left.\Psi_t \right|_{t=0} = \Psi \ . \\
    \end{cases}
\end{equation}
By Stone's theorem, the solution is given by
\begin{equation}
    \Psi_t = e^{-i \frac{t}{\hbar}\mathbb{H}_{\hbar}} \Psi \ . \\
\end{equation}
We will provide a second result for a specific class of mixed states, which are described by \textit{density matrices} on $\mathfrak{H}$,  i.e.,  positive trace-class operators:
\begin{equation}
    \mathcal{S}(\mathfrak{H}) = \left\{ \rho \textnormal{ self-adjoint operator on } \mathfrak{H} \ \| \ \rho \geq 0 \ , \textnormal{Tr}_{\mathfrak{H}} \left[ \rho \right] = 1 \right\} \ . \\
\end{equation}
Their time evolution is given by the von-Neumann equation 
\begin{equation}
\label{eq : von neumann equation}
    \begin{cases}
        i \hbar \partial_t \rho(t) = \com{\mathbb{H}_{\hbar}}{\rho(t)} \\
        \left. \rho(t) \right |_{t=0} = \rho \ , \\
    \end{cases}
\end{equation}
with solution 
\begin{equation} \label{eq : von neumann evolution}
    \rho(t) =  e^{-i \frac{t}{\hbar}\mathbb{H}_{\hbar}} \rho  e^{i \frac{t}{\hbar}\mathbb{H}_{\hbar}} \ . \\
\end{equation}

\section{Main results}

The goal of this article is to study the time evolution generated by the Pauli--Fierz Hamiltonian in the classical limit $\hbar \rightarrow 0$.  The reduced Planck constant $\hbar$ has a fixed value, and from a physical point of view,  this limit should be understood as an idealization of situations in which $\hbar$ is effectively small compared to the other parameters of the system.  
By dividing the Schr\"odinger equation~\eqref{eq : sch eq} by $\hbar$, one easily sees that this limit is equivalent to the choice where the speed of light and the reduced Planck constant are fixed, the mass of the particles scales inversely with the scaling parameter (which,  in our notation -- rather confusingly -- is still called $\hbar$),  and the particle charges scale inversely with the square root of the scaling parameter.   
In this sense,  the classical regime can be viewed as a system of heavy particles coupled to a strong electromagnetic field, as was done in \cite{knowles2009limiting}. This interpretation motivates our choice of initial data and the method of the proof.  More concretely,  we think of the system as being composed of heavy tracer particles that are immersed in a coherent state of many photons.  If the photons remain in a coherent states during the time evolution it is expected that the quantized electromagentic field can approximately be described by a classical field satisfying Maxwell's equations.
If the tracer particles are initially sharply localized in position and momentum (as is the case for Gaussian one-particle states), their large mass ensures that their variances in position and velocity remain of the same order of magnitude under the free Schr\"odinger evolution.  Our goal is to show that this still holds under the Pauli--Fierz dynamics and that the fluctuations around the coherent state are of subleading order.  The Ehrenfest equations of the particles' positions and momenta as well as the annihilation operator then suggest that the mean values of these quantities are described by the Newton--Maxwell equations. That this actually holds is the content of Theorem~\eqref{thm : main result}.  \\

In order to formulate our main result, let us first give some more definitions and assumptions. A coherent state of photons can be defined as the (bosonic) Fock space vector given by the action of the Weyl operator on the \textit{vacuum} $\Omega$ :
\begin{equation} \label{eq : coherent state}
    W(\alpha) \Omega = e^{-\norm{\alpha}^2/2}\bigoplus_{n=0}^\infty \frac{\alpha^{\otimes n}}{ \sqrt{n !}} \ , \\
\end{equation}
where, for any $\alpha \in \mathfrak{h}$, the unitary Weyl operator is defined as
\begin{equation} \label{eq : def weyl}
W(\alpha) = \exp \left( a(\alpha) - a^*(\alpha) \right) \ . \\
\end{equation}
We refer to \cite[Lemma 2.2]{rodnianski2007quantumfluctuationsrateconvergence} for some basic properties of the Weyl operator. Let us note one useful property that we will use very often in the article, that is the \textit{shifting property}
\begin{equation} \label{eq : shifting pty weyl}
	\begin{split}
		W^*(f) a_{k,\lambda} W(f) & = a_{k,\lambda} + f(k,\lambda)  \ . \\
	\end{split}
\end{equation}

We now specify some assumptions that we shall make throughout the rest of the article. First,  we make the following assumptions on the cutoff function $\kappa$.
\begin{assumption}[Assumptions on the cutoff] \label{assumptions : kappa}
    The cutoff $\kappa : \mathbb{R}^3 \rightarrow \mathbb{R}$ is assumed to be real-valued, even, and such that
    \begin{equation}
        \begin{split}
            |\cdot|^{-1} \mathcal{F}[\kappa] & \in L^2(\mathbb{R}^3) \ ,  \\
            |\cdot|^{\frac{3}{2} - \sigma} \mathcal{F}[\kappa] & \in L^2(\mathbb{R}^3) \\
        \end{split}
    \end{equation}
    for some $\sigma \in \left[1/2,1\right]$. \\
\end{assumption}
We then require the following conditions on the initial many-body state.
\begin{assumption}[Initial data] \label{assumptions : init 1}
Let $q_0 = (q_{0,1}, \dots, q_{0,N}) \in \mathbb{R}^{3N}$, $p_0 = (p_{0,1}, \dots, p_{0,N}) \in \mathbb{R}^{3N}$, $\alpha_0 \in \mathfrak{h}$, and  $C > 0$.
Moreover,  let $\Omega$ be the \textit{vacuum} of the Fock space and let $W$ be the Weyl operator defined in \eqref{eq : def weyl}. For each fixed $\hbar > 0$, we consider $\Psi_0 \in  \mathcal{D}\left(\mathbb{H}_{\hbar}^{1/2}\right) \cap \mathcal{D}\left(\mathcal{N}^{1/2}\right) \cap \mathcal{D}(\mathbf{x}) $ with $\norm{\Psi_0} = 1$ such that
\begin{equation} \label{eq : initial assumptions}
\begin{split}
         \norm{\left(\mathbf{x} - q_0 \right) \otimes \mathbb{1}_{\mathfrak{F}} \Psi_0}_{\mathfrak{H}}  &\leq C \hbar^{1/2}, 
        \\
         \norm{\left(-i \hbar \nabla  - p_{0} \right)\otimes \mathbb{1}_{\mathfrak{F}}  \Psi_0}_{\mathfrak{H}} &\leq C \hbar^{1/2},
        \\
          \ps{\Psi_0}{\mathbb{1}_{L^2(\mathbb{R}^{3N})}\otimes W(\hbar^{-1/2}\alpha_0) \mathcal{N}W^*(\hbar^{-1/2}\alpha_0)\Psi_0}_{\mathfrak{H}}  &\leq C . \\
\end{split}
\end{equation}
\end{assumption}

Let us point out some important facts regarding Assumption~\ref{assumptions : init 1}. \\

\begin{remark} \label{remark : rks on the initial assumptions}
\begin{itemize}
\item Assumption~\ref{assumptions : init 1} is satisfied if each particle is in a Gaussian state with variance $\hbar$ and the photon field is in a coherent state with mean photon number $\hbar^{-1}\norm{\alpha_0}^2_\mathfrak{h}$. More explicitly,  this state can be written as 
\begin{equation}
\label{eq:definition of Gaussian plus coherent state}
    \Psi^\hbar_u = \psi^\hbar_{q,p} \otimes W(\hbar^{-1/2} \alpha) \Omega \ , \quad u = (q,p, \alpha) \in \mathbb{R}^{3N} \times \mathbb{R}^{3N} \times \mathfrak{h} \ , \\
\end{equation}
where $\psi^\hbar_{q,p}$ is the Gaussian wave function with position $q$, momentum $p$ and variance $\hbar$, defined by :
\begin{equation} \label{eq : gaussian wave packet 3d}
\psi^\hbar_{q,p}(x) = (\pi \hbar)^{-3/4} \exp \left( - \frac{|x - q|^2}{2 \hbar} \right) \exp \left( \frac{i}{\hbar}p \cdot \left(x - q \right)\right) \ , \ \forall x \in \mathbb{R}^{3N} \ . \\
\end{equation}
It is a standard computation to show that the state $\Psi^\hbar_u$ indeed satisfies \eqref{eq : initial assumptions}. In particular,  we have $\ps{\Psi^\hbar_u}{W(\hbar^{-1/2}\alpha_0) \mathcal{N}W^*(\hbar^{-1/2}\alpha_0)\Psi^\hbar_u} = 0$. \\
\item As explained in Lemma~\ref{lemma : invariance of the domains for beta} below, the domain assumptions on $\Psi_0$ ensure that the LHS quantities in \eqref{eq : initial assumptions} are well defined during the time evolution of the Pauli--Fierz Hamiltonian, i.e when $\Psi_0$ is replaced by $\Psi_t $ . \\
\end{itemize}
\end{remark}
In this article, we rely on a well-posedness result for the Newton--Maxwell system, proved in \cite{ammari2022derivationclassicalelectrodynamicscharges}, which ensures the existence of a unique strong solution to \eqref{eq : newton motion eqs PF} : 
\begin{proposition}\cite[Theorem 1.3]{ammari2022derivationclassicalelectrodynamicscharges} \label{Theorem : well posedness AFH}
    Let $\sigma \in [1/2,1]$, and assume that Assumption \eqref{assumptions : kappa} is satisfied. Then, there exists a unique global strong solution 
\begin{equation}
    u(.) \in \mathcal{C}(\mathbb{R}; X^\sigma) \cap \mathcal{C}^1(\mathbb{R}; X^{\sigma-1})
\end{equation}
to the Newton--Maxwell system \eqref{eq : newton motion eqs PF}. \\ 
\end{proposition}
As a direct consequence, we obtain the conservation of the energy \eqref{eq : energy NM}.
\begin{lemma}[Conservation of the energy] \label{lemma : conservation of the energy}
    Let $\sigma \in [1/2,1]$. Let $u_0 = (q_0, p_0, \alpha_0) \in X^\sigma$ and let $u(t) = (q(t), p(t), \alpha_t)$ be the solution to \eqref{eq : newton motion eqs PF} with initial data $u_0$. Then, the corresponding energy functional $\mathcal{H}$ defined in \eqref{eq : energy NM} is conserved along the flow : for all $t \geq 0$, there holds
    \begin{equation} \label{eq : Conservation of the energy}
        \mathcal{H}(u(t)) = \mathcal{H}(u_0) \ . \\
    \end{equation}
\end{lemma}

\begin{proof}
The statement is proven by direct a calculation using Duhamel formula. \\
\end{proof}

Our main result shows that the inequalities \eqref{eq : initial assumptions} remain true during the time evolution,  when $u_0 = (q_0, p_0, \alpha_0)$ is replaced by $u(t) = (q(t), p(t), \alpha(t))$, the unique solution of the Newton--Maxwell equations \eqref{eq : newton motion eqs PF} with initial data $u_0$. 
\begin{theorem} \label{thm : main result}
    Let $\sigma \in \left[ 1/2, 1 \right]$. Let $u_0 = (q_0, p_0, \alpha_0) \in X^\sigma$, and let $\Psi_0 \in \mathfrak{H}$ satisfy Assumption~\ref{assumptions : init 1}.  Let $u(.)$ be the unique solution to the Newton--Maxwell system \eqref{eq : newton motion eqs PF} with the initial data $u_0$, and let $\Psi_t = e^{-i\frac{t}{\hbar}\mathbb{H}_{\hbar}} \Psi_0$ be the solution to the Schrödinger equation \eqref{eq : sch eq}. Then, there exist constants $C,c > 0$ such that for all $t \geq 0$ 
    \begin{equation} 
    \label{eq : cvgce in variance thm}
        \begin{split}
            \norm{\left(\mathbf{x} - q(t) \right) \otimes \mathbb{1}_{\mathfrak{F}} \Psi_t}_{\mathfrak{H}} & \leq Ce^{ct^2} \hbar^{1/2}, \\
            \norm{\left(-i\hbar \nabla - p(t) \right) \otimes \mathbb{1}_{\mathfrak{F}} \Psi_t}_{\mathfrak{H}} & \leq Ce^{ct^2} \hbar^{1/2}, \\
             \ps{\Psi_t}{\mathbb{1}_{L^2(\mathbb{R}^{3N})}\otimes W(\hbar^{-1/2}\alpha_t)\mathcal{N}W^*(\hbar^{-1/2}\alpha_t)\Psi_t}_{\mathfrak{H}} & \leq C e^{ct^2}. \\
        \end{split}
    \end{equation}
\end{theorem} 

\begin{remark}
Note that a similar result can be proven for the classical limit of the Nelson model with ultraviolet regularization, as considered in \cite{10.1063/5.0239967}.  For reasons of  presentation, we do not include the result and its proof in this article. In Remark~\ref{remark:Proof for Nelson model}, we explain how the proof must be adapted. \\
\end{remark}

From Theorem~\ref{thm : main result}, we can infer approximations for the expectation values of observables that depend on the particles' positions, momenta, or the field operator. To be more precise, we will consider the quantities $f(\mathcal{J}), f(J(t))$ of quantum and classical observables, where
\begin{equation} \label{eq : quantum and classical observables and f}
    \begin{split}
    		& \left( \mathcal{J}, J(t) \right) \in \left\{ \left(\mathbf{x}, q(t) \right), \left(-i\hbar\nabla, p(t) \right)\right\} \quad \textnormal{and } \  f \in W^{1,\infty}(\mathbb{R}^{3N}) \\
    		\textnormal{or } & \left( \mathcal{J}, J(t) \right) = \left(\hbar^{1/2}\Phi(g) , \phi_t(g)\right) \quad \textnormal{and } \  f \in W^{1,\infty}(\mathbb{R}^{3}) \ . \\
    \end{split}
\end{equation} 
Here, the last couple of observables is defined by 
\begin{equation}
\begin{split}
\Phi(g) & = a^*(g) + a(g) \\
\phi_t(g)  & = \ps{\alpha_t}{g} + \ps{g}{\alpha_t} \ , \\
\end{split}
\end{equation}
for some ($\mathbb{R}^3$-valued) $g \in \mathfrak{h}$. For example, if one takes $g = \mathbf{G}_x$ as defined in \eqref{eq : def G}, then $\phi_t(g) = \kappa \ast \mathbf{A}_{\alpha_t}(x)$ and $\Phi(g) = \kappa \ast \hat{\mathbf{A}}(x)$, the classical and quantized vector potentials. \\

In addition,  we can approximate the one-photon reduced density matrix $\gamma_{\Psi}$,  which is defined through its integral kernel as
\begin{equation} \label{eq : def 1PRDM wave fct}
    \gamma_{\Psi}(k,\lambda ; \ell, \nu) = \hbar \ps{\Psi}{\mathbb{1}_{L^2(\mathbb{R}^{3N})} \otimes a^*_{\ell, \nu}a_{k,\lambda}\Psi}_{\mathfrak{H}} .
\end{equation}
Note the normalization convention $\textnormal{Tr}_{\mathfrak{h}} (\gamma_{\Psi}) = \hbar \ps{\Psi}{ \mathcal{N} \Psi}_{\mathfrak{H}}$. \\

\begin{corollary} \label{Corollary : corollary of main theorem for pure states}
Let Theorem \ref{thm : main result} hold.  Let $(\mathcal{J}, J(t))$ and $f$ be defined as in \eqref{eq : quantum and classical observables and f}. Then, there exist constants $C, c > 0$ such that 
\begin{equation} \label{eq : expectation of observables first result}
    \left \lvert \ps{\Psi_t}{f(\mathcal{J}) \Psi_t} - f(J(t)) \right \lvert \leq Ce^{ct^2} \hbar^{1/2} \\
\end{equation}
and
\begin{equation} \label{eq : cvge density matrices}
\norm{ \gamma_{\Psi_t} - \ket{\alpha_t} \bra{\alpha_t}}_{\mathfrak{S}^1(\mathfrak{h})} \leq Ce^{ct^2} \min \left\{ \hbar^{1/2}, \hbar \right\} \ , \\
\end{equation}
where $\norm{\cdot}_{\mathfrak{S}^1(\mathfrak{h})} = \textnormal{Tr}_\mathfrak{h} |\cdot|$. \\
\end{corollary}

Using the linearity of the von-Neumann equation~\eqref{eq : von neumann equation} we can moreover approximate the time evolution of density matrices of the form 
\begin{equation} \label{eq : anti-wick quantization def}
    \rho_\hbar := \int_{X^\sigma} \ket{\Psi^\hbar_u} \bra{\Psi^\hbar_u} \textnormal{d}\mu(u) .
\end{equation}
Here,  $X^\sigma$ is the Hilbert space defined in \eqref{eq : def X et Xdot sigma},  $\Psi^\hbar_u$ is defined as in \eqref{eq:definition of Gaussian plus coherent state},  and $\mu$ is a Borel probability measure on $X^\sigma$. We denote by $\mathfrak{B}(X^\sigma)$ the set of such measures. 
Note that \eqref{eq : anti-wick quantization def} is well defined as a Bochner integral, with the integrand taking values in the Banach space $\mathfrak{S}^1(\mathfrak{H})$ of trace-class operators on $\mathfrak{H}$.  The one-photon reduced density matrix,  denoted $\Gamma_\rho$,  is defined in this setting by its integral kernel as
\begin{equation} \label{eq : def 1PRDM density matrix}
    \Gamma_\rho(k,\lambda ; \ell, \nu) = \hbar \tr{\mathbb{1}_{L^2(\mathbb{R}^{3N})} \otimes a^*_{\ell, \nu}a_{k,\lambda} \rho} \ . \\
\end{equation}
Moreover, in view of \eqref{eq : von neumann evolution},  the time evolution of $\rho_\hbar$ can be written as
\begin{equation} \label{eq : time evolved density matrix}
    \rho_\hbar(t) = \int_{X^\sigma} \ket{\Psi^\hbar_u(t)} \bra{\Psi^\hbar_u(t)} \textnormal{d}\mu(u) ,
\quad \text{where} \quad 
    \Psi^\hbar_u(t) = e^{-i \frac{t}{\hbar}\mathbb{H}_{\hbar}} \Psi^\hbar_u \ . \\
\end{equation}
The following result can be seen as a generalization of Corollary \ref{Corollary : corollary of main theorem for pure states} to the class of states defined in \eqref{eq : anti-wick quantization def}. \\

\begin{corollary} \label{corollary : more general states}
    Let $\sigma \in \left[1/2, 1 \right]$ and $u \in X^\sigma$.  Let $\mu \in \mathfrak{B}(X^\sigma)$ be such that
    \begin{equation} \label{eq : assumption on mu}
        \int_{X^\sigma} \norm{u}^2_{X^\sigma} \textnormal{d}\mu(u) < + \infty \ . \\
    \end{equation}
Consider the family $(\rho_\hbar)_{\hbar \in (0,1)}$  of density matrices defined in \eqref{eq : anti-wick quantization def},  given by
    \begin{equation}
        \rho_\hbar = \int_{X^\sigma} \ket{\Psi^\hbar_u} \bra{\Psi^\hbar_u} \textnormal{d}\mu(u) \ .
    \end{equation}
    Let $\rho_\hbar(t)$ be the solution of the von-Neumann equation \eqref{eq : von neumann evolution} with initial data $\rho_\hbar$, and let $(\mathcal{J}, J(t))$ and $f$ be defined as in \eqref{eq : quantum and classical observables and f} where, $\mu$-almost surely, $u(t) = (q(t), p(t), \alpha_t)$ is the solution to \eqref{eq : newton motion eqs PF}  with the initial data $u \in X^\sigma$. Then, there exist constants $C,c > 0$ such that for all $t \geq 0$
    \begin{equation} \label{eq : more general states convergence}
        \left \lvert \tr{f(\mathcal{J}) \rho_\hbar(t)} - \int_{X^\sigma} f(J(t)) \textnormal{d} \mu(u) \right \lvert \leq Ce^{ct^2} \hbar^{1/2} \\
    \end{equation}
    and
    \begin{equation} \label{eq : more general states reduced density matrices convergence}
        \norm{ \Gamma_{\rho_{\hbar}(t)} - \int_{X^\sigma} \ket{\alpha_t} \bra{\alpha_t} \textnormal{d}\mu(u)}_{\mathfrak{S}^1(\mathfrak{h})} \leq Ce^{ct^2} \min \left\{ \hbar^{1/2}, \hbar \right\} \ . \\
    \end{equation} \\
\end{corollary}

\paragraph{\textbf{Comparison with the literature}}

The present work belongs to the area of rigorous derivations of effective evolution equations,  which originated in the 1970s,  amongst others, with important contributions \cite{BH1977,ginibreveloclassicallimit,GV1979b,H1974,Lanford1975,S1981}
 by Braun,  Ginibre, Hepp, Lanford, Spohn and Velo and has since become and active field of research.  The first rigorous derivation of the Newton--Maxwell equations from the dynamics of the Pauli--Fierz Hamiltonian in the classical limit has been established in  \cite{knowles2009limiting}.  Based on the approach of \cite{H1974},  it proves that the Heisenberg evolution of Weyl operators converges strongly in the classical limit to the phase-space flow generated  by the Newton--Maxwell equations.  In \cite{ammari2022derivationclassicalelectrodynamicscharges}, a derivation by means of the Wigner measure approach is given,  extending the previous result to a larger class of states with fewer restrictions on the cutoff function $\kappa$.  Using the correspondence between the quantum and classical dynamics the authors also prove new global well-posedness results for the Newton--Maxwell equations in  energy-space under weak assumptions on the cutoff  function.
Within this work we make the same assumptions on the cutoff as in \cite{ammari2022derivationclassicalelectrodynamicscharges} and provide an alternative derivation for a special class of initial states,  as specified in Assumption~\ref{assumptions : init 1} and \eqref{eq : anti-wick quantization def}.  The novelty of our work is that we provide an explicit rate of convergence by explicitly estimating the fluctuations around the Newton--Maxwell dynamics.   We think of the classical regime as a system of heavy particles which couple to a strong magnetic field (as in  \cite{knowles2009limiting}) and prove our result by an adaption of the techniques from \cite{CARDENAS2022109403,Deckert2014},  which are concerned with tracer particles in  Bose--Einstein condensates.
In the sense that we study the growth of excitations generated by a fluctuation dynamics defined via Weyl operators,  our proof is closely related to \cite{rodnianski2007quantumfluctuationsrateconvergence}.
The main conceptual idea underlying our proof is to consider the fluctuations around the kinetic momentum,  which avoids the appearance of higher moments arising from the minimal coupling term in the Pauli--Fierz Hamiltonian (see Subsection~\ref{subsec : Strategy of the proof} for further details).
Let us also mention the works \cite{10.1063/5.0239967,G2025} which study the classical limits of the Nelson model with ultraviolet cutoff and the Fr\"ohlich model with ultraviolet cutoff by means of the Wigner measure approach.  In \cite{BCFFS2013} it is shown that the dynamics of an electron in a slowly varying external potential and interacting with the quantized electromagnetic field can be approximated by the dynamics of an electron with a renormalized dispersion law in the external potential.
 In addition,  effective dynamics for the Pauli--Fierz dynamics have been rigorously derived in the partially classically limit \cite{CFO2019,CFO2022},  the bosonic mean-field limit \cite{falconi2023derivation,LP2020},  and the fermionic mean-field limit \cite{leopold2024derivationmaxwellschrodingervlasovmaxwellequations, LS2026}.  For further references on the derivation of effective equations for other non-relativistic quantum field models in these limits,  we refer to \cite[Section~II.2]{leopold2024derivationmaxwellschrodingervlasovmaxwellequations}.  \\

\section{Proof of the main results}

\paragraph{\textbf{Notation}}

Throughout the article, we will denote by $C$ or $c$ a generic positive constant which may differ from one line to another, and can depend on fixed parameters of our setting, for example $N, \sigma$ or finite norms of the cutoff $\kappa$, but not $t$ or $\hbar$.  We will use the notation $A \lesssim B$ when $A \leq C B$ for some positive constant $C$ which can depend on the parameters as specified above. We will also use the notation $\lesssim_S$ to indicate the dependence of the implicit constant on an arbitrary set of parameters $S$. When the meaning is obvious, we will omit the notation for the tensor product with the identity, as well as the subscripts denoting the spaces on which a scalar product or norm is taken. For simplicity, the $N$ particles Laplacian (i.e acting on $L^2(\mathbb{R}^{3N})$) will be denoted as $- \Delta = - \sum_{j=1}^N \Delta_j$, where $-\Delta_j$ is the one particle Laplacian (i.e acting on $L^2(\mathbb{R}^{3})$). For a separable Hilbert space $\mathfrak{K}$, we denote by $\mathfrak{S}^1(\mathfrak{K})$ the space of trace-class operators on $\mathfrak{K}$ and by $\mathfrak{S}^2(\mathfrak{K})$ the space of Hilbert-Schmidt operators on $\mathfrak{K}$, with norm defined by
\begin{equation}
\norm{\mathcal{O}}^p_{\mathfrak{S}^p(\mathfrak{K})} = \textnormal{Tr} \left \lvert \mathcal{O} \right \lvert^p \quad , \ p = 1,2 \ . \\
\end{equation}
We moreover define the following shorthand notations to denote the regularized classical and quantized fields, by dropping the $\kappa$ dependence :
\begin{equation}
    \begin{split}
        & \boldsymbol{\hat{\mathbb{A}}} := \kappa \ast \hat{\mathbf{A}} \ , \quad \boldsymbol{\mathbb{A}}_\alpha := \kappa \ast \mathbf{A}_\alpha \\
        & \boldsymbol{\hat{\mathbb{E}}} := \kappa \ast \hat{\mathbf{E}} \ , \quad \boldsymbol{\mathbb{E}}_\alpha := \kappa \ast \mathbf{E}_\alpha \ . \\
    \end{split}
\end{equation} \\

\subsection{Strategy of the proof} \label{subsec : Strategy of the proof}

As mentioned earlier, we view the quantum system as consisting of heavy tracer particles interacting with a coherent state of photons. Since a coherent state of photons exhibits very few correlations, it can be regarded as behaving similarly to a condensate of particles. This interpretation suggests applying the techniques from \cite{ CARDENAS2022109403,Deckert2014}, which were used to show that a system of tracer particles immersed in a Bose–Einstein condensate can effectively be described by a coupled system consisting of Newton’s equations for the tracer particles and a PDE for the condensate wave function. 
In short,  the strategy for proving the main result is to analyze the time evolution of the three quantities in \eqref{eq : cvgce in variance thm} and then to conclude the statement of Theorem~\ref{thm : main result} using a Gr\"onwall estimate.

In pursuing this strategy, difficulties arise due to the fact that the time derivative of the variance of the momentum cannot be controlled by the quantities in \eqref{eq : cvgce in variance thm} in a way that is suitable for a Grönwall argument.  The situation is similar to the one on the classical level if one would like to estimate the absolute value of $p(t)$ by $|p(t)|$ and $\| \alpha_t \|_{\mathfrak{h}}$ with the help of the Duhamel expansion \footnote{Note that for the sake of exposition we place ourselves in the setting where $N=1$ and $V \equiv 0$ here.}
\begin{equation}
p(t) = p(0) + 2 \sum_{\ell=1}^3 \int_0^t  \left( p^\ell(s) - \boldsymbol{\mathbb{A}}^\ell_{\alpha_s}(q(s)) \right) \nabla \left(   \boldsymbol{\mathbb{A}}^\ell_{\alpha_s}\right) (q(s))  \,  \textnormal{d}s
\end{equation}
as the second term on the right hand side gives rise to $\| \alpha_s \|_{\mathfrak{h}}^2$. Now, observe that in order to control $|p(t)|$, it is actually much better to consider the \textit{kinetic momentum}
\begin{equation} \label{eq : physical momentum}
     \Tilde{p}(t) := \frac{1}{2} \partial_t q(t) = p(t) - \boldsymbol{\mathbb{A}}_{\alpha_t}(q(t)) \ , \\
\end{equation}
whose time evolution is given by the \textit{Lorentz force}
\begin{equation} \label{eq : classical evolution of the shifted velocity}
    \partial_t \Tilde{p}(t) =  \boldsymbol{\mathbb{E}}_{\alpha_t}(q(t)) + 2 \Tilde{p}(t) \times \left( \nabla \times \boldsymbol{\mathbb{A}}_{\alpha_t}(q(t)) \right) \ . \\
\end{equation}
Since the kinetic momentum $\Tilde{p}$ is perpendicular to the magnetic part $\nabla \times \boldsymbol{\mathbb{A}}$ of the Lorentz force,  we get
\begin{equation} \label{eq : classical cancellation}
\partial_t |\Tilde{p}(t)|^2 
= 2  \Tilde{p}(t) \cdot \boldsymbol{\mathbb{E}}_{\alpha_t}(q(t))
  \ . \\
\end{equation}
Consequently,  the Duhamel expansion of $|\Tilde{p}(t)|^2$ will lead to reduced powers of $|\Tilde{p}|$ and $\norm{\alpha}_{\mathfrak{h}}$ on the right-hand side, by using the fact that the part of the Lorentz force arising from the vector potential changes only the direction, but not the magnitude, of the particle's momentum.  Since $\Tilde{p}$ and $p$ can be related to each other by means of \eqref{eq : p tilda controled by p}, this gives rise to an estimate for $p$ (see the proof of Lemma~\ref{lemma : A priori estimates for the Newton--Maxwell solutions} for further details).
Note additionally that one can express \eqref{eq : classical evolution of the shifted velocity} by means of the \textit{Faraday} tensor
\begin{equation} \label{eq : classical faraday tensor}
    F^{\ell m}_\alpha = \partial_m \boldsymbol{\mathbb{A}}^\ell_{\alpha} - \partial_\ell \boldsymbol{\mathbb{A}}^m_{\alpha} \ , \\
\end{equation}
since  
\begin{equation} \label{eq : faraday tensor classical identity magnetic}
    \left( \Tilde{p} \times \left( \nabla \times \boldsymbol{\mathbb{A}}_{\alpha} \right) \right) ^m = \sum_{\ell=1}^3 \Tilde{p}^\ell F^{\ell m}_\alpha \ . \\
\end{equation}
Using this rewriting, one can see more easily how \eqref{eq : classical evolution of the shifted velocity} relates to QED. Indeed, the natural quantum analogue  of \eqref{eq : physical momentum} being 
\begin{equation} \label{eq : physical momentum quantum}
    P := -i \hbar \nabla - \hbar^{1/2} \boldsymbol{\hat{\mathbb{A}}}(x) \ , \\
\end{equation}
one formally computes its time evolution in the Heisenberg picture as
\begin{equation} \label{eq : heisenberg evolution P}
\dt e^{i \frac{t}{\hbar}\mathbb{H}_\hbar} P^m e^{-i \frac{t}{\hbar}\mathbb{H}_\hbar} = e^{i \frac{t}{\hbar}\mathbb{H}_\hbar} \left( \hbar^{1/2}\boldsymbol{\hat{\mathbb{E}}}^m(x) + \hbar^{1/2} \sum_{\ell=1}^3 \left\{P^\ell, \hat{F}^{\ell m} \right\} \right)  e^{-i \frac{t}{\hbar}\mathbb{H}_\hbar}  \ , \\
\end{equation}
where $\left\{ , \right\}$ stands for the anticommutator, and 
\begin{equation} \label{eq : quantum faraday tensor}
    \hat{F}^{\ell m} = \partial_m \boldsymbol{\hat{\mathbb{A}}}^\ell - \partial_\ell \boldsymbol{\hat{\mathbb{A}}}^m \\
\end{equation}
is the quantum analogue of \eqref{eq : classical faraday tensor}.  Using \eqref{eq : faraday tensor classical identity magnetic}, one sees the quantum - classical correspondence between \eqref{eq : heisenberg evolution P} and \eqref{eq : classical evolution of the shifted velocity}, where the classical field $\boldsymbol{\mathbb{E}}_{\alpha_t}$ is replaced by the quantum field $\hbar^{1/2}\boldsymbol{\hat{\mathbb{E}}}$, and the classical Lorentz force $\Tilde{p} \times \left( \nabla \times \boldsymbol{\mathbb{A}}_{\alpha} \right)$ is replaced by its quantum counterpart, with $m^{\textnormal{th}}$ component given by $\hbar^{1/2} \sum_\ell \left\{P^\ell, \hat{F}^{\ell m} \right\}$. Moreover, one obtains similarly
\begin{equation} \label{eq : heisenberg evolution P^2}
\dt e^{i \frac{t}{\hbar}\mathbb{H}_\hbar} \lvert P \lvert^2 e^{-i \frac{t}{\hbar}\mathbb{H}_\hbar} = e^{i \frac{t}{\hbar}\mathbb{H}_\hbar} \left( \hbar^{1/2} \left\{P,\boldsymbol{\hat{\mathbb{E}}}(x) \right\} - \hbar^{1/2} \sum_{\ell=1}^3 \sum_{m=1}^3 \left\{P^\ell, \left\{P^m,\hat{F}^{\ell m}(x) \right\} \right\} \right) e^{-i \frac{t}{\hbar}\mathbb{H}_\hbar} \ . \\
\end{equation}
There, \eqref{eq : classical cancellation} and the correspondance established above would suggest that the remaining anticommutators should at least be small in $\hbar$. In fact, manipulations using the antisymmetry of the Faraday tensor allow to express this remainder in terms of suitable commutators which, as we will see in the proof of Lemma \ref{lemma : change of beta b in time} below, prove that this intuition happens to be true at the rigorous level. \\

The previous discussion leads us to consider the variance in \textit{kinetic momentum} rather than the variance in momentum originally considered in \eqref{eq : initial assumptions}.  We define it below with the other quantities of interest in a functional which will be suitable for a Grönwall argument. We will use the classical and quantum kinetic momentums in the $N$-particles setting, which we define respectively by $\Tilde{p} = (\Tilde{p}_1, \dots, \Tilde{p}_N)$ and $P = (P_1, \dots, P_N)$, where :
\begin{equation} \label{eq : def kinetic momentums N particles}
\begin{split}
 \Tilde{p}_j & = p_j - \boldsymbol{\mathbb{A}}_{\alpha}(q_j) \\
 P_j & = -i\hbar \nabla_j - \hbar^{1/2} \boldsymbol{\hat{\mathbb{A}}}(x_j) \ , \\
\end{split}
\end{equation}
for all $j = 1, \dots, N$. \\

\begin{definition}[Definition of the functional] \label{def : Definition of the functional}
Let $\sigma \in [1/2,1]$.  For $u = (q,p,\alpha) \in X^\sigma$ and 
$\Psi \in \mathcal{D} \left( \left( - \Delta \right)^{1/2} \right) \cap \mathcal{D}(\mathbf{x}) \cap \mathcal{D}\left(\mathcal{N}^{1/2}\right)$ we define  
    \begin{equation} \label{eq : def of fctnals}
    \begin{split}
        \boldsymbol{\beta}^a \left(\Psi, q \right) & = \norm{ \left(\mathbf{x} - q \right) \Psi}^2 , \\
        \boldsymbol{\beta}^b \left(\Psi, u\right) & = \sum_{j=1}^N \norm{ \left(P_j - \Tilde{p}_j \right)\Psi}^2 , \\
        \boldsymbol{\beta}^c \left(\Psi, \alpha \right) & = \hbar \norm{\mathcal{N}^{1/2}W^*(\hbar^{-1/2}\alpha)\Psi}^2 .
    \end{split}
\end{equation}
Moreover,  we define $\boldsymbol{\beta} : \mathcal{D} \left( \left( - \Delta \right)^{1/2} \right) \cap \mathcal{D}(\mathbf{x}) \cap \mathcal{D}\left(\mathcal{N}^{1/2}\right) \times X^\sigma\longrightarrow \mathbb{R}^+$ by
\begin{equation}
    \boldsymbol{\beta}(\Psi,u) \coloneqq  \boldsymbol{\beta}^a \left(\Psi, q \right) + \boldsymbol{\beta}^b \left(\Psi, u\right) + \boldsymbol{\beta}^c \left(\Psi, \alpha \right) \ .
\end{equation} \\
\end{definition}
Note the following remarks about the definition above.
\begin{remark}

    \begin{itemize}
        \item Inequality \eqref{eq : operator inequality that makes beta b wd} ensures that $\boldsymbol{\beta}^b$ is indeed well defined on the given domain. \\
        \item 
The functional $\boldsymbol{\beta}$ is still well defined during the time evolution, i.e when one replaces $\Psi$ and $u$ by their time evolved counterparts under the Pauli--Fierz and Newton--Maxwell dynamics, respectively. This holds on the one hand thanks to Proposition~\ref{Theorem : well posedness AFH} which ensures that $u(t) \in X^\sigma$ for all $t \geq 0$, and on the other hand thanks to Lemma \ref{lemma : invariance of the domains for beta}, which shows that the domain of the many-body initial state is invariant during the time evolution. \\
        \item Let us point out the following identity, which follows from the shifting property of the Weyl operator, and which will be useful for later purposes :
        \begin{equation} \label{eq : rewriting for beta c}
    \boldsymbol{\beta}^c \left(\Psi, \alpha \right) = \sum_{\lambda = 1,2} \int_{\mathbb{R}^3} \norm{\left(\hbar^{1/2}a_{k,\lambda} - \alpha(k,\lambda) \right) \Psi}^2 \textnormal{d}k \ . \\
\end{equation}
    \end{itemize}
\end{remark}

\begin{remark} \label{remark:Proof for Nelson model}
Note that a similar strategy can be applied to derive Newton equations coupled to a scalar field from the dynamics generated by the Nelson model with ultraviolet cutoff in the classical limit, as considered for example in \cite{10.1063/5.0239967}. Due to the linear coupling, one can  study the growth of the fluctuations directly by estimating the quantities in \eqref{eq : cvgce in variance thm},  instead of introducing the functional $\boldsymbol{\beta}^b$ as in \eqref{eq : def of fctnals}. \\
\end{remark}

In order to ensure the well-definedness of the functional $\boldsymbol{\beta}$ during the time evolution of the Pauli--Fierz Hamiltonian, the invariance of $\mathcal{D}(\mathbb{H}_{\hbar}^{1/2})$, $\mathcal{D}(\mathcal{N}^{1/2})$ and $\mathcal{D}(\mathbf{x})$ is required, which we prove in Subsection \ref{subsect : invariance of the domains} below. Theorem \ref{thm : main result} follows from a Grönwall argument on the functional $\boldsymbol{\beta} = \boldsymbol{\beta}^a + \boldsymbol{\beta}^b + \boldsymbol{\beta}^c$. We thus have to study the time evolution of each of the three functionals, and see how they relate to each other. In the following lemma we collect the estimates which will allow us to close the argument. \\

\begin{lemma}[Growth of the functionals] \label{lemma : Growth of the functional}
Let $\sigma \in \left[1/2, 1 \right]$. Let $u_0 = (q_0,p_0,\alpha_0) \in X^\sigma$, and consider $u(t) = (q(t),p(t),\alpha_t)$ the associated solution of \eqref{eq : newton motion eqs PF}. Let $\Psi_0 \in  \mathcal{D}\left(\mathbb{H}_{\hbar}^{1/2}\right) \cap \mathcal{D}\left(\mathcal{N}^{1/2}\right) \cap \mathcal{D}(\mathbf{x})$, and let $\Psi_t = e^{-i \frac{t}{\hbar}\mathbb{H}_{\hbar}}\Psi_0$ be the associated solution of \eqref{eq : sch eq}. Then, there exists a constant $C >0$ such that the following bounds hold :
\begin{equation}
	\begin{split}
		\boldsymbol{\beta}^a(\Psi_t, q(t)) & \leq \boldsymbol{\beta}^a(\Psi_0, q_0) +  C\int_0^t \left(\boldsymbol{\beta}^a(\Psi_s, q(s)) + \boldsymbol{\beta}^b(\Psi_s, u(s)) \right) \textnormal{d}s \ ,  \\
		\boldsymbol{\beta}^b(\Psi_t, u(t)) & \leq \boldsymbol{\beta}^b(\Psi_0, u_0) + C(t+1) \int_0^t \left( \boldsymbol{\beta}(\Psi_s, u(s)) + \hbar \right) \textnormal{d}s \ , \\
		\boldsymbol{\beta}^c(\Psi_t, \alpha_t) & \leq \boldsymbol{\beta}^c(\Psi_0, \alpha_0) + C\int_0^t  \boldsymbol{\beta}(\Psi_s, u(s))  \textnormal{d}s\ . \\
	\end{split}
\end{equation}
\end{lemma}

The rest of the article is structured as follows. In Subsection \ref{subsection : Preliminary estimates} we collect  estimates that will be useful within the proofs. In Subsection \ref{subsect : invariance of the domains}, we prove the well definedness of the functional $\boldsymbol{\beta}$ during the time evolution of the Pauli--Fierz Hamiltonian, by proving the invariance of the domains on which it is initially defined. In Subsection \ref{subsect : growth of the functional} we analyze the growth of the functional $\boldsymbol{\beta}$ and give the proof of Lemma \ref{lemma : Growth of the functional}. Finally, we give in Subsection \ref{subsection : Proof of the main results} the proof of the main results, namely Theorem \ref{thm : main result}, Corollary \ref{Corollary : corollary of main theorem for pure states} and Corollary \ref{corollary : more general states}. \\

\subsection{Preliminary estimates} \label{subsection : Preliminary estimates}
Let us start by giving some preliminary results that we will use throughout the rest of the article.

\begin{lemma}[Creation and annihilation operators]
    Let $g \in \mathfrak{h}$. Let $\Psi \in \mathcal{D}\left(\mathcal{N}^{1/2}\right) \cap \mathcal{D} \big( H_{\textnormal{f}}^{1/2} \big)$. Then, 
    \begin{equation} \label{eq : estimates on creators and annihilators}
	\begin{split}
        \norm{a(g)\Psi} & \leq \norm{g}_{\mathfrak{h}} \norm{\mathcal{N}^{1/2}\Psi} \ , \\
		\norm{a^*(g)\Psi} & \leq \norm{g}_{\mathfrak{h}} \norm{\left(\mathcal{N}+1\right)^{1/2}\Psi}  \,  \\
		\norm{a(g)\Psi} & \leq \norm{|\cdot|^{-1/2}g}_{\mathfrak{h}} \norm{H_{\textnormal{f}}^{1/2}\Psi} \ , \\
		\norm{a^*(g)\Psi} & \leq \norm{\left(1+|\cdot|^{-1/2}\right)g}_{\mathfrak{h}} \norm{\left(H_{\textnormal{f}}+1\right)^{1/2}\Psi} \ . \\
	\end{split}
\end{equation}

\end{lemma}
\begin{proof}
    See \cite[Lemma 2.1]{rodnianski2007quantumfluctuationsrateconvergence} for the first two bounds. The last two follow similarly using $g = |\cdot|^{-1/2}|\cdot|^{1/2}g$. \\
\end{proof}

\begin{lemma}[Bounds on the classical electromagnetic fields] \label{lemma : estimates on the classical elec potential}
    Let $\kappa$ satisfy Assumption \eqref{assumptions : kappa}. Let $\alpha \in \dot{\mathfrak{h}}^{1/2}$. Then,  
    \begin{equation}
    \norm{\boldsymbol{\mathbb{A}}_\alpha}_{L^\infty(\mathbb{R}^3)}  \lesssim\norm{|\cdot|^{ - 1}\mathcal{F}[\kappa]}_{L^2(\mathbb{R}^3)} \norm{\alpha}_{\dot{\mathfrak{h}}^{1/2}} \ . \\
    \end{equation}
    Moreover, let $\sigma \in \left[ 1/2, 1\right]$, and $\alpha \in \mathfrak{h}^\sigma \cap \dot{\mathfrak{h}}^{\sigma}$. Then,  
    \begin{equation}
        \begin{split}
                         \norm{\nabla\boldsymbol{\mathbb{A}}_\alpha}_{L^\infty(\mathbb{R}^3)} \ , \ \norm{\boldsymbol{\mathbb{E}}_{\alpha}}_{L^\infty(\mathbb{R}^3)} & \lesssim \norm{|\cdot|^{\frac{1}{2}-\sigma}\mathcal{F}[\kappa]}_{L^2(\mathbb{R}^3)} \norm{\alpha}_{\dot{\mathfrak{h}}^{\sigma}}  \ ,  \\
             \norm{\boldsymbol{\mathbb{A}}_\alpha}_{L^\infty(\mathbb{R}^3)}  &\lesssim\norm{|\cdot|^{-\frac{1}{2}} \left(1+|\cdot|^2\right)^{-\frac{\sigma}{2}}\mathcal{F}[\kappa]}_{L^2(\mathbb{R}^3)} \norm{\alpha}_{\mathfrak{h}^\sigma}  \ ,  \\
            \norm{\boldsymbol{\mathbb{E}}_{\alpha}}_{L^\infty(\mathbb{R}^3)} & \lesssim \norm{|\cdot|^{\frac{1}{2}}\left(1+|\cdot|^2 \right)^{-\frac{\sigma}{2}}\mathcal{F}[\kappa]}_{L^2(\mathbb{R}^3)} \norm{\alpha}_{\mathfrak{h}^\sigma} \ . \\
        \end{split}
    \end{equation}
\end{lemma}
\begin{proof}
    The bounds are a direct application of 
    \begin{equation}
    \begin{split}
        \boldsymbol{\mathbb{A}}_\alpha(q) & = 2 \Ree \sum_{\lambda=1,2} \int_{\mathbb{R}^3} \frac{\mathcal{F}[\kappa](k)}{\sqrt{2|k|}} \boldsymbol{\epsilon}_\lambda(k) e^{ik \cdot q}\alpha(k,\lambda) \textnormal{d}k \\
        \boldsymbol{\mathbb{E}}_\alpha(q) & = 2 \Ree \sum_{\lambda=1,2} \int_{\mathbb{R}^3} \sqrt{\frac{|k|}{2}} \mathcal{F}[\kappa](k) \boldsymbol{\epsilon}_\lambda(k) e^{ik \cdot q}\alpha(k,\lambda) \textnormal{d}k \ , \\
    \end{split}
\end{equation}
together with \eqref{eq : def norm alpha} and the Cauchy-Schwarz inequality.  The finiteness of the  weighted $L^2$ norms of $\mathcal{F}[\kappa]$ follows from splitting the respective integrals in high and low modes, as well as Assumption \ref{assumptions : kappa}. \\
\end{proof}

The next lemma gives bounds on some operators, which involve the quantized electromagnetic field $\boldsymbol{\hat{\mathbb{A}}}$.
\begin{lemma}[Bounds on the quantized electromagnetic field] \label{lemma : Bounds on the quantized electromagnetic field}
    Let $\sigma \in \left[ 1/2, 1\right]$, and let $\alpha \in \mathfrak{h}^\sigma \cap \dot{\mathfrak{h}}^{\sigma}$. Consider the classical and quantum kinetic momenta defined respectively in \eqref{eq : physical momentum} and \eqref{eq : physical momentum quantum}, and the quantum Faraday tensor defined in \eqref{eq : quantum faraday tensor}. The following bound holds in the sense of quadratic forms:
\begin{equation} \label{eq : operator inequality that makes beta b wd}
	\left(P - \Tilde{p}\right)^2 \lesssim_{N,\hbar,\kappa} - \Delta + \mathcal{N} + 1 \ . \\
\end{equation}    
Moreover,  let $\psi \in \mathcal{D} \big(\mathbb{H}_{\hbar}^{1/2} \big)$. Then, 
    \begin{equation} \label{eq : estimates in lemma to show cvgce}
        \begin{split}
            \norm{P^\ell \psi} & \lesssim_{N, \hbar, \kappa} \norm{\left( \mathbb{H}_{\hbar} + 1 \right)^{1/2}\psi}   \ ,  \\
            \norm{\boldsymbol{\hat{\mathbb{A}}}^\ell (x)  \psi}  , \ \norm{\partial_m \boldsymbol{\hat{\mathbb{A}}}^\ell (x)  \psi} & \lesssim_{N, \hbar, \kappa} \norm{\left( \mathbb{H}_{\hbar} + 1 \right)^{1/2}\psi}  \ , \\
            \norm{\boldsymbol{\hat{\mathbb{E}}}^\ell (x)  \psi} & \lesssim_{N, \hbar, \kappa} \norm{\left( \mathbb{H}_{\hbar} + 1 \right)^{1/2}\psi} \ ,\\
        \end{split}
    \end{equation}
 for all $\ell = 1,2,3$    In addition, if $\psi \in \mathcal{D}\left(\mathbb{H}_{\hbar} \right)$, then  
    \begin{equation} \label{eq : bound on quantum faraday}
                \norm{\left \lvert \hat{F}^{\ell m}\right \lvert^2\psi} \lesssim_{N, \hbar, \kappa} \norm{\left( \mathbb{H}_{\hbar} + 1 \right)\psi} \ , \\
    \end{equation}
    for all $\ell, m = 1,2,3$.
\end{lemma}
\begin{proof}
    The bound \eqref{eq : operator inequality that makes beta b wd} is a consequence of Lemma \ref{lemma : estimates on the classical elec potential} and \eqref{eq : estimates on creators and annihilators}. To prove the bounds in \eqref{eq : estimates in lemma to show cvgce}, one first uses the definitions \eqref{eq : quantized magnetic field} and \eqref{eq : quantized electric field} together with the last bound in \eqref{eq : estimates on creators and annihilators} in order to relate the LHS quantities to the free Hamiltonian $\mathbb{H}_{\hbar}^0$, and then conclude with the inequality
    \begin{equation}
        \norm{\left(\mathbb{H}_{\hbar}^0 + 1 \right)^{1/2}\psi} \lesssim_{N,\hbar,\kappa}\norm{\left( \mathbb{H}_{\hbar} + 1 \right)^{1/2}\psi} \ , \\
    \end{equation}
    which is proven in \cite[Corollary 1.4]{Hiroshima2002}. The estimate in \eqref{eq : bound on quantum faraday} follows similarly. \\
\end{proof}

\begin{lemma}[Properties of the smeared Coulomb potential] \label{lemma : Properties of the smeared Coulomb potential}
    Let $\kappa$ satisfy Assumption \eqref{assumptions : kappa}, and consider the smeared Coulomb potential defined by
    \begin{equation}
        V = \kappa \ast \kappa \ast |\cdot|^{-1} \ . \\
    \end{equation}
    Then, we have $V \in \mathcal{C}^2_{\textnormal{b}}(\mathbb{R}^3 ; \mathbb{R})$, and the following bounds hold
    \begin{equation} \label{eq : pointwise bounds for V and nabla V}
        \begin{split}
            \norm{V}_{L^\infty(\mathbb{R}^3)} & \leq \norm{|\cdot|^{-1}\mathcal{F}[\kappa]}_{L^2(\mathbb{R}^3)}^2 \\
            \norm{\partial_\ell V}_{L^\infty(\mathbb{R}^3)} & \leq \norm{|\cdot|^{-1/2}\mathcal{F}[\kappa]}_{L^2(\mathbb{R}^3)}^2 \ , \quad \ell = 1,2,3 \\
            \norm{\partial_m \partial_\ell V}_{L^\infty(\mathbb{R}^3)} & \leq \norm{\mathcal{F}[\kappa]}_{L^2(\mathbb{R}^3)}^2 \ , \quad m,\ell = 1,2,3 \ . \\
        \end{split}
    \end{equation}
    Moreover, we have 
    \begin{equation} \label{eq : growth of nabla V}
        \left \lvert \nabla V (x) \right \lvert \lesssim_\kappa |x| \\
    \end{equation}
    almost everywhere, and the constant depends on $\kappa$ via $\norm{\mathcal{F}[\kappa]}_{L^3(\mathbb{R}^3)}$. \\
\end{lemma}
\begin{proof}
    This is mainly \cite[Lemma 2.1]{ammari2022derivationclassicalelectrodynamicscharges}. Let us briefly give the argument for completeness. The fact that $V \in \mathcal{C}^2_{\textnormal{b}}(\mathbb{R}^3 ; \mathbb{R})$ follows from the identities
    \begin{equation} \label{eq : identity for fourier V}
	\begin{split}
    \mathcal{F}[V](k) & = 4 \pi |k|^{-2}|\mathcal{F}[\kappa](k)|^2\\
	\mathcal{F}[\partial_\ell V](k) & = 4 i \pi \frac{k_\ell}{|k|^2} |\mathcal{F}[\kappa](k)|^2 \\
	\mathcal{F}[\partial_m \partial_\ell V](k) & = -4\pi \frac{k_m k_\ell}{|k|^2} |\mathcal{F}[\kappa](k)|^2 \ , \\
	\end{split}
    \end{equation}
    which hold in $\mathcal{S}'(\mathbb{R}^3)$.  Assumption \ref{assumptions : kappa} implies that the RHS of the above are in $L^1(\mathbb{R}^3)$, which then yield that $V$, $\partial_\ell V$ and $\partial^2_\ell V$ are well-defined continuous functions, by the Riemann-Lebesgue lemma. The bounds \eqref{eq : pointwise bounds for V and nabla V} then follow from taking the inverse Fourier transform. To prove \eqref{eq : growth of nabla V}, we use the fundamental theorem of calculus to write 
    \begin{equation} \label{eq : FTC nabla V}
        \nabla V(x) - \nabla V(y) = \int_0^1 \nabla^2V (tx + (1-t)y) \cdot (x-y) \textnormal{d}t \ , \\
    \end{equation}
    where $\nabla^2$ stands for the hessian operator. Next, note that \eqref{eq : identity for fourier V} above allows us to write
    \begin{equation}
        \begin{split}
            \nabla V(0) & = 4i\pi \int_{\mathbb{R}^3} \frac{k}{|k|^2} |\mathcal{F}[\kappa](k)|^2 \textnormal{d}k = 0 \ , \\
        \end{split}
    \end{equation}
    since by assumption, $\mathcal{F}[\kappa]$ is an even function, hence making the integrand in the above an odd function of $k$. Moreover, we have 
    \begin{equation}
        \norm{\nabla^2 V}_{L^\infty(\mathbb{R}^3)} := \sup_{x \in \mathbb{R}^3} \norm{\nabla^2 V(x)}_{\textnormal{op}} \lesssim \norm{\mathcal{F}[\kappa]}_{L^2(\mathbb{R}^3)}^2 \ . \\
    \end{equation}
    The bound \eqref{eq : growth of nabla V} then follows from \eqref{eq : FTC nabla V} evaluated at $y = 0$ and the previous bounds. \\
\end{proof}

\begin{lemma}[A priori estimates for the Newton--Maxwell solutions] \label{lemma : A priori estimates for the Newton--Maxwell solutions}
    Let $\sigma \in \left[1/2,1 \right]$. Let $u_0 = (q_0, p_0, \alpha_0) \in X^\sigma$, and consider $u(t) = (q(t), p(t), \alpha_t)$ the associated solution to the Newton--Maxwell system \eqref{eq : newton motion eqs PF}. Then, there exists a constant $C > 0$ depending on $\kappa$ and $\norm{u_0}_{X^\sigma}$ such that, for all $t \geq 0$
    \begin{equation} \label{eq : first bounds in a priori for NM}
        \begin{split}
          \sup_{j \in \{1, \ldots,N \}} \ |p_j(t)| & \leq C \ ,  \\
            \norm{\alpha_t}_{\dot{\mathfrak{h}}^{1/2}} & \leq C \ ,  \\
            \norm{u(t)}_{X^\sigma} & \leq C(t+1) \ . \\
        \end{split}
    \end{equation}
    Moreover, under the same assumptions there exist constants $\Tilde{C}$ (which does not depend on $\norm{u_0}_{X^\sigma}$) and $c > 0$ depending on $\kappa$ and $\sigma$ such that for all $t \geq 0$ 
    \begin{equation} \label{eq : second bounds in a priori for NM}
        \norm{u(t)}_{X^\sigma} \leq \Tilde{C}e^{ct}  \norm{u_0}_{X^\sigma} \ . \\
    \end{equation}
\end{lemma}

\begin{remark}
The bound \eqref{eq : second bounds in a priori for NM}, even though its dependence on $t$ is worse than that of  \eqref{eq : first bounds in a priori for NM}, is useful  for several purposes.  On the one hand, it is more suited for density arguments (see the proof of Lemma \ref{lemma : growth of beta c} in Appendix \ref{appendix : Growth of the third functional} below) thanks to its linear dependance on $\norm{u_0}_{X^\sigma}$. On the other hand, it allows us to check Assumption \ref{eq : assumption on mu} when proving Corollary \ref{corollary : more general states}. However we will still make use of \eqref{eq : first bounds in a priori for NM} in order to prove Lemma \ref{lemma : Growth of the functional}, as it yields a better time dependence than would be obtained using \eqref{eq : second bounds in a priori for NM}. \\
\end{remark}

\begin{proof}
Let us start by proving \eqref{eq : second bounds in a priori for NM}. We use the shorthand notation $F^{\ell m}_{\alpha_t}(q_j(t)) = F^{j \ell m}_t$ for the classical Faraday tensor. To begin with, notice that if we define for any $u \in X^\sigma$ the quantity $\Tilde{u} = (q, \Tilde{p}, \alpha) \in X^\sigma$, where $\Tilde{p} = q - \boldsymbol{\mathbb{A}}_{\alpha}(q)$ is the kinetic momentum defined in \eqref{eq : physical momentum}, then there holds 
\begin{equation} \label{eq : equivalence between X sigma and tilde X sigma}
    C^{-1} \norm{u}_{X^\sigma} \leq \norm{\Tilde{u}}_{X^\sigma} \leq C \norm{u}_{X^\sigma} \ , \\
\end{equation}
for some constant $C > 0$ (depending on $N$ and $\kappa$). This comes from the simple bounds
\begin{equation} \label{eq : p tilda controled by p}
    \begin{split}
        |p_j(t)| & \leq |\Tilde{p}_j(t)| + C\norm{\alpha_t}_{\mathfrak{h}^\sigma} \\
        |\Tilde{p}_j(t)| & \leq |p_j(t)| + C\norm{\alpha_t}_{\mathfrak{h}^\sigma} \ , \\
    \end{split}
\end{equation}
which follow from \eqref{eq : physical momentum} and Lemma \ref{lemma : estimates on the classical elec potential}. In view of the heuristic explanation in Section \ref{subsec : Strategy of the proof}, we rather want to estimate $\norm{\Tilde{u}(t)}^2_{X^\sigma}$ in order to have a good control on $\norm{u(t)}^2_{X^\sigma}$. From \eqref{eq : newton motion eqs PF}, we have  
\begin{equation}
    \begin{split}
        \dt | q_j(t) | ^2 & = 4 q_j(t) \cdot \Tilde{p}_j(t) \\
        & \lesssim | q_j(t) | ^2 + | \Tilde{p}_j(t)| ^2 \ , \\
    \end{split}
\end{equation}
so that 
\begin{equation} \label{eq : bound duhamel q}
    | q_j(t) |^2 \leq | q_{j,0} | ^2 + C \int_0^t \left( | q_j(s) | ^2 + | \Tilde{p}_j(s)| ^2 \right) \textnormal{d}s \ , \\
\end{equation}
for some constant $C > 0$. Next, differentiating $\Tilde{p}_j$ using the chain rule yields :
\begin{equation} \label{eq : expression for derivative of p tilda}
    \partial_t \Tilde{p}^m_j =  \boldsymbol{\mathbb{E}}^m_{\alpha_t}(q_j) + 2 \sum_{\ell = 1}^3 \Tilde{p}^\ell_j F^{j \ell m}_t - \sum_{k \neq j} \left( \partial_{m} V \right) \left(q_j - q_k \right)  \\
\end{equation}
for all $m = 1,2,3$. We thus obtain :
\begin{equation}
    \begin{split}
        \dt |\Tilde{p}_j(t)|^2 & = 2 \sum_{m = 1}^3\Tilde{p}_j^m(t) \partial_t\Tilde{p}_j^m(t) \\
        & = 4\sum_{m=1}^3 \sum_{\ell = 1}^3 \Tilde{p}_j^m(t) \Tilde{p}_j^\ell(t) F^{j \ell m}_t + 2 \Tilde{p}_j(t) \cdot \boldsymbol{\mathbb{E}}_{\alpha_t}(q_j(t)) - 2 \sum_{k \neq j} \Tilde{p}_j(t) \cdot (\nabla V)(q_j(t) - q_k(t))  \ . \\
    \end{split}
\end{equation}
Now, note that the antisymmetry of the Faraday tensor implies
\begin{equation} \label{eq : antisymmetry cancels the term}
    \sum_{m=1}^3 \sum_{\ell = 1}^3 \Tilde{p}_j^m(t) \Tilde{p}_j^\ell(t) F^{j \ell m}_t = 0 \ . \\
\end{equation}
Hence by means of a Duhamel formula we control :
\begin{equation} \label{eq : bound duhamel p tilda}
    \begin{split}
        |\Tilde{p}_j(t)|^2 & \leq |\Tilde{p}_{j,0}|^2 + 2\int_0^t |\Tilde{p}_j(s)| \norm{\boldsymbol{\mathbb{E}}_{\alpha_s}}_{L^\infty(\mathbb{R}^3)} \textnormal{d}s + 2 \int_0^t |\Tilde{p}_j(s)| \sum_{k \neq j} \left \lvert \nabla V(q_j(s) - q_k(s)) \right \lvert \textnormal{d}s \\
        & \leq |\Tilde{p}_{j,0}|^2 + C\int_0^t \left( |\Tilde{p}_j(s)|^2 +\norm{\alpha_s}^2_{\mathfrak{h}^\sigma} \right) \textnormal{d}s + C \int_0^t \sum_{k \neq j} |\Tilde{p}_j(s)|  | q_j(s) - q_k(s) | \textnormal{d}s \\
        & \leq |\Tilde{p}_{j,0}|^2 + C\int_0^t \left( |\Tilde{p}_j(s)|^2 + \sum_{j=1}^N |q_j(s)|^2  + \norm{\alpha_s}^2_{\mathfrak{h}^\sigma} \right) \textnormal{d}s \ . \\
    \end{split}
\end{equation}
Here we used again Lemma \ref{lemma : estimates on the classical elec potential}, and \eqref{eq : growth of nabla V} in the second line. Let us now control $\norm{\alpha_t}_{\mathfrak{h}^\sigma}^2$.  Using the equation for $\alpha_t$ in \eqref{eq : newton motion eqs PF} allows us to derive the Duhamel formula 
\begin{equation}
\norm{\alpha_t}^2_{\mathfrak{h}^\sigma} = \norm{\alpha_0}^2_{\mathfrak{h}^\sigma} - \int_0^t 4 \Imm \sum_{\lambda = 1,2} \int_{\mathbb{R}^3} (1+|k|^2)^\sigma \overline{\alpha_s(k,\lambda)}  \sum_{j=1}^N \mathbf{G}_{q_j(s)}(k,\lambda) \cdot \Tilde{p}_j(s) \textnormal{d}k \textnormal{d}s \ , \\
\end{equation}
which implies
\begin{equation} \label{eq : bound for alpha t in gronwall}
    \norm{\alpha_t}^2_{\mathfrak{h}^\sigma} \leq \norm{\alpha_0}_{\mathfrak{h}^\sigma}^2  + C \int_0^t \left( \sum_{j=1}^N |\Tilde{p}_j(s)|^2 + \norm{\alpha_s}^2_{\mathfrak{h}^\sigma} \right) \textnormal{d}s \ . \\
\end{equation}
Collecting \eqref{eq : bound duhamel q}, \eqref{eq : bound duhamel p tilda} and \eqref{eq : bound for alpha t in gronwall} yields
\begin{equation}
    \norm{\Tilde{u}(t)}^2_{X^\sigma} \leq \norm{\Tilde{u}_0}^2_{X^\sigma} + C \int_0^t \norm{\Tilde{u}(s)}^2_{X^\sigma} \textnormal{d}s \ , \\
\end{equation}
which together with Grönwall's lemma gives \eqref{eq : second bounds in a priori for NM}. Let us now turn to the proof of \eqref{eq : first bounds in a priori for NM}, which uses the conservation of the energy. Indeed, in view of \eqref{eq : energy NM} and \eqref{eq : pointwise bounds for V and nabla V}, we have
\begin{equation} \label{eq : energy estimates in proof lemma}
    \begin{split}
        \norm{\alpha_t}_{\dot{\mathfrak{h}}^{1/2}} & \lesssim_\kappa \mathcal{H}(u(t))^{1/2} + 1 \\
        \sup_{j \in \{1, \ldots,N \}}  \ |p_j(t)| & \lesssim_\kappa \mathcal{H}(u(t))^{1/2} + 1 \ . \\
    \end{split}
\end{equation}
The above combined with $\mathcal{H}(u(t)) = \mathcal{H}(u_0)$ (see Lemma \ref{lemma : conservation of the energy}) and the bound 
\begin{equation} \label{eq : bound on the energy of u0}
    \mathcal{H}(u_0) \lesssim_\kappa \norm{u_0}_{X^\sigma}^2 + 1 \\
\end{equation}
which follows from $\norm{\alpha}_{\dot{\mathfrak{h}}^{1/2}} \leq \norm{\alpha}_{\mathfrak{h}^\sigma}$, give the first two bounds in \eqref{eq : first bounds in a priori for NM}. In order to prove the third bound, we first note that from \eqref{eq : energy estimates in proof lemma} and Lemma \ref{lemma : estimates on the classical elec potential}, there holds :
\begin{equation} \label{eq : energy estimate on tilde p in proof lemma}
\sup_{j \in \{1, \ldots,N \}}  |\Tilde{p}_j(t)| \lesssim_\kappa \mathcal{H}(u_0)^{1/2} + 1 \ . \\
\end{equation}
Together with the definition of $\norm{\cdot}_{X^\sigma}$ and \eqref{eq : equivalence between X sigma and tilde X sigma}, this implies
\begin{equation} \label{eq : energy estimate on u t in proof lemma}
\norm{u(t)}_{X^\sigma} \lesssim_{\kappa} \mathcal{H}(u_0)^{1/2} + \sup_{j \in \{1, \ldots,N \}}  \ |q_j(t)| + \norm{\alpha_t}_{\mathfrak{h}^\sigma} + 1
\ . \\
\end{equation}
In order to bound $\norm{\alpha_t}_{\mathfrak{h}^\sigma}$, we use the Duhamel formula
\begin{equation}
\alpha_t(k,\lambda) = e^{-it|k|}\alpha_0(k,\lambda) + 2i \int_0^t e^{-i(t-s)|k|} \sum_{j=1}^N \mathbf{G}_{q_j(s)}(k,\lambda) \cdot \Tilde{p}_j(s) \textnormal{d}s \ , \\
\end{equation}
the triangle inequality, the bounds \eqref{eq : bound on the energy of u0}, \eqref{eq : energy estimate on tilde p in proof lemma} and $\norm{\mathbf{G}_q}_{\mathfrak{h}^\sigma} \lesssim_\kappa 1$ to estimate
\begin{equation} \label{eq : final bound alpha t in lemma}
\norm{\alpha_t}_{\mathfrak{h}^\sigma} \leq \norm{\alpha_0}_{\mathfrak{h} ^\sigma} + Ct \ , \\
\end{equation}
where the constant $C$ depends (not linearly) on $\norm{u_0}_{X^\sigma}$.  Similarly,  we get
\begin{align}
\label{eq : final bound q t in lemma}
\begin{split}
\sup_{j \in \{1, \ldots,N \}}  \ |q_j(t)| &\leq \sup_{j \in \{1, \ldots,N \}}  \ |q_j(0)| + \int_0^t \sup_{j \in \{1, \ldots,N \}}  \ |\Tilde{p}_j(s)| \, ds
\\
&\leq \sup_{j \in \{1, \ldots,N \}}  \ |q_j(0)|  + Ct
\end{split} 
\end{align} 
by means of \eqref{eq : energy estimate on tilde p in proof lemma}. Combining \eqref{eq : bound on the energy of u0},  \eqref{eq : energy estimate on u t in proof lemma},   \eqref{eq : final bound alpha t in lemma},  and
\eqref{eq : final bound q t in lemma} gives the desired result. \\
\end{proof}

\subsection{Invariance of the domains} \label{subsect : invariance of the domains}

In order to justify the well-definedness of the functional $\boldsymbol{\beta}$ under the time evolution of the Pauli--Fierz Hamiltonian, we need to ensure that the domains involved in Definition \ref{def : Definition of the functional} are invariant under the action of $e^{-i\frac{t}{\hbar}\mathbb{H}_{\hbar}}$. This is proven in the following lemma :
\begin{lemma} \label{lemma : invariance of the domains for beta}
Let $\kappa$ satisfy Assumption \ref{assumptions : kappa}, and consider $\mathbb{H}_\hbar$ the Pauli--Fierz Hamiltonian defined in \eqref{eq : PF Hamiltonian}. Then for all $t \geq 0$, there holds
\begin{equation} \label{eq : invariance of the domains for beta}
    e^{-i\frac{t}{\hbar}\mathbb{H}_{\hbar}} \mathcal{D}\left(\mathbb{H}_{\hbar}^{1/2}\right) \cap \mathcal{D}\left(\mathcal{N}^{1/2}\right) \cap \mathcal{D}(\mathbf{x}) = \mathcal{D}\left(\mathbb{H}_{\hbar}^{1/2}\right)\cap \mathcal{D}\left(\mathcal{N}^{1/2}\right) \cap \mathcal{D}(\mathbf{x})  \ . \\
\end{equation}
\end{lemma}
\begin{proof}
The invariance of $\mathcal{D}\left(\mathbb{H}_{\hbar}^{1/2}\right)$ is implied by Stone's theorem, and that of $\mathcal{D}\left(\mathcal{N}^{1/2}\right)$ is proven in complete analogy to \cite[Lemma 3.2]{leopold2024derivationmaxwellschrodingervlasovmaxwellequations}. There remains to prove the invariance of $\mathcal{D}(\mathbf{x})$, which we obtain by proving the following fact : consider $\Psi \in \mathcal{D}\left(\mathbb{H}_{\hbar}^{1/2}\right) \cap \mathcal{D}(\mathbf{x})$. Then, there exist constants $C, c> 0$ such that for all $t \geq 0$,
    \begin{equation} \label{eq : invariance D x2 conclu}
    \norm{\mathbf{x} e^{-i\frac{t}{\hbar}\mathbb{H}_{\hbar}} \Psi} \leq Ce^{ct} 	\left( \norm{\mathbf{x}\Psi} + \norm{\left(\mathbb{H}_{\hbar} + 1\right)^{1/2} \Psi} \right) \ . \\
    \end{equation}    
In order to establish \eqref{eq : invariance D x2 conclu} in a rigorous manner, we introduce the following regularization of the position operator :
\begin{equation} \label{eq : regularized position op}
\mathbf{x}_{n} = \frac{\mathbf{x}}{(1 + n^{-2}|\mathbf{x}|^2)^{1/2}} \quad \forall n \geq 1 \ , \\
\end{equation}
with entries noted $x_{n,j}$, $j = 1 \dots, N$. It is a bounded symmetric operator, bounded by $n$. Let 
\begin{equation}
\Psi_t = e^{-i\frac{t}{\hbar}\mathbb{H}_{\hbar}} \Psi \\
\end{equation}
denote the solution of \eqref{eq : sch eq} with initial data $\Psi$. Then, we compute :
\begin{equation} \label{eq : derivative xn 1}
	\begin{split}
		\dt	\norm{\mathbf{x}_n \Psi_t}^2 & = 2 \Ree \sum_{j=1}^N \sum_{\ell = 1}^3 \ps{\Psi_t}{x_{n,j}^\ell i \hbar^{-1}\com{\mathbb{H}_{\hbar}}{x_{n,j}^\ell}\Psi_t} \\
		& \lesssim \sum_{j=1}^N \sum_{\ell = 1}^3 \norm{x_{n,j}^\ell\Psi_t} \norm{i \hbar^{-1}\com{\mathbb{H}_{\hbar}}{x_{n,j}^\ell} \Psi_t} \ . \\
	\end{split}
\end{equation}
Next, denoting $\left\{A,B\right\} = AB + BA$ the anticommutator of two operators $A$ and $B$, we expand the commutator in \eqref{eq : derivative xn 1} as :
\begin{equation} \label{eq : com H x}
	\begin{split}
		i \hbar^{-	1}\com{\mathbb{H}_{\hbar}}{x_{n,j}^\ell} & = \sum_{k=1}^N \sum_{m=1}^3 \left\{ -i \hbar \partial_{km} - \hbar^{1/2} \boldsymbol{\hat{\mathbb{A}}}^m(x_k), \com{\partial_{km}}{x^\ell_{n,j}} \right\} \\
		& = -2 \hbar^{1/2} \sum_{k=1}^N \sum_{m=1}^3  \boldsymbol{\hat{\mathbb{A}}}^m(x_k) \partial_{km}x^\ell_{n,j} -2i \hbar \sum_{k=1}^N \sum_{m=1}^3 \partial_{km}x^\ell_{n,j}\partial_{km} - i \hbar \sum_{k=1}^N \sum_{m=1}^3 \partial^2_{k m}x^\ell_{n,j} \ . \\
	\end{split}
\end{equation}
Denote $F_n(\mathbf{x}) = (1 + n^{-2}|\mathbf{x}|^2)^{-1/2}$. A computation gives 
\begin{equation} \label{eq : 1st and 2nd derivative xnjell}
	\begin{split}
		\partial_{km}x^\ell_{n,j} & = F_n(\mathbf{x}) \delta_{jk} \delta_{\ell m} - n^{-2} F_n(\mathbf{x}) x^m_{n,k} x^\ell_{n,j} \\
		\partial^2_{k m}x^\ell_{n,j} & = -2 n^{-2} F_n(\mathbf{x})^2 x^m_{n,k}\delta_{jk} \delta_{\ell m} + 3 n^{-4} F_n(\mathbf{x})^2 (x^m_{n,k})^2 x^\ell_{n,j} - n^{-2}F_n(\mathbf{x})^2 x^\ell_{n,j} \ . \\
	\end{split}
\end{equation}
Inserting this in \eqref{eq : com H x} and using that $F_n(\mathbf{x}) \leq 1$, $\|\mathbf{x}_n\|_{\textnormal{op}} \leq n$ and $\com{\kappa \ast \mathbf{A}_m(x_k)}{x^m_{n,k} x^\ell_{n,j}} = 0$, we find
\begin{equation}
\begin{split}
	\norm{i \hbar^{-1}\com{\mathbb{H}_{\hbar}}{x_{n,j}^\ell} \Psi_t} & \lesssim \hbar^{1/2} \norm{\boldsymbol{\hat{\mathbb{A}}}^\ell(x_j) \Psi_t} + \hbar^{1/2} \sum_{k=1}^N \sum_{m=1}^3 \norm{\boldsymbol{\hat{\mathbb{A}}}^m(x_k)\Psi_t} \\
	& \quad \quad + \hbar \norm{\partial_{j,l} \Psi_t} + \hbar \sum_{k=1}^N \sum_{m=1}^3 \norm{\partial_{km} \Psi_t} \\
	& \quad \quad + \hbar n^{-2} \norm{x^\ell_{n,j} \Psi_t} + \hbar n^{-2} \sum_{k=1}^N \sum_{m=1}^3 \norm{x^m_{n,k} \Psi_t} \ . \\
\end{split}
\end{equation}
Hence, going back to \eqref{eq : derivative xn 1} and applying several times Young's inequality, we obtain 
\begin{equation} \label{eq : derivative xn 2}
	\begin{split}
		\dt	\norm{\mathbf{x}_n \Psi_t}^2 & \lesssim \hbar^{1/2} \left( \norm{\mathbf{x}_n \Psi_t}^2 + \sum_{j=1}^N \norm{\boldsymbol{\hat{\mathbb{A}}}^\ell(x_j) \Psi_t}^2  \right) \\
		& \quad \quad + \hbar \left( \norm{\mathbf{x}_n \Psi_t}^2 + \sum_{j=1}^N \norm{\nabla_j \Psi_t}^2 \right) + \hbar n^{-2} \norm{\mathbf{x}_n \Psi_t}^2 \\
		& \lesssim (\hbar^{1/2} + \hbar(1+n^{-2})) \norm{\mathbf{x}_n \Psi_t}^2 + \hbar^{1/2}\norm{\left(|\cdot|^{-1/2} + |\cdot|^{-1} \right) \mathcal{F}[\kappa]}^2_{L^2(\mathbb{R}^3)} \norm{\left(H_{\textnormal{f}} + 1\right)^{1/2}\Psi_t}^2 \\
		& \quad \quad + \hbar \ps{\Psi_t}{\sum_{j=1}^N(-\Delta_j) \Psi_t} \\
		& \lesssim (\hbar^{1/2} + \hbar(1+n^{-2})) \norm{\mathbf{x}_n \Psi_t}^2  \\
		& \quad \quad + C  \left( \hbar + \hbar^{1/2}\norm{\left(|\cdot|^{-1/2} + |\cdot|^{-1} \right) \mathcal{F}[\kappa]}^2_{L^2(\mathbb{R}^3)} \right) \norm{\left(\mathbb{H}_{\hbar} + 1\right)^{1/2}\Psi_t}^2 \ . \\
	\end{split}
\end{equation}
Here in the second line, we used the fact that $\boldsymbol{\hat{\mathbb{A}}}(x) = a(\mathbf{G}_x) + a^*(\mathbf{G}_x)$, together with \eqref{eq : estimates on creators and annihilators}, and in the third line we used inequalities from \cite[Equation 3.15]{falconi2023derivation} and the subsequent discussion. Hence, applying Grönwall's lemma, we obtain
\begin{equation}
	\norm{\mathbf{x}_n \Psi_t}^2 \leq e^{ct} \left( \norm{\mathbf{x}_n \Psi}^2 + C \int_0^t \norm{\left(\mathbb{H}_{\hbar} + 1\right)^{1/2}\Psi_s}^2 \textnormal{d}s \right) \\
\end{equation}
for some constants $C,c > 0$ depending on $N,\hbar,\kappa$. Taking the limit $n \rightarrow \infty$ on both sides of the above inequality and using monotone convergence allows us to conclude the claim, given the initial assumptions on $\Psi$. \\
\end{proof}

\subsection{Growth of the functional} \label{subsect : growth of the functional}
In this subsection we will analyze each functional defined in Definition~\ref{def : Definition of the functional} in order to prove Lemma \ref{lemma : Growth of the functional}. To do so, we state and prove three lemmas which give us its change in time.

\begin{lemma}
    Let $\sigma \in \left[1/2, 1 \right]$. Let $u_0 = (q_0,p_0,\alpha_0) \in X^\sigma$ and $\Psi_0 \in  \mathcal{D}(\mathbb{H}_{\hbar}^{1/2}) \cap \mathcal{D}(\mathcal{N}^{1/2}) \cap \mathcal{D}(\mathbf{x})$. Let $u(t) = (q(t),p(t),\alpha_t)$ and  $\Psi_t = e^{-i \frac{t}{\hbar}\mathbb{H}_{\hbar}}\Psi_0$ be the respective solutions of \eqref{eq : newton motion eqs PF} and \eqref{eq : sch eq} with initial data $u_0$ and $\Psi_0$. Then,
    \begin{equation} \label{eq : change of beta a in time}
        		 \boldsymbol{\beta}^a(\Psi_t, q(t)) - \boldsymbol{\beta}^a(\Psi_0, q_0) 
        		= \int_0^t 4 \Ree \sum_{j=1}^N \ps{\Psi_s}{(x_j - q_j(s)) \mathbb{G}_{j,s} \Psi_s} \textnormal{d}s \ .\\	
    \end{equation}
\end{lemma}

\begin{proof}
We start by defining 
\begin{equation}
\boldsymbol{\beta}^a_n(\Psi_t, q(t)) = \ps{\Psi_t}{ \left( \mathbf{x}_{n} - q(t) \right)^2 \Psi_t}_{\mathfrak{H}} ,
\end{equation}
where $\mathbf{x}_n$ is the regularized position operator defined in \eqref{eq : regularized position op}. First, we see that
\begin{equation} \label{eq : cvgce beta a n to beta a}
	\begin{split}
		\left \lvert \boldsymbol{\beta}^a_n(\Psi_t, q(t)) - \boldsymbol{\beta}^a(\Psi_t, q(t)) \right \lvert & \leq \norm{(\mathbf{x}_n - \mathbf{x}) \Psi_t} \norm{\left(\mathbf{x}_n + \mathbf{x} - 2 q(t) \right) \Psi_t} \underset{n \rightarrow + \infty}{\longrightarrow} 0 \\
	\end{split}
\end{equation}
by monotone convergence. Now, using the first line of \eqref{eq : com H x} and \eqref{eq : 1st and 2nd derivative xnjell}, we compute :
\begin{equation} \label{eq : beta a n derivative 1}
    \begin{split}
        \dt \boldsymbol{\beta}_n^a(t) & = 4 \Ree \sum_{j,k=1}^N \sum_{\ell,m=1}^3 \ps{\Psi_t}{\left( x^\ell_{n,j} - q^\ell_j(t)\right)\left(-i\hbar \partial_{km} - \hbar^{1/2}\boldsymbol{\hat{\mathbb{A}}}^m(x_{k})\right) \com{\partial_{km}}{x^\ell_{n,j}}  \Psi_t} \\
        & \quad \quad + 2 \Ree \  i \hbar \sum_{j,k=1}^N \sum_{\ell,m=1}^3 \ps{\Psi_t}{ \left( x^\ell_{n,j} - q^\ell_j(t)\right)\left(\partial^2_{km}x^\ell_{n,j}\right)  \Psi_t} \\
        & \quad \quad - 4 \Ree \sum_{j=1}^N \sum_{\ell=1}^3 \ps{\Psi_t}{\left( x^\ell_{n,j} - q^\ell_j(t)\right) \left( p^\ell_j(t) - \boldsymbol{\mathbb{A}}^\ell_{\alpha_t}(q_j(t)) \right)\Psi_t} \\
        & = 4 \Ree \sum_{j=1}^N \sum_{\ell = 1}^3 \ps{\Psi_t}{\left( x^\ell_{n,j} - q^\ell_j(t)\right) \left(\left(-i\hbar \partial_{j \ell} - \hbar^{1/2}\boldsymbol{\hat{\mathbb{A}}}^\ell(x_{j})\right)F_n(\mathbf{x}) - \left( p^\ell_j(t) - \boldsymbol{\mathbb{A}}^\ell_{\alpha_t}(q_j(t)) \right) \right)\Psi_t} \\
         & \quad \quad -4n^{-2} \Ree \sum_{j,k=1}^N \sum_{\ell,m = 1}^3 \ps{\Psi_t}{\left( x^\ell_{n,j} - q^\ell_j(t)\right)\left(-i\hbar \partial_{km} - \hbar^{1/2}\boldsymbol{\hat{\mathbb{A}}}^m(x_{k})\right)F_n(\mathbf{x})^2x^m_{n,k} x^\ell_{j}\Psi_t} \\
        & := \textnormal{I}_{n} + \textnormal{II}_n \ . \\
    \end{split}
\end{equation}
In the above we used that 
\begin{equation}
\ps{\Psi_t}{ \left( x^\ell_{n,j} - q^\ell_j(t)\right)\left(\partial^2_{km}x^\ell_{n,j}\right)  \Psi_t} \in \mathbb{R} \ , \\
\end{equation}
since $\mathbf{x}_n$ is symmetric and commutes with $\partial^2_{km}x^\ell_{n,j}$ which is also symmetric in view of \eqref{eq : 1st and 2nd derivative xnjell}. Let us treat the term $\textnormal{I}_n$. We first define : 
\begin{equation}
	\begin{split}
		\textnormal{I} & := 4 \Ree \sum_{j=1}^N \ps{\Psi_t}{\left( x_{j} - q_j(t)\right) \left(-i\hbar \nabla_j - \hbar^{1/2}\boldsymbol{\hat{\mathbb{A}}}(x_{j}) - \left( p_j(t) -\boldsymbol{\mathbb{A}}_{\alpha_t}(q_j(t)) \right) \right)\Psi_t} \ . \\
	\end{split}
\end{equation}
In view of $\com{-i \hbar \nabla_j}{F_n(\mathbf{x})} = i\hbar n^{-2}F_n(\mathbf{x})^2x_{n,j}$, we write 
\begin{equation}
    \left(-i\hbar \nabla_j - \hbar^{1/2}\boldsymbol{\hat{\mathbb{A}}}(x_{j})\right)\left(F_n(\mathbf{x}) - 1 \right) = \left(F_n(\mathbf{x}) - 1 \right)\left(-i\hbar \nabla_j - \hbar^{1/2}\boldsymbol{\hat{\mathbb{A}}}(x_{j})\right) + i\hbar n^{-2}F_n(\mathbf{x})^2x_{n,j} \ . \\
\end{equation}
Together with $x_{n,j} - q_j(t) = (x_{n,j} - x_j) + (x_j - q_j(t))$, this gives
\begin{equation} \label{eq : I1 in beta a}
\begin{split}
	\left \lvert \textnormal{I}_n - \textnormal{I} \right \lvert & \lesssim \sum_{j=1}^N \norm{(x_{n,j} - x_j) \Psi_t} \norm{\left(\left(-i\hbar \nabla_j - \hbar^{1/2}\boldsymbol{\hat{\mathbb{A}}}(x_{j})\right)F_n(\mathbf{x}) - \left( p_j(t) - \boldsymbol{\mathbb{A}}_{\alpha_t}(q_j(t)) \right) \right)\Psi_t} \\
	& \quad \quad +  \sum_{j=1}^N \norm{\left(F_n(\mathbf{x}) - 1 \right)(x_{j} - q_j(t)) \Psi_t} \norm{\left(-i\hbar \nabla_j - \hbar^{1/2}\boldsymbol{\hat{\mathbb{A}}}(x_{j})\right)\Psi_t} \\
    & \quad \quad + \hbar n^{-2} \sum_{j=1}^N \norm{(x_{j} - q_j(t)) \Psi_t} \norm{F_n(\mathbf{x})^2x_{n,j}\Psi_t} \ .
\end{split}
\end{equation}
Then, using once again $\com{-i \hbar \nabla_j}{F_n(\mathbf{x})} = i\hbar n^{-2}F_n(\mathbf{x})^2x_{n,j}$ and $|F_n(\mathbf{x})| \leq 1$, we get
\begin{equation}
	\norm{\left(-i\hbar \nabla_j - \hbar^{1/2}\boldsymbol{\hat{\mathbb{A}}}(x_{j})\right)F_n(\mathbf{x})\Psi_t} \leq \norm{\left(-i\hbar \nabla_j - \hbar^{1/2}\boldsymbol{\hat{\mathbb{A}}}(x_{j})\right)\Psi_t} + \hbar n^{-2} \norm{x_{n,j} \Psi_t} \ .
\end{equation} 
The first term in the above is finite in view of \eqref{eq : estimates in lemma to show cvgce}, and the second goes to zero thanks to $\norm{\mathbf{x}_n}_{\textnormal{op}} \leq n$. Moreover, 
\begin{equation}
\norm{\left(\mathbf{x}_n - \mathbf{x} \right) \Psi_t} \underset{n \rightarrow + \infty}{\longrightarrow} 0
\end{equation} 
by dominated convergence, thanks to Lemma \ref{lemma : invariance of the domains for beta}. Hence, the first term in \eqref{eq : I1 in beta a} goes to zero, so does the second one, using that $F_n(\mathbf{x}) \longrightarrow 1$ almost everywhere and dominated convergence again. The third term is delt with in a similar manner, thanks to $|F_n(\mathbf{x})| \leq 1$ and $\norm{\mathbf{x}_n}_{\textnormal{op}} \leq n$. Let us continue with $\textnormal{II}_n$. We have :
\begin{equation} \label{eq : bound on II in lemma beta a}
\begin{split}
	|\textnormal{II}_n| & \lesssim n^{-4} \sum_{j,k=1}^N \sum_{\ell,m=1}^3  \left \lvert \left \langle \Psi_t, \left( x^\ell_{n,j} - q^\ell_j(t)\right) F_n(\mathbf{x})^5 \lvert x^m_k \lvert^2 x^\ell_j \Psi_t \right \rangle \right \lvert + n^{-2}\sum_{j=1}^N \sum_{\ell=1}^3 \left \lvert \left \langle \Psi_t, \left( x^\ell_{n,j} - q^\ell_j(t)\right) F_n(\mathbf{x})^3 x^\ell_j \Psi_t \right \rangle \right \lvert \\
	& \quad + n^{-2} \sum_{j,k=1}^N \sum_{\ell,m=1}^3 \left \lvert \left \langle \Psi_t, \left( x^\ell_{n,j} - q^\ell_j(t)\right) F_n(\mathbf{x})^3 x^m_k x^\ell_j \partial_{km} \Psi_t \right \rangle \right \lvert \\
	& \quad + n^{-2}  \sum_{j,k=1}^N \sum_{\ell,m=1}^3 \left \lvert \left \langle \Psi_t, \left( x^\ell_{n,j} - q^\ell_j(t)\right) \boldsymbol{\hat{\mathbb{A}}}^m(x_{k}) F_n(\mathbf{x})^3 x^m_k x^\ell_j  \Psi_t \right \rangle \right \lvert \ . \\
\end{split}
\end{equation}
Using \eqref{eq : cvgce beta a n to beta a} and $\|\mathbf{x}_n\|_{\textnormal{op}} \leq n$, one easily sees that the first and second term in the above go to $0$. Regarding the third term, note that the operator 
\begin{equation}
G_n(\mathbf{x}) : = \frac{n^{-2}|\mathbf{x}|^2}{(1+n^{-2}|\mathbf{x}|^2)^{3/2}} \\
\end{equation}
is a bounded operator, and it is actually bounded uniformly in $n$ (by $1$). Moreover, we have $G_n(\mathbf{x}) \longrightarrow 0$ almost everywhere. Hence, using \eqref{eq : inner product on H} we write :
\begin{equation} \label{eq : dominated convergence II beta a}
\begin{split}
	n^{-2} \left \lvert \left \langle \Psi_t, \left( x^\ell_{n,j} - q^\ell_j(t)\right) F_n(\mathbf{x})^3 x^m_k x^\ell_j \partial_{km} \Psi_t \right \rangle \right \lvert & \lesssim \sum_{m , \boldsymbol{\lambda}} \iint_{\mathbb{R}^{3N} \times \mathbb{R}^{3n}} (|\mathbf{x}| + |q(t)|) \lvert \Psi^{(m)}_t(\mathbf{x}, \mathbf{k}, \boldsymbol{\lambda}) \lvert G_n(\mathbf{x}) \lvert \partial_{km} \Psi^{(m)}_t(\mathbf{x}, \mathbf{k}, \boldsymbol{\lambda}) \lvert \textnormal{d} \mathbf{x} \textnormal{d}\mathbf{k} \\
	& \underset{n \rightarrow \infty}{\longrightarrow} 0 \\
\end{split}
\end{equation}
by dominated convergence. In the above we used that the right hand side can be uniformly bounded by a constant times
\begin{equation}
\left \lVert \mathbf{x} \Psi_t \right \lVert^2 + \left \lVert P \Psi_t \right \lVert^2 + \left \lVert \boldsymbol{\hat{\mathbb{A}}}(x) \Psi_t \right \lVert^2 \ , \\
\end{equation}
which is finite in view of Lemma \ref{lemma : Bounds on the quantized electromagnetic field} and the assumptions on $\Psi_t$. The fourth term in \eqref{eq : bound on II in lemma beta a} is handled by similar means as the previous one, up to replacing the momentum operator by the quantized electromagnetic field in \eqref{eq : dominated convergence II beta a}. Finally, using the fundamental theorem of calculus and dominated convergence in \eqref{eq : beta a n derivative 1} leads to the desired identity. \\
\end{proof}

The following lemma computes the change in time of the second functional $\boldsymbol{\beta}^b$. 

\begin{lemma} \label{lemma : change of beta b in time}
    Let $\sigma \in \left[1/2, 1 \right]$. Let $u_0 = (q_0,p_0,\alpha_0) \in X^\sigma$ and $\Psi_0 \in  \mathcal{D}(\mathbb{H}_{\hbar}^{1/2}) \cap \mathcal{D}(\mathcal{N}^{1/2}) \cap \mathcal{D}(\mathbf{x})$. Let $u(t) = (q(t),p(t),\alpha_t)$ and  $\Psi_t = e^{-i \frac{t}{\hbar}\mathbb{H}_{\hbar}}\Psi_0$ be the respective solutions of \eqref{eq : newton motion eqs PF} and \eqref{eq : sch eq} with initial data $u_0$ and $\Psi_0$. Then,
\begin{align} 
             \boldsymbol{\beta}^b(\Psi_{t}, u(t)) - \boldsymbol{\beta} & ^b(\Psi_{0}, u_0) = \notag \\
             & \ \int_0^t 4 \Ree \sum_{j=1}^N\sum_{\ell = 1}^3\sum_{m=1}^3 \Tilde{p}^m_j(s) \ps{\Psi_s}{\mathbb{G}^\ell_{js}  \left( F^{\ell m}_{\alpha_s}(q_j(s)) - \hbar^{1/2}\hat{F}^{\ell m}(x_j) \right)\Psi_s} \textnormal{d}s \label{eq : beta b term I} \\
             & \ +  \int_0^t 2\Ree \sum_{j=1}^N \sum_{k \neq j}\sum_{\ell = 1}^3 \ps{\Psi_s}{\mathbb{G}^\ell_{js}  \big( \left(\partial_{ \ell} V \right)(q_j(s) - q_k(s)) - \left(\partial_{\ell} V \right)(x_j - x_k) \big)\Psi_s} \textnormal{d}s \label{eq : beta b term II} \\
             & \ + \int_0^t 2\Ree \sum_{j=1}^N\sum_{\ell = 1}^3\ps{\Psi_s}{ \mathbb{G}^\ell_{js} \left(\hbar^{1/2}\boldsymbol{\hat{\mathbb{E}}}^\ell(x_j) - \boldsymbol{\mathbb{E}}_{\alpha_s}^\ell(q_j(s)) \right)\Psi_s} \textnormal{d}s \ , \label{eq : beta b term III}
\end{align}    
where we defined $\mathbb{G}^\ell_{jt} := P_j^\ell - \Tilde{p}^\ell_{j}(t)$ , with $P_j$ and $\Tilde{p}_j(t)$ being defined in \eqref{eq : def kinetic momentums N particles}. \\
\end{lemma}

\begin{proof}
    As we will explain below, the following computations actually only make sense on the domain $\mathcal{D} \left( \mathbb{H}_{\hbar} \right)$, whereas $\boldsymbol{\beta}^b$ is only defined on $\mathcal{D} \left( \mathbb{H}_{\hbar}^{1/2} \right)$. We will therefore use a density argument to legitimate the above formula. Recall that $\mathcal{D} \left( \mathbb{H}_{\hbar} \right)$ is dense in $\mathcal{D} \left( \mathbb{H}_{\hbar}^{1/2} \right)$ with respect to the graph norm of $\mathbb{H}_{\hbar}^{1/2}$, which is defined by
    \begin{equation}
        \norm{\Psi}^2_{\mathbb{H}_{\hbar}^{1/2}} := \norm{\Psi}^2 + \norm{\mathbb{H}_{\hbar}^{1/2}\Psi}^2 \ . \\
    \end{equation}
    Hence, consider a sequence $\left(\Psi^{(n)}_0\right)_{n \in \mathbb{N}} \subseteq \mathcal{D} \left( \mathbb{H}_{\hbar} \right)$ such that 
\begin{equation} \label{eq : density of D h in D h half graph norm}
    \norm{\Psi^{(n)}_0 - \Psi_0}_{\mathbb{H}^{1/2}_{N,\hbar}} \underset{ n \rightarrow \infty}{\longrightarrow} 0 \ , \\
\end{equation}
and denote $\Psi^{(n)}_t = e^{-i\frac{t}{\hbar}\mathbb{H}_{\hbar}}\Psi^{(n)}_0$. Note that by unitarity of $e^{-i\frac{t}{\hbar}\mathbb{H}_{\hbar}}$, we have for all $t \geq 0$, 
\begin{equation} \label{eq : density of D h in D h half graph norm}
    \norm{\Psi^{(n)}_t - \Psi_t}_{\mathbb{H}^{1/2}_{N,\hbar}} \underset{ n \rightarrow + \infty}{\longrightarrow} 0 \ . \\
\end{equation}
We then compute :
 \begin{equation} \label{eq : derivative beta b}
     \dt \boldsymbol{\beta}^b \left(\Psi^{(n)}_t, u(t)\right) = 2 \Ree \sum_{j=1}^N \ps{\Psi^{(n)}_t}{ \sum_{\ell=1}^3\mathbb{G}^\ell_{jt} \left(i\hbar^{-1} \com{\mathbb{H}_{\hbar}}{\mathbb{G}^\ell_{jt}} + \dot{\mathbb{G}}^\ell_{j,t}\right)\Psi^{(n)}_t} \ . \\
 \end{equation}
 Let us start by computing the commutator $\com{\mathbb{H}_{\hbar}}{\mathbb{G}^\ell_{jt}}$. We use the following shorthand notations for the classical and quantum Faraday tensors :
 \begin{equation}
 	\begin{split}
 	F^{j \ell m}_t & = F^{\ell m}_{\alpha_t}(q_j(t)) = \partial_m \boldsymbol{\mathbb{A}}_{\alpha_t}^\ell(q_j(t)) - \partial_\ell \boldsymbol{\mathbb{A}}_{\alpha_t}^m(q_j(t)) \\
 	\hat{F}^{j \ell m} & = \hat{F}^{\ell m}(x_j) = \partial_m \boldsymbol{\hat{\mathbb{A}}}^\ell(x_j) - \partial_\ell \boldsymbol{\hat{\mathbb{A}}}^m(x_j) \ . \\
 	\end{split}
 \end{equation}
 Using the definitions of $\mathbb{H}_{\hbar}$ and $\mathbb{G}_{j,t}$ and the identity
  \begin{equation}
     \com{H_\textnormal{f}}{\boldsymbol{\hat{\mathbb{A}}}(x_j)} = i \boldsymbol{\hat{\mathbb{E}}}(x_j) \ , \\
 \end{equation}
 we have immediately 
 \begin{equation}
     \begin{split}
         \com{\mathbb{H}_{\hbar}}{\mathbb{G}^\ell_{jt}} & = \sum_{k=1}^N \sum_{m=1}^3 \com{\left(P_k^m \right)^2}{P^\ell_j} + \sum_{1 \leq k < r \leq N} \com{V(x_k-x_r)}{P^\ell_j} + \hbar \com{H_{\textnormal{f}}}{P^\ell_j} \\
         & = \sum_{k=1}^N \sum_{m=1}^3 \left( P^m_k \com{P_k^m }{P^\ell_j} + \com{P_k^m }{P^\ell_j} P^m_k \right) + i\hbar \sum_{k \neq j} (\partial_{ \ell} V)(x_j - x_k) - i \hbar^{3/2} \boldsymbol{\hat{\mathbb{E}}}^\ell(x_j) \\
         & = \sum_{k=1}^N \sum_{m=1}^3 \left( 2P^m_k \com{P_k^m }{P^\ell_j} + \com{\com{P_k^m }{P^\ell_j}}{P^m_k} \right) + i\hbar \sum_{k \neq j} (\partial_{\ell} V)(x_j - x_k) - i \hbar^{3/2} \boldsymbol{\hat{\mathbb{E}}}^\ell(x_j) \ .
     \end{split}
 \end{equation}
 Now, a computation using the definition of $P_j$ gives 
 \begin{equation} \label{eq : commutator P m k P l j}
     \begin{split}
         \com{P_k^m }{P^\ell_j} & = i \hbar^{3/2} \left( \com{\partial_{k m}}{\boldsymbol{\hat{\mathbb{A}}}^\ell(x_j)} - \com{\partial_{j \ell}}{\boldsymbol{\hat{\mathbb{A}}}^m(x_k)} \right) + \hbar \com{\boldsymbol{\hat{\mathbb{A}}}^m(x_k)}{\boldsymbol{\hat{\mathbb{A}}}^\ell(x_j)} \\
     \end{split}
 \end{equation}
Now, a computation using the CCR gives
\begin{equation}
    \begin{split}
        \com{\boldsymbol{\hat{\mathbb{A}}}^m(x_k)}{\boldsymbol{\hat{\mathbb{A}}}^\ell(x_j)} & = 2i \Imm \ps{\mathbf{G}^m_{x_k}}{\mathbf{G}^\ell_{x_j}} \\
        & = i\Imm \sum_{\lambda = 1,2} \int_{\mathbb{R}^3} \frac{|\mathcal{F}[\kappa](p)|^2}{|p|} \boldsymbol{\epsilon}^m_\lambda(p) \boldsymbol{\epsilon}^\ell_\lambda(p) e^{ip\cdot (x_k - x_j)} \textnormal{d}p \\
        & = i\int_{\mathbb{R}^3} \frac{|\mathcal{F}[\kappa](p)|^2}{|p|} \left( \sum_{\lambda = 1,2}\boldsymbol{\epsilon}^m_\lambda(p) \boldsymbol{\epsilon}^\ell_\lambda(p) \right)\sin \left( p\cdot (x_k - x_j) \right) \textnormal{d}p \\
        & = i\int_{\mathbb{R}^3} \frac{|\mathcal{F}[\kappa](p)|^2}{|p|} \left( \delta_{m \ell} - \frac{p_m p_\ell}{|p|^2} \right)\sin \left( p\cdot (x_k - x_j) \right) \textnormal{d}p \\
        & = 0 \ . \\
    \end{split}
\end{equation}
Here in the third line, we used the \textit{completeness relation}
\begin{equation} \label{eq : completeness relation}
    \sum_{\lambda = 1,2}\boldsymbol{\epsilon}^m_\lambda(p) \boldsymbol{\epsilon}^\ell_\lambda(p) = \delta_{m \ell} - \frac{p_m p_\ell}{|p|^2} \ , \\
\end{equation}
which can be shown to hold by using the fact that for all $k \neq 0$,
\begin{equation}
    \left\{ \boldsymbol{\epsilon}_\lambda(k), \frac{k}{|k|} \right\}_{\lambda = 1,2} \\
\end{equation}
is an O.N.B of $\mathbb{R}^3$, and by testing the LHS and the RHS of the identity against all three of the elements of this O.N.B. We thus obtain : 
\begin{equation} \label{eq : identity for com Pm Pell}
\com{P_k^m }{P^\ell_j} = i \hbar^{3/2} \delta_{jk} \hat{F}^{j  \ell m} .
\end{equation}
Consequently, gathering the above expressions together, and using the definition of $\hat{F}^{j  \ell m}$, we get
\begin{equation} \label{eq : com H G in comp beta b}
    \begin{split}
        i \hbar^{-1} \com{\mathbb{H}_{\hbar}}{\mathbb{G}^\ell_{jt}} & = -2 \hbar^{1/2} \sum_{m=1}^3P^m_j \hat{F}^{j  \ell m}  - \sum_{k \neq j} (\partial_{\ell} V)(x_j - x_k) + \hbar^{1/2} \boldsymbol{\hat{\mathbb{E}}}^\ell(x_j) + i\hbar^{-1} \sum_{m=1}^3  \com{\com{P_j^m }{P^\ell_j}}{P^m_j} \ . \\
    \end{split}
\end{equation}
Now, in view of \eqref{eq : expression for derivative of p tilda}, we write $\dot{\mathbb{G}}^\ell_{j,t} = -\partial_t \Tilde{p}_j^\ell(t)$ as 
\begin{equation} \label{eq : derivative G in comp beta b}
\dot{\mathbb{G}}^\ell_{j,t} = 2 \sum_{m = 1}^3 \Tilde{p}^m_j(t) F^{j \ell m}_t + \sum_{k \neq j} \left( \partial_{\ell} V \right) \left(q_j(t) - q_k(t) \right)  -\boldsymbol{\mathbb{E}}^\ell_{\alpha_t}(q_j(t))  \\
\end{equation}
Note here that comparing \eqref{eq : com H G in comp beta b} and \eqref{eq : derivative G in comp beta b} is reminiscent of the heuristic discussion in Subsection \ref{subsec : Strategy of the proof}. Now, adding the first terms of the RHS of \eqref{eq : com H G in comp beta b} and \eqref{eq : derivative G in comp beta b}, summing over $\ell$ against $\mathbb{G}^\ell_{jt}$ gives :
\begin{equation}
    \begin{split}
        \sum_{\ell=1}^3\sum_{m=1}^3\mathbb{G}^\ell_{jt} \left( \Tilde{p}^m_j(t) F^{j \ell m}_t - P^m_j \hbar^{1/2}\hat{F}^{j  \ell m} \right) & = \sum_{\ell=1}^3\sum_{m=1}^3 P^\ell_j \Tilde{p}^m_j(t) F^{j \ell m}_t - \sum_{\ell=1}^3\sum_{m=1}^3 \Tilde{p}^\ell_j(t) \Tilde{p}^m_j(t) F^{j \ell m}_t \\
        & \quad - \sum_{\ell=1}^3\sum_{m=1}^3 P^\ell_j P^m_j \hbar^{1/2}\hat{F}^{j  \ell m} + \sum_{\ell=1}^3\sum_{m=1}^3 \Tilde{p}^\ell_j(t) P^m_j \hbar^{1/2}\hat{F}^{j  \ell m} \ , \\
    \end{split}
\end{equation}
up to a factor of $2$. The second term on the right-hand side  is $0$ because of \eqref{eq : antisymmetry cancels the term}, the third term doesn't cancel (because $P_j^\ell$ and $P_j^m$ do not commute), and we will estimate its contribution to \eqref{eq : derivative beta b} below. Regarding the first and the fourth terms, the antisymmetry of $\hat{F}^{j  \ell m}$ gives :
\begin{equation}
    \begin{split}
        \sum_{\ell=1}^3\sum_{m=1}^3 P^\ell_j \Tilde{p}^m_j(t) F^{j \ell m}_t + \sum_{\ell=1}^3\sum_{m=1}^3 \Tilde{p}^\ell_j(t) P^m_j \hbar^{1/2}\hat{F}^{j  \ell m} = \sum_{\ell=1}^3\sum_{m=1}^3 P^\ell_j \Tilde{p}^m_j(t) D^{j \ell m }_t
    \end{split}
\end{equation}
where we defined 
\begin{equation}
    D^{j \ell m }_t := F^{j \ell m}_t - \hbar^{1/2}\hat{F}^{j  \ell m} \ . \\
\end{equation}
Note that $D^{j m \ell }_t = -D^{j \ell m }_t$, hence 
\begin{equation}
    \sum_{\ell=1}^3\sum_{m=1}^3 \Tilde{p}^\ell_j(t) \Tilde{p}^m_j(t) D^{j \ell m }_t = 0 \\
\end{equation}
and finally
\begin{equation}
    \sum_{\ell=1}^3\sum_{m=1}^3 P^\ell_j \Tilde{p}^m_j(t) D^{j \ell m }_t = \sum_{\ell=1}^3\sum_{m=1}^3 \mathbb{G}^\ell_{jt} \Tilde{p}^m_j(t) D^{j \ell m }_t \ . \\
\end{equation}
Consequently, we arrive at 
 \begin{equation} \label{eq : dt of beta b sum of terms}
     \dt \boldsymbol{\beta}^b \left(\Psi^{(n)}_t, u(t)\right) = \sum_{\alpha = 1}^5 \mathcal{E}_{n,t}^{(b,\alpha)} \ , \\
 \end{equation}
 where 
 \begin{equation} \label{eq : E terms for beta b dt}
     \begin{split}
         \mathcal{E}_{n,t}^{(b,1)} & = 4 \Ree \sum_{j=1}^N\sum_{\ell,m = 1}^3 \ps{\Psi^{(n)}_t}{\mathbb{G}^\ell_{jt} \Tilde{p}^m_j(t) \left( F^{j \ell m}_t - \hbar^{1/2}\hat{F}^{j  \ell m} \right)\Psi^{(n)}_t} \ , \\
         \mathcal{E}_{n,t}^{(b,2)} & = -4 \Ree \sum_{j=1}^N \sum_{\ell,m=1}^3 \ps{\Psi^{(n)}_t}{P_j^\ell P^m_j \hbar^{1/2}\hat{F}^{j  \ell m}\Psi^{(n)}_t} \ , \\
         \mathcal{E}_{n,t}^{(b,3)} & = 2\Ree \  i\hbar^{-1} \sum_{j=1}^N  \sum_{\ell,m=1}^3  \ps{\Psi^{(n)}_t}{\mathbb{G}^\ell_{jt} \com{\com{P_j^m }{P^\ell_j}}{P^m_j}\Psi^{(n)}_t} \ , \\
         \mathcal{E}_{n,t}^{(b,4)} & = 2\Ree \sum_{\substack{j,k = 1 \\k\neq j}}^N \sum_{\ell = 1}^3 \ps{\Psi^{(n)}_t}{\mathbb{G}^\ell_{jt}  \left( \left(\partial_{\ell} V \right)(q_j(t) - q_k(t)) - \left(\partial_{\ell} V \right)(x_j - x_k) \right)\Psi^{(n)}_t} \ , \\
         \mathcal{E}_{n,t}^{(b,5)} & = 2\Ree \sum_{j=1}^N\sum_{\ell = 1}^3\ps{\Psi^{(n)}_t}{ \mathbb{G}^\ell_{jt} \left(\hbar^{1/2} \boldsymbol{\hat{\mathbb{E}}}^\ell(x_j) - \boldsymbol{\mathbb{E}}_{\alpha_t}^\ell(q_j(t)) \right)\Psi^{(n)}_t} \ . \\
     \end{split}
 \end{equation}
Let us quickly discuss why the above terms are well defined on $\mathcal{D}\left( \mathbb{H}_{\hbar} \right)$. We only focus on $\mathcal{E}_{n,t}^{(b,2)}$ and $\mathcal{E}_{n,t}^{(b,3)}$, as the other ones are clearly well-defined, even by requiring $\Psi^{(n)}_t \in \mathcal{D}\left( (-\Delta)^{1/2}\right) \cap \mathcal{D}\left( H_{\textnormal{f}}^{1/2}\right) = \mathcal{D}\left( \mathbb{H}_{\hbar}^{1/2} \right)$, thanks to the assumptions on $\kappa$. The calculations in \eqref{eq : commutators P m P ell} below show that we need 
\begin{equation}
    \norm{P^mP^\ell \Psi} < +\infty \ , \\
\end{equation}
for $\mathcal{E}_{n,t}^{(b,3)}$ to be finite (here we dropped the $j$, $t$ indices as well as the $(n)$ exponent for convenience) which, if satisfied, also guarantees that $\mathcal{E}_{n,t}^{(b,2)}$ is well-defined. Now, a straightforward computation yields
\begin{equation}
    \begin{split}
        \norm{P^mP^\ell \Psi} & \lesssim_\hbar \norm{\partial_m \partial_\ell \Psi} + \norm{\com{\partial_m}{\boldsymbol{\hat{\mathbb{A}}}^\ell(x)}\Psi} + \norm{\boldsymbol{\hat{\mathbb{A}}}^\ell(x)\partial_m\Psi} + \norm{\boldsymbol{\hat{\mathbb{A}}}^m(x)\partial_\ell\Psi} + \norm{\boldsymbol{\hat{\mathbb{A}}}^m(x)\boldsymbol{\hat{\mathbb{A}}}^\ell(x)\Psi} \\
    & \lesssim_{N,\hbar, \kappa} \norm{\left(-\Delta \right)\Psi} + \norm{H_{\textnormal{f}}^{1/2}\Psi} + \norm{H_{\textnormal{f}}^{1/2}(-\Delta)^{1/2}\Psi} + \norm{H_{\textnormal{f}}\Psi} \ , \\
    \end{split}
\end{equation}
which is finite if $\Psi \in \mathcal{D}\left( -\Delta\right) \cap \mathcal{D}\left( H_{\textnormal{f}}\right) = \mathcal{D} \left( \mathbb{H}_{\hbar} \right)$. In fact, it turns out that we have
    \begin{equation} \label{eq : E 2 3 are zero}
        \begin{split}
            \mathcal{E}_{n,t}^{(b,2)} & = 0 \\
            \mathcal{E}_{n,t}^{(b,3)} & = 0  \\
        \end{split}
    \end{equation}
when $\Psi_t \in \mathcal{D}\left( \mathbb{H}_{\hbar} \right)$. We now prove this claim. \\

\noindent \textbf{Term $\mathcal{E}_{n,t}^{(b,2)}$}. Using the antisymmetry of the quantum Faraday tensor,  permuting indices in the sums, and applying \eqref{eq : identity for com Pm Pell} we get
\begin{equation}
    \begin{split}
    \sum_{\ell=1}^3 \sum_{m=1}^3 P_j^\ell P^m_j \hat{F}^{j  \ell m} & =  \sum_{\ell=1}^3 \sum_{m=1}^3 P_j^m P^\ell_j \hat{F}^{j  \ell m} + \sum_{\ell=1}^3 \sum_{m=1}^3\com{P_j^\ell}{P^m_j } \hat{F}^{j  \ell m} \\
    & = -  \sum_{\ell=1}^3 \sum_{m=1}^3 P_j^\ell P^m_j \hat{F}^{j  \ell m} - i \hbar^{3/2} \sum_{m=1}^3 \left(\hat{F}^{j  \ell m}\right)^2 \ .
    \end{split}
\end{equation}
This implies :
\begin{equation}
     \sum_{\ell=1}^3 \sum_{m=1}^3 P_j^\ell P^m_j \hat{F}^{j  \ell m} = - \frac{i \hbar^{3/2}}{2} \sum_{m=1}^3 \left(\hat{F}^{j  \ell m}\right)^2 \ , \\
\end{equation}
which leads to :
\begin{equation}
    \begin{split}
        \mathcal{E}_{n,t}^{(b,2)} & = 2 \hbar^2 \Ree \ i\sum_{j=1}^N\sum_{\ell=1}^3 \sum_{m=1}^3 \ps{\Psi^{(n)}_t}{\left( \hat{F}^{j  \ell m}\right)^2 \Psi^{(n)}_t} \\
        & = 0 \ , \\
    \end{split}
\end{equation}
since $\hat{F}^{j  \ell m}$ is self-adjoint and $\Psi_t^{(n)} \in \mathcal{D}\left(\left( \hat{F}^{j  \ell m}\right)^2 \right)$ thanks to \eqref{eq : estimates in lemma to show cvgce}. \\

\noindent \textbf{Term $\mathcal{E}_{n,t}^{(b,3)}$}. We split :
\begin{equation}
    \ps{\Psi^{(n)}_t}{\mathbb{G}^\ell_{jt} \com{\com{P_j^m }{P^\ell_j}}{P^m_j}\Psi^{(n)}_t} = \ps{\Psi^{(n)}_t}{P^\ell_j \com{\com{P_j^m }{P^\ell_j}}{P^m_j}\Psi^{(n)}_t} - \Tilde{p}^\ell_j(t) \ps{\Psi^{(n)}_t}{ \com{\com{P_j^m }{P^\ell_j}}{P^m_j}\Psi^{(n)}_t} \ . \\
\end{equation}
Since the commutator of two symmetric operators is skew-symmetric, it holds that $\com{\com{A}{B}}{C}$ is a symmetric operator whenever $A,B,C$ are symmetric. Hence, the second term on the right-hand side is real. We thus have :
\begin{equation}
    \mathcal{E}_{n,t}^{(b,3)}  = 2\Ree \  i\hbar^{-1} \sum_{j=1}^N \sum_{\ell=1}^3 \sum_{m=1}^3  \ps{\Psi^{(n)}_t}{P^\ell_j \com{\com{P_j^m }{P^\ell_j}}{P^m_j}\Psi^{(n)}_t} \ . \\
\end{equation}
Let us drop the subscripts $j$, $t$ and the exponent ${(n)}$ for clarity, as it won't impact the calculations. A direct calculation gives
\begin{equation} \label{eq : commutators P m P ell}
    \begin{split}
        \ps{\Psi}{P^\ell \com{\com{P^m}{P^\ell}}{P^m}\Psi} & = 2 \ps{\Psi}{P^\ell P^m P^\ell P^m \Psi} - \norm{P^m P^\ell \Psi}^2 - \ps{\Psi}{P^\ell P^\ell P^m  P^m \Psi} \ . \\
    \end{split}
\end{equation}
Then, straightforward manipulations with the above identity and permuting indices in the sums lead to
\begin{equation}
    \sum_{\ell=1}^3 \sum_{m=1}^3  \ps{\Psi}{P^\ell \com{\com{P^m }{P^\ell}}{P^m}\Psi} \in \mathbb{R} \ . \\
\end{equation}
And thus 
\begin{equation}
    \mathcal{E}_{n,t}^{(b,3)} = 0 \ . \\
\end{equation}

Thanks to \eqref{eq : E 2 3 are zero} we can simplify \eqref{eq : dt of beta b sum of terms}, and an application of the fundamental theorem of calculus gives
\begin{equation} \label{eq : beta b FTC but with the n}
    \boldsymbol{\beta}^b \left(\Psi^{(n)}_t, u(t)\right) - \boldsymbol{\beta}^b \left(\Psi^{(n)}_0, u_0\right) = \int_0^t \left( \mathcal{E}^{(b,1)}_{n,s} + \mathcal{E}^{(b,4)}_{n,s} + \mathcal{E}^{(b,5)}_{n,s} \right) \textnormal{d}s \ . \\
\end{equation}
Note that the integrand terms in the RHS are now actually defined on $\mathcal{D}\left(\mathbb{H}_{\hbar}^{1/2} \right)$, which means that we can pass to the limit in the above expression. Let us start with the LHS. Note that we can write 
\begin{equation}
     \boldsymbol{\beta}^b \left(\Psi^{(n)}_t, u(t)\right) = \sum_{j=1}^N \sum_{\ell=1}^3 \ps{\Psi^{(n)}_t}{\left( \mathbb{G}^\ell_{jt} \right)^2 \Psi^{(n)}_t} \ . \\
\end{equation}
In what follows we will use the fact that for every self-adjoint $A$, every (not necessarily self-adjoint) $B$ and every $\psi, \phi \in \mathcal{D}(A) \cap \mathcal{D}(B)$, there holds
    \begin{equation} \label{eq : inequality to show cvgce beta b n beta b}
        \left \lvert \ps{\psi}{AB\psi} - \ps{\phi}{AB\phi} \right \lvert \leq \norm{A(\psi-\phi)} \norm{B\psi}  + \norm{A\phi} \norm{B(\psi - \phi)} \ . \\
    \end{equation}
This allows us to control :
\begin{equation} \label{eq : beta b n goes to beta b}
    \begin{split}
        \left \lvert \boldsymbol{\beta}^b \left(\Psi^{(n)}_t, u(t)\right) - \boldsymbol{\beta}^b \left(\Psi_t, u(t)\right) \right \lvert & \leq \sum_{j=1}^N \sum_{\ell = 1}^3  \norm{\mathbb{G}^\ell_{jt} \left( \Psi^{(n)}_t - \Psi_t \right)} \left(\norm{\mathbb{G}^\ell_{jt}\Psi^{(n)}_t} + \norm{\mathbb{G}^\ell_{jt}\Psi^{}_t}  \right) \\
        & \lesssim_{N, \hbar, \kappa} \left( \norm{\left(\mathbb{H}_{\hbar} + 1 \right)^{1/2}\left( \Psi^{(n)}_t - \Psi_t \right) } + \norm{ \Psi^{(n)}_t - \Psi_t}  \right) \\
        & \quad \quad \quad \quad \quad \times  \left( \norm{\left(\mathbb{H}_{\hbar} + 1 \right)^{1/2} \Psi^{(n)}_t  } + \norm{ \Psi^{(n)}_t} + \norm{\left(\mathbb{H}_{\hbar} + 1 \right)^{1/2} \Psi_t  } + \norm{ \Psi_t}  \right) \\
        & \underset{n \rightarrow + \infty}{\longrightarrow} 0 \ . \\
    \end{split}
\end{equation}
Here in the second line we used \eqref{eq : estimates in lemma to show cvgce} together with the definition of $\mathbb{G}^\ell_{jt}$ and Lemma \ref{lemma : A priori estimates for the Newton--Maxwell solutions}. Let us then treat the RHS of \eqref{eq : beta b FTC but with the n}. We want to show that   
\begin{equation}
    \int_0^t \left \lvert \mathcal{E}^{(b,\alpha)}_{n,s} - \mathcal{E}^{(b,\alpha)}_s \right \lvert \textnormal{d}s \underset{n \rightarrow + \infty}{\longrightarrow} 0 \\
\end{equation}
for $\alpha = 1,4,5$, where $\mathcal{E}^{(b,\alpha)}_s$ are the terms obtained when replacing $\Psi_t^{(n)}$ by $\Psi_t$ in $\mathcal{E}^{(b,\alpha)}_{n,s}$. We will actually show that there holds
\begin{equation}
    \sup_{s \in [0,t]}  \left \lvert \mathcal{E}^{(b,\alpha)}_{n,s} - \mathcal{E}^{(b,\alpha)}_s \right \lvert \underset{n \rightarrow + \infty}{\longrightarrow} 0 \ . \\
\end{equation}
Let us treat each term separately. \\

\noindent \textbf{Term $\int_0^t \left \lvert \mathcal{E}^{(b,1)}_{n,s} - \mathcal{E}^{(b,1)}_s \right \lvert \textnormal{d}s$}. Using \eqref{eq : inequality to show cvgce beta b n beta b} and Lemma \ref{lemma : A priori estimates for the Newton--Maxwell solutions}, we write :
\begin{equation}
\begin{split}
\left \lvert \mathcal{E}^{(b,1)}_{n,s} - \mathcal{E}^{(b,1)}_s \right \lvert  & \lesssim_{N,\kappa} \sum_{j=1}^N \sum_{\ell, m=1}^3 \Big( \norm{\mathbb{G}^\ell_{js} \left( \Psi^{(n)}_s - \Psi_s \right)} \norm{\left( F^{j \ell m}_s - \hbar^{1/2}\hat{F}^{j  \ell m} \right)\Psi^{(n)}_s} \\
& \quad \quad \quad \quad \quad \quad \quad \quad \quad + \norm{\mathbb{G}^\ell_{js}\Psi_s } \norm{\left( F^{j \ell m}_s - \hbar^{1/2}\hat{F}^{j  \ell m} \right) \left( \Psi^{(n)}_s - \Psi_s \right)} \Big) \\
& \underset{n \rightarrow + \infty}{\longrightarrow} 0 \ , \\
\end{split}
\end{equation}
uniformly in $s \in [0,t]$. Here we used similar arguments as in \eqref{eq : beta b n goes to beta b}, with the additional use of Lemma \ref{lemma : estimates on the classical elec potential} and Lemma \ref{lemma : Bounds on the quantized electromagnetic field} to control the classical and quantum Faraday tensors. The uniformity in $s$ comes from Lemma \ref{lemma : A priori estimates for the Newton--Maxwell solutions} and \eqref{eq : density of D h in D h half graph norm}. \\

\noindent \textbf{Term $\int_0^t \left \lvert \mathcal{E}^{(b,4)}_{n,s} - \mathcal{E}^{(b,4)}_s \right \lvert \textnormal{d}s$}. A similar calculation as for the first term yields
\begin{equation}
    \begin{split}
         \left \lvert \mathcal{E}^{(b,4)}_{n,s} - \mathcal{E}^{(b,4)}_s \right \lvert & \lesssim \sum_{\substack{j,k = 1 \\k\neq j}}^N \sum_{\ell=1}^3 \Big( \norm{\mathbb{G}^\ell_{js} \left( \Psi^{(n)}_s - \Psi_s \right)} \norm{ \left( \left(\partial_{\ell} V \right)(q_j(t) - q_k(t)) - \left(\partial_{\ell} V \right)(x_j - x_k) \right)\Psi^{(n)}_s} \\
& \quad \quad \quad \quad \quad \quad \quad \quad \quad + \norm{\mathbb{G}^\ell_{js}\Psi_s } \norm{ \left( \left(\partial_{\ell} V \right)(q_j(t) - q_k(t)) - \left(\partial_{\ell} V \right)(x_j - x_k) \right) \left( \Psi^{(n)}_s - \Psi_s \right))} \Big) \\
& \underset{n \rightarrow + \infty}{\longrightarrow} 0 \ , \\
    \end{split}
\end{equation}
uniformly in $s \in [0,t]$. We used in addition \eqref{eq : pointwise bounds for V and nabla V} to control the interaction terms. \\

\noindent \textbf{Term $\int_0^t \left \lvert \mathcal{E}^{(b,5)}_{n,s} - \mathcal{E}^{(b,5)}_s \right \lvert \textnormal{d}s$}. The reasoning is similar to that of $\alpha = 1$, up to replacing the $\hat{\mathbf{A}}$ field by the $\hat{\mathbf{E}}$ field. It allows us to obtain 
\begin{equation}
    \sup_{s \in [0,t]}  \left \lvert \mathcal{E}^{(b,5)}_{n,s} - \mathcal{E}^{(b,5)}_s \right \lvert \underset{n \rightarrow + \infty}{\longrightarrow} 0 \ , \\
\end{equation}
and then we can pass to the limit in \eqref{eq : beta b FTC but with the n}, and recover the desired identity for $\boldsymbol{\beta}^b(\Psi_t, u(t))$. This concludes the proof. \\

\end{proof}

Finally, the following lemma computes the change in time of the third functional $\boldsymbol{\beta}^c$. \\

\begin{lemma} \label{lemma : growth of beta c}
    Let $\sigma \in \left[1/2, 1 \right]$. Let $u_0 = (q_0,p_0,\alpha_0) \in X^\sigma$ and $\Psi_0 \in  \mathcal{D}(\mathbb{H}_{\hbar}^{1/2}) \cap \mathcal{D}(\mathcal{N}^{1/2}) \cap \mathcal{D}(\mathbf{x})$. Let $u(t) = (q(t),p(t),\alpha_t)$ and  $\Psi_t = e^{-i \frac{t}{\hbar}\mathbb{H}_{\hbar}}\Psi_0$ be the respective solutions of \eqref{eq : newton motion eqs PF} and \eqref{eq : sch eq} with initial data $u_0$ and $\Psi_0$. Then,
    \begin{align} 
             & \boldsymbol{\beta}^c(\Psi_t, \alpha_t) - \boldsymbol{\beta} ^c(\Psi_0,\alpha_0) = \notag \\
             & \quad \quad \int_0^t 4 \Imm \sum_{j=1}^N \sum_{\lambda=1,2} \int_{\mathbb{R}^3} \Tilde{p}_j(s) \left \langle \Psi_s,\left(\hbar^{1/2}a_{k,\lambda} -  \alpha_s(k,\lambda) \right)  \left(\overline{\mathbf{G}_{q_j(s)}(k,\lambda) - \mathbf{G}_{x_j}(k,\lambda)} \right) \Psi_s \right \rangle \textnormal{d}k \textnormal{d}s \label{eq : beta c term I} \\
             &\quad \quad \quad \quad + \int_0^t 4 \Imm \sum_{j=1}^N \sum_{\lambda=1,2} \int_{\mathbb{R}^3} \ps{\Psi_s}{\mathbb{G}_{j,s} \overline{\mathbf{G}_{x_j}(k,\lambda)} \left(\hbar^{1/2}a_{k,\lambda} - \alpha_t(k,\lambda) \right) \Psi_s} \textnormal{d}k \textnormal{d}s \label{eq : beta c term II} \ . 
    \end{align}
\end{lemma}

The proof of the above lemma is done in a very similar fashion as in \cite[Lemma 3.3]{falconi2023derivation}. We postpone it to Appendix \ref{appendix : Growth of the third functional}. \\

Now, in view of the previous lemmas, we are able to control the growth of the functional $\boldsymbol{\beta}$, which leads to the 
\begin{proof}[Proof of Lemma \ref{lemma : Growth of the functional}]
    We can control $ \boldsymbol{\beta}^a(\Psi_t, q(t))$ using \eqref{eq : change of beta a in time} together with the Young and Cauchy-Schwarz inequalities :
    \begin{equation} \label{eq : final bound on beta a in proof}
        \begin{split}
            \boldsymbol{\beta}^a(\Psi_t, q(t)) & \leq \boldsymbol{\beta}^a(\Psi_0, q_0) + C\int_0^t \left(\boldsymbol{\beta}^a(\Psi_s, q(s)) + \boldsymbol{\beta}^b(\Psi_s, u(s)) \right) \textnormal{d}s \\
        \end{split}
    \end{equation}
    Now, let us treat $ \boldsymbol{\beta}^b(\Psi_t, u(t))$, which is explicitely written in Lemma \ref{lemma : change of beta b in time}. Let us denote $\textnormal{I}_b = \eqref{eq : beta b term I}$, $\textnormal{II}_b = \eqref{eq : beta b term II}$, $\textnormal{III}_b = \eqref{eq : beta b term III}$ and let us estimate each term separately. \\

\noindent \textbf{Term $\textnormal{I}_b$}. Note first that from \eqref{eq : energy estimate on tilde p in proof lemma}, one has the uniform bound 
\begin{equation}
\left \lvert \tilde{p}(t) \right \lvert \lesssim_\kappa 1 \ . \\
\end{equation}
This and the Cauchy-Schwarz inequality yield
\begin{equation} \label{eq : first bound on term I 1 beta b}
    \left \lvert \textnormal{I}_b \right \lvert \lesssim \int_0^t \sum_{j=1}^N \sum_{\ell = 1}^3 \sum_{m=1}^3 \norm{\mathbb{G}^\ell_{js} \Psi_s} \norm{\left( F^{j \ell m}_s - \hbar^{1/2}\hat{F}^{j  \ell m} \right) \Psi_s} \textnormal{d}s \ . \\
\end{equation}
Next, we write for $s \in [0,t]$, 
\begin{equation} \label{eq : developed expressions for derivatives of A fields}
    \begin{split}
        \partial_{jm}\boldsymbol{\mathbb{A}}^\ell_{\alpha_s} (q_j(s)) &= i \sum_{\lambda = 1,2} \int_{\mathbb{R}^3} \frac{\mathcal{F}[\kappa](k)}{\sqrt{2|k|}} \boldsymbol{\epsilon}^\ell_\lambda(k) k_m \left( e^{ik \cdot q_j(s)} \alpha_s(k,\lambda) -  e^{-ik \cdot q_j(s)} \overline{\alpha_s(k,\lambda)} \right)\textnormal{d}k \\
    \partial_{jm}\boldsymbol{\hat{\mathbb{A}}}^\ell(x_j) &= i \sum_{\lambda=1,2} \int_{\mathbb{R}^3} \frac{\mathcal{F}[\kappa](k)}{\sqrt{2|k|}}\boldsymbol{\epsilon}^\ell_\lambda(k) k_m \left( e^{ik\cdot x_j}a_{k,\lambda} - e^{-ik\cdot x_j}a^*_{k,\lambda}\right) \textnormal{d}k \ , \\
    \end{split}
\end{equation}
which leads to :
\begin{align} 
         F^{j \ell m}_s - \hbar^{1/2}\hat{F}^{j  \ell m} & = \partial_{jm} \boldsymbol{\mathbb{A}}^\ell_{\alpha_s} (q_j(s)) - \partial_{jm}\boldsymbol{\mathbb{A}}^\ell_{\alpha_s} (x_j) - \left( \partial_{j\ell} \boldsymbol{\mathbb{A}}^m_{\alpha_s} (q_j(s)) -  \partial_{j\ell} \boldsymbol{\mathbb{A}}^m_{\alpha_s} (x_j) \right) \notag \\
         & \quad  + \partial_{jm} \boldsymbol{\mathbb{A}}^\ell_{\alpha_s} (x_j) - \hbar^{1/2} \partial_{jm} \boldsymbol{\hat{\mathbb{A}}}^\ell(x_j) - \left(  \partial_{j\ell} \boldsymbol{\mathbb{A}}^m_{\alpha_s}(x_j) - \hbar^{1/2} \partial_{j\ell}\boldsymbol{\hat{\mathbb{A}}}^m(x_j) \right) \notag \\
         & = 2 \Ree \ i \sum_{\lambda = 1,2} \int_{\mathbb{R}^3} \frac{\mathcal{F}[\kappa](k)}{\sqrt{2|k|}} \boldsymbol{\epsilon}^\ell_\lambda(k) k_m \left( e^{ik \cdot q_j(s)} -  e^{ik \cdot x_j} \right) \alpha_s(k,\lambda)  \textnormal{d}k - \left( m \longleftrightarrow \ell \right) \label{eq : classical minus quantum Faraday term I} \\
         & \ + i \sum_{\lambda = 1,2} \int_{\mathbb{R}^3} \frac{\mathcal{F}[\kappa](k)}{\sqrt{2|k|}} \boldsymbol{\epsilon}^\ell_\lambda(k) k_m \left[ e^{ik \cdot x_j}\left( \alpha_s(k,\lambda) -  \hbar^{1/2} a_{k,\lambda} \right) - \textnormal{h.c.} \right]\textnormal{d}k  - \left( m \longleftrightarrow \ell \right) \label{eq : classical minus quantum Faraday term II} \ . 
\end{align}
Here we used the notation $\left( m \longleftrightarrow \ell \right)$ to denote the expression preceding it with $\ell$ and $m$ exchanged. Let us denote $\textnormal{I}_b^{(1)} = \eqref{eq : classical minus quantum Faraday term I}$ and $\textnormal{I}_b^{(2)} = \eqref{eq : classical minus quantum Faraday term II}$. In view of the identity
\begin{equation} \label{eq : taylor but no eq}
        \begin{split}
            e^{i k \cdot q_j(s)} - e^{i k \cdot x_j} & = \int_0^1 i k \cdot(q_j(s)-x_j) e^{i k \cdot(\lambda q_j(s)+ (1-\lambda) x_j)} \textnormal{d}\lambda \ , \\
    \end{split}
\end{equation}
we control 
\begin{equation} \label{eq : bound for T j ell m 1 in beta b term 1}
    \begin{split}
        \norm{\textnormal{I}_b^{(1)} \Psi_s} & \lesssim \norm{|\cdot|^{3/2}(1+|\cdot|^2)^{-\sigma/2}\mathcal{F}[\kappa]}_{L^2(\mathbb{R}^3)} \norm{\alpha_s}_{\mathfrak{h}^\sigma} \norm{(x_j - q_j(s))\Psi_s} \\
    & \lesssim_\kappa (s+1)\norm{(x_j - q_j(s))\Psi_s} \ . \\
    \end{split}
\end{equation}
Let us remark here that we used $\norm{\alpha_s}_{\mathfrak{h}^\sigma}$ rather than $\norm{\alpha_s}_{\dot{\mathfrak{h}}^{1/2}}$, because the latter would have lead to the quantity $\norm{|\cdot|\mathcal{F}[\kappa]}_{L^2(\mathbb{R}^3)}$, which is not finite under Assumption \ref{assumptions : kappa} in case $\sigma \neq 1/2$. However, notice that
\begin{equation}
\norm{|\cdot|^{3/2}(1+|\cdot|^2)^{-\sigma/2}\mathcal{F}[\kappa]}_{L^2(\mathbb{R}^3)} \leq \norm{|\cdot|^{3/2-\sigma}\mathcal{F}[\kappa]}_{L^2(\mathbb{R}^3)}
\end{equation}
which is finite, but the price to pay is the extra $(s+1)$ factor, given by \eqref{eq : first bounds in a priori for NM}. Now, in order to estimate the term $\textnormal{I}_b^{(2)}$, we proceed by defining the operator
\begin{equation}
    \mathfrak{C}_{j \ell m}(s) = i \sum_{\lambda = 1,2} \int_{\mathbb{R}^3} \frac{\mathcal{F}[\kappa](k)}{\sqrt{2|k|}} \boldsymbol{\epsilon}^\ell_\lambda(k) k_m  e^{ik \cdot x_j}\left( \alpha_s(k,\lambda) -  \hbar^{1/2} a_{k,\lambda} \right) \textnormal{d}k \ , \\
\end{equation}
so that there holds
\begin{equation}
    \textnormal{I}_b^{(2)} =  \mathfrak{C}_{j \ell m}(s) +  \mathfrak{C}_{j \ell m}(s)^* \ . \\
\end{equation}
Using \eqref{eq : rewriting for beta c} and Cauchy-Schwarz, we have
\begin{equation}
    \norm{\mathfrak{C}_{j \ell m}(s) \Psi_s} \lesssim \norm{|\cdot|^{1/2}\mathcal{F}[\kappa]}_{L^2(\mathbb{R}^3)} \boldsymbol{\beta}^c(\Psi_s,\alpha_s) \ . \\
\end{equation}
Moreover, a calculation using the CCR yields
\begin{equation} \label{eq : CCR reasoning to obtain beta c}
    \begin{split}
        \norm{\mathfrak{C}_{j \ell m}(s) \Psi_s}^2 - \norm{\mathfrak{C}_{j \ell m}(s)^* \Psi_s}^2 &  = - \hbar \norm{|\cdot|^{1/2} \mathcal{F}[\kappa]}^2_{L^2(\mathbb{R}^3)} \ , \\
    \end{split}
\end{equation} 
hence giving immediately
\begin{equation} \label{eq : control for Ibetab2}
\norm{\textnormal{I}_b^{(2)} \Psi_s} \lesssim_\kappa \boldsymbol{\beta}^c(\Psi_s,\alpha_s) + \hbar^{1/2} \ . \\
\end{equation}
Combining the above estimates with \eqref{eq : first bound on term I 1 beta b} and the Young inequality, we find
\begin{equation} \label{eq : last estimate I b 1}
    \begin{split}
         \left \lvert \textnormal{I}_b\right \lvert & \lesssim_\kappa (t+1) \int_0^t \left( \boldsymbol{\beta}(\Psi_s, u(s)) + \hbar \right) \textnormal{d}s \ . \\
    \end{split}
\end{equation}

\noindent \textbf{Term $\textnormal{II}_b$}. We start by bounding with Cauchy-Schwarz and Young,
\begin{equation}
    \left \lvert \textnormal{II}_b\right \lvert \lesssim_N \int_0^t \left( \boldsymbol{\beta}^b(\Psi_s, u(s)) + \sum_{j=1}^N \sum_{k\neq j} \norm{\left( (\nabla V )(q_j(s) - q_k(s)) - (\nabla V )(x_j-x_k)\right)\Psi_s}^2 \right) \textnormal{d}s \ . \\
\end{equation}
Next, using the fundamental theorem of calculus as in \eqref{eq : FTC nabla V} together with the last estimate in \eqref{eq : pointwise bounds for V and nabla V} we control 
\begin{equation}
    \begin{split}
        \norm{\left( (\nabla V )(q_j(s) - q_k(s)) - (\nabla V )(x_j-x_k)\right)\Psi_s} \lesssim \norm{\mathcal{F}[\kappa]}^2_{L^2(\mathbb{R}^3)} \left( \norm{(x_j - q_j(s))\Psi_s} + \norm{(x_k - q_k(s))\Psi_s} \right) \ , \\
    \end{split}
\end{equation}
so that
\begin{equation} \label{eq : last estimate I b 2}
    \left \lvert \textnormal{II}_b\right \lvert \lesssim_{N,\kappa} \int_0^t \left( \boldsymbol{\beta}^a(\Psi_s, q(s)) + \boldsymbol{\beta}^b(\Psi_s, u(s)) \right) \textnormal{d}s \ . \\
\end{equation}

\noindent \textbf{Term $\textnormal{III}_b$}. This term is handled in a very similar fashion as the first one ; we first write 
\begin{equation}
    \hbar^{1/2} \boldsymbol{\hat{\mathbb{E}}}^\ell(x_j) - \boldsymbol{\mathbb{E}}_{\alpha_s}^\ell(q_j(s)) = \hbar^{1/2} \boldsymbol{\hat{\mathbb{E}}}^\ell(x_j) - \boldsymbol{\mathbb{E}}_{\alpha_s}^\ell(x_j) +\boldsymbol{\mathbb{E}}_{\alpha_s}^\ell(x_j) - \boldsymbol{\mathbb{E}}_{\alpha_s}^\ell(q_j(s)) \ , \\
\end{equation}
and then use the definition \eqref{eq : quantized electric field} to obtain terms similar to that obtained in \eqref{eq : classical minus quantum Faraday term I} and \eqref{eq : classical minus quantum Faraday term II}. We then apply a similar reasoning to obtain 
\begin{equation} \label{eq : last estimate I b 3}
     \left \lvert \textnormal{III}_b\right \lvert \lesssim_{N,\kappa} (t+1) \int_0^t \left( \boldsymbol{\beta}(\Psi_s, u(s)) + \hbar \right) \textnormal{d}s \ . \\
\end{equation}

Now, combining \eqref{eq : last estimate I b 1}, \eqref{eq : last estimate I b 2} and \eqref{eq : last estimate I b 3} yields :
\begin{equation} \label{eq : final bound betab}
	\boldsymbol{\beta}^b(\Psi_t, u(t)) \leq \boldsymbol{\beta}^b(\Psi_0, u_0) +  C(t+1)\int_0^t  \left( \boldsymbol{\beta}(\Psi_s, u(s)) + \hbar \right) \textnormal{d}s \ . \\
\end{equation}

Let us now treat $\boldsymbol{\beta}^c(\Psi_t, \alpha_t)$. Similarly as for the previous functional, we denote $\textnormal{I}_c = \eqref{eq : beta c term I}$ and $\textnormal{II}_c = \eqref{eq : beta c term II}$. We use \eqref{eq : first bounds in a priori for NM}, \eqref{eq : rewriting for beta c} and the operator inequality
\begin{equation} \label{eq : taylor but not comp for Gx}
	\lvert \mathbf{G}_{q_j(s)}(k,\lambda) - \mathbf{G}_{x_j}(k,\lambda)\lvert \lesssim |k|^{1/2} |\mathcal{F}[\kappa](k)|  \lvert x_j - q_j(s) \lvert \\
\end{equation}
which follows from \eqref{eq : taylor but no eq}, to bound with Cauchy-Schwarz and Young :
\begin{equation}
    \begin{split}
        \left \lvert \textnormal{I}_c \right\lvert & \lesssim \int_0^t \left( \boldsymbol{\beta}^a(\Psi_s, q(s)) + \boldsymbol{\beta}^c(\Psi_s, \alpha_s) \right) \textnormal{d}s \ . \\
    \end{split}
\end{equation}
By similar means, we estimate the second term as
\begin{equation}
    \begin{split}
        \left \lvert \textnormal{II}_c \right\lvert & \lesssim  \int_0^t \left( \boldsymbol{\beta}^b(\Psi_s, u(s)) + \boldsymbol{\beta}^c(\Psi_s, \alpha_s) \right) \textnormal{d}s \ . \\
    \end{split}
\end{equation}
Hence,  we obtain 
\begin{equation} \label{eq : final bound betac}
     \boldsymbol{\beta}^c(\Psi_t, \alpha_t) \leq C\int_0^t \boldsymbol{\beta}(\Psi_s, u(s)) \textnormal{d}s  \ .
\end{equation} \\

\noindent Combining \eqref{eq : final bound on beta a in proof}, \eqref{eq : final bound betab} and \eqref{eq : final bound betac} gives the desired result.  \\
\end{proof}

\subsection{Proof of the main results} \label{subsection : Proof of the main results}

In view of the previous subsections, we are now able to give the 

\begin{proof}[Proof of Theorem \ref{thm : main result}]
	Combining the estimates from Lemma \ref{lemma : Growth of the functional}, we find immediately :
\begin{equation}
    \boldsymbol{\beta}(\Psi_t, u(t)) \leq \boldsymbol{\beta}(\Psi_0, u_0) + C(t+1) \int_0^t \left( \boldsymbol{\beta}(\Psi_s, u(s)) + \hbar \right) \textnormal{d}s \ , \\
\end{equation}
which yields, by Grönwall's lemma :
\begin{equation}
   \boldsymbol{\beta}(\Psi_t, u(t)) \leq Ce^{ct^2} \left( \boldsymbol{\beta}(\Psi_0, u_0) + \hbar \right) \ . \\
\end{equation}
Assumption \ref{assumptions : init 1} and \eqref{eq : appendix beta b} ensure that $\boldsymbol{\beta}(\Psi_0, u_0) \lesssim \hbar$, from which follow the first and the third bound in \eqref{eq : cvgce in variance thm}. The second bound in \eqref{eq : cvgce in variance thm} is then obtained using \eqref{eq : appendix beta b tilda}. \\
\end{proof}

Next, we give the 

\begin{proof}[Proof of Corollary \ref{Corollary : corollary of main theorem for pure states}]
We start by proving \eqref{eq : expectation of observables first result}. The Cauchy-Schwarz inequality gives 
\begin{equation} \label{eq : spectral measure and FTC to conclude corollary}
    \begin{split}
        \left \lvert \ps{\Psi_t}{f(\mathcal{J}) \Psi_t} - f(J(t)) \right \lvert^2 & \leq \ps{\Psi_t}{\left(f(\mathcal{J}) - f(J(t)) \right)^2\Psi_t} \\
    & =  \int_{\mathbb{R}^{3N}}\left(f(s) - f(J(t) \right)^2 \textnormal{d}\nu(s) \\
    & \leq \norm{\nabla f}^2_{L^\infty} \ps{\Psi_t}{\left(\mathcal{J} - J(t) \right)^2\Psi_t} \\
    & \leq  \norm{\nabla f}^2_{L^\infty} Ce^{ct^2}\hbar \ . \\
    \end{split}
\end{equation}
Here we used $\nu = \nu_{\mathcal{J}, \Psi_t}$ the spectral measure associated with $\mathcal{J}$ and $\Psi_t$ as well as the fundamental theorem of calculus, and finally \eqref{eq : cvgce in variance thm}.
In \eqref{eq : spectral measure and FTC to conclude corollary}, going from the third to the fourth line is immediate if $\mathcal{J} \in \left\{ -i\hbar\nabla, \mathbf{x} \right\}$. If $\mathcal{J} = \hbar ^{1/2}\Phi(g)$, we write :
\begin{equation}
    \mathcal{J} - J(t) = \sum_{\lambda = 1,2} \int_{\mathbb{R}^3} \overline{g(k,\lambda)} \left(\hbar^{1/2}a_{k,\lambda} - \alpha_t(k,\lambda) \right) \textnormal{d}k + \textnormal{h.c.} \\
\end{equation}
and use \eqref{eq : CCR reasoning to obtain beta c} to estimate :
\begin{equation}
    \begin{split}
        \norm{\left( \mathcal{J} - J(t)\right) \Psi^\hbar_u(t)} & \lesssim \sum_{\lambda = 1,2} \int_{\mathbb{R}^3} |g(k,\lambda)| \norm{(\hbar^{1/2}a_{k,\lambda} - \alpha_t(k,\lambda))\Psi^\hbar_u(t)} \textnormal{d}k + \hbar^{1/2} \\
        & \lesssim (1+\norm{g}_{\mathfrak{h}})Ce^{ct^2}\hbar^{1/2} \\
    \end{split}
\end{equation}
thanks to \eqref{eq : rewriting for beta c} and Lemma \eqref{lemma : Growth of the functional}. Let us now prove \eqref{eq : cvge density matrices}. In complete analogy to \cite[Lemma 7.1]{Leopold_2018} one can prove the following estimate :
\begin{equation}
\norm{ \gamma_{\Psi_t} - \ket{\alpha_t} \bra{\alpha_t} }_{\mathfrak{S}^1(\mathfrak{h})} \leq 3 \boldsymbol{\beta}^c(\Psi_t, \alpha_t) + 6 \norm{\alpha_t}_{\mathfrak{h}} \sqrt{\boldsymbol{\beta}^c(\Psi_t, \alpha_t)} \ . \\
\end{equation}
Hence using Theorem~\ref{thm : main result}, $\norm{\alpha_t}_{\mathfrak{h}} \leq \norm{\alpha_t}_{\mathfrak{h}^\sigma} \leq \norm{u(t)}_{X^\sigma}$ and \eqref{eq : second bounds in a priori for NM}, we deduce \eqref{eq : cvge density matrices}. \\
\end{proof}

Let us now turn to the proof of Corollary \ref{corollary : more general states}. We begin with two preliminary lemmas :
\begin{lemma} \label{lemma : Fubini for bochner with rho}
Let $\sigma \in \left[1/2, 1 \right]$, and $\mathcal{J}, f$ as defined in \eqref{eq : quantum and classical observables and f}. Let $\mu \in \mathfrak{B}(X^\sigma)$ be such that
\begin{equation} \label{eq : assumption for Fubini}
\int_{X^\sigma} \norm{u}^2_{X^{\sigma}} \textnormal{d}\mu(u) < + \infty \ . \\
\end{equation}
Then, the following identity holds :
\begin{equation} \label{eq : fubini identity}
\tr{f(\mathcal{J}) \rho_\hbar(t)} = \int_{X^\sigma} \ps{\Psi^\hbar_u(t)}{f(\mathcal{J})\Psi^\hbar_u(t)}\textnormal{d}\mu(u) \ . \\
\end{equation}
\end{lemma} 

\begin{lemma} \label{lemma : fubini with bochner integral for the kernels of RDM}
    Let $\sigma \in [1/2,1]$.  Let $\rho_\hbar(t)$ and $\Psi_u^\hbar(t)$ be as defined in \eqref{eq : time evolved density matrix}.     
Consider their associated one-particle reduced density matrices, respectively given by $\Gamma_{\rho_\hbar(t)}$ and $\gamma_{\Psi_u^\hbar(t)}$. Let $\mu \in \mathfrak{B}(X^\sigma)$ be such that
    \begin{equation} \label{eq : assumption for the kernel lemma}
    \int_{X^\sigma} \norm{u}^2_{X^{\sigma}} \textnormal{d}\mu(u) < + \infty \ . \\
    \end{equation}
    Then,
    \begin{equation} \label{eq : operator equality RDM general state}
        \Gamma_{\rho_\hbar(t)} = \int_{X^\sigma} \gamma_{\Psi^\hbar_u(t)} \textnormal{d}\mu(u) \\
    \end{equation}
on $\mathfrak{S}^1(\mathfrak{h})$. \\
\end{lemma}

\noindent 
To prove these lemmas, we will rely on the following Fubini type theorem for Bochner integrals :
\begin{theorem}[\cite{yosida2012functional}] \label{theorem : Yosida bochner fubini}
Let $(S, \mathfrak{B}, \mu)$ be a measure space. Let $X$ and $Y$ be two Banach spaces. Let $T : X \rightarrow Y$ be a bounded linear operator. If $x$ is an $X$-valued Bochner $\mu$-integrable function, then $Tx$ is a $Y$-valued Bochner $\mu$-integrable function, and there holds
\begin{equation}
\int_B Tx (u) \textnormal{d}\mu(u) = T \int_B x(u) \textnormal{d}\mu(u) \\
\end{equation}
on $Y$, for any set $B \in \mathfrak{B}$. \\
\end{theorem}

\begin{proof}[Proof of Lemma \ref{lemma : Fubini for bochner with rho}]
To begin with, let us justify that the following identity holds :
\begin{equation} \label{eq : fubini with f J in proof lemma}
    f(\mathcal{J}) \rho_{\hbar}(t) = \int_{X^\sigma} \ket{f(\mathcal{J}) \Psi_u^\hbar(t)} \bra{\Psi_u^\hbar(t)} \textnormal{d}\mu(u) \\
\end{equation}
on $\mathcal{B}(\mathfrak{H})$. To prove this, we will apply Theorem \ref{theorem : Yosida bochner fubini}. Consider the following linear map :
\begin{equation}
    \begin{split}
        \mathcal{L}_{f(\mathcal{J})} & : \mathcal{X} \longrightarrow \ \ \mathcal{Y} \\
        & \ \ \ \ \mathcal{O} \longmapsto f(\mathcal{J}) \mathcal{O} \ , \\
    \end{split}
\end{equation}
where we defined $\mathcal{X} = \mathcal{B}(\mathfrak{H}, \mathcal{D}(\mathcal{J}))$, endowed with the operator norm :
\begin{equation}
    \norm{\mathcal{O}}_{\mathcal{X}} := \sup_{\norm{\psi}_{\mathfrak{H}} = 1} \norm{(\mathcal{J}+1) \mathcal{O}\psi} \ , \\
\end{equation}
and $\mathcal{Y} = \mathcal{B}(\mathfrak{H})$. Then, we have :
\begin{equation}
    \begin{split}
    \norm{\mathcal{L}_{f(\mathcal{J})}}_{\mathcal{X} \rightarrow \mathcal{Y}} & = \sup_{\norm{\mathcal{O}}_{\mathcal{X}} = 1} \norm{f(\mathcal{J})\mathcal{O}}_{\mathfrak{H} \rightarrow \mathfrak{H}} \\
    & = \sup_{\norm{\mathcal{O}}_{\mathcal{X}} = 1} \sup_{\norm{\psi}_{\mathfrak{H}} = 1} \norm{f(\mathcal{J})\mathcal{O}\psi}_{\mathfrak{H} } \ . \\
    \end{split}
\end{equation}
Now, let $\phi = \mathcal{O} \psi$ with $\mathcal{O}, \psi$ as in the above. Consider $\nu = \nu_{\mathcal{J},\phi}$ the spectral measure associated to the self-adjoint operator $\mathcal{J}$ and the vector $\phi$. Then, we can write :
\begin{equation} \label{eq : ftc to see boundedness of f of op}
    \begin{split}
        \norm{\left(f(\mathcal{J}) - f(0) \right)\phi}_{\mathfrak{H}}^2 & = \int_{\mathbb{R}} \left \lvert f(s) - f(0) \right\lvert^2 \textnormal{d}\nu(u) \\
        & = \int_{\mathbb{R}} \left \lvert \int_0^1 \nabla f(\eta s) \cdot s \ \textnormal{d} \eta\right\lvert^2 \textnormal{d}\nu(u) \\
        & \leq \norm{\nabla f}^2_{L^\infty} \norm{\mathcal{J}\phi}_{\mathfrak{H}}^2 \ , \\
    \end{split}
\end{equation}
where we used the fundamental theorem of calculus since $f \in W^{1,\infty} \subset W^{1,1}_{\textnormal{loc}}$. We thus obtain :
\begin{equation}
    \begin{split}
        \norm{f(\mathcal{J}) \mathcal{O}\psi}_{\mathfrak{H}} & \leq \norm{\nabla f}_{L^\infty} \norm{\mathcal{J}\mathcal{O}\psi}_{\mathfrak{H}} + |f(0)| \norm{\mathcal{O}\psi}_{\mathfrak{H}} \\
        & \lesssim_f \norm{\left(\mathcal{J}+1\right)\mathcal{O}\psi}_{\mathfrak{H}} \ , \\
    \end{split}
\end{equation}
which yields :
\begin{equation}
    \norm{\mathcal{L}_{f(\mathcal{J})}}_{\mathcal{X} \rightarrow \mathcal{Y}} \leq C < + \infty \ , \\
\end{equation}
hence $\mathcal{L}_{f(\mathcal{J})}$ is a bounded linear operator. Moreover, $\ket{\Psi^\hbar_u(t)} \bra{\Psi^\hbar_u(t)}$ is $\mathcal{X}$-valued and $\mu$-integrable, since by Cauchy-Schwarz,
\begin{equation}
    \norm{\ket{\Psi^\hbar_u(t)} \bra{\Psi^\hbar_u(t)}}_{\mathcal{X}} \leq \norm{\mathcal{J}\Psi_u^\hbar(t)}_{\mathfrak{H}} + 1 \ . \\
\end{equation}
In order to see that the above norm is finite and $\mu$-integrable, we remark that it in fact relates to the three functionals defined in \eqref{def : Definition of the functional}. Indeed, if $\mathcal{J} = -i\hbar \nabla$, then using \eqref{eq : appendix beta b tilda} we can write 
\begin{equation} \label{eq : relating fJ with the functionals}
    \begin{split}
        \norm{-i\hbar \nabla \Psi_u^\hbar(t)}_{\mathfrak{H}}& \leq \boldsymbol{\beta}(\Psi_u^\hbar(t), u(t))^{1/2} + |p(t)| \\
        & \lesssim_{\kappa, \sigma, t, \hbar} \norm{u}_{X^\sigma} + 1 \ , \\
    \end{split}
\end{equation}
which is finite and $\mu$-integrable given our assumptions. Here, $t \mapsto u(t) = (q(t), p(t), \alpha_t)$ is the unique solution to \eqref{eq : newton motion eqs PF} associated with the initial data $u = (q,p,\alpha)$. We also used \eqref{eq : second bounds in a priori for NM}, and Theorem \ref{thm : main result}, since there holds 
\begin{equation}
    \begin{split}
        \norm{(-i\hbar\nabla - p) \Psi^\hbar_u}^2 & \leq C \hbar \\
        \norm{(x- q) \Psi^\hbar_u}^2 & \leq C \hbar \\
        \ps{\Psi^\hbar_u}{W(\hbar^{-1/2}\alpha) \mathcal{N}W^*(\hbar^{-1/2}\alpha)\Psi^\hbar_u}_{\mathfrak{H}}  & = 0 \ . \\
    \end{split}
\end{equation}
These are obtained directly from the definition of $\Psi^\hbar_u$, as pointed out in Remark \ref{remark : rks on the initial assumptions}. If $\mathcal{J} = \mathbf{x}$, the reasoning is similar. If $\mathcal{J} = \hbar^{1/2}\Phi(g)$ for some $g \in \mathfrak{h}$, we control
\begin{equation}
    \norm{\Phi(g) \Psi^\hbar_u(t)}_{\mathfrak{H}} \lesssim \norm{g}_\mathfrak{h} \norm{\mathcal{N}^{1/2}\Psi^\hbar_u(t)}_{\mathfrak{H}} \ . \\
\end{equation}
Then, using the shifting relation
\begin{equation}
    W^*(f) \mathcal{N}W(f) = \mathcal{N} + \Phi(f) + \norm{f}^2 \\
\end{equation}
we compute :
\begin{equation} \label{eq : relating norm N with beta c}
    \begin{split}
        \hbar \norm{\mathcal{N}^{1/2}\Psi_u^\hbar(t)}_{\mathfrak{H}}^2 & = \hbar \norm{\mathcal{N}^{1/2}W^*(\hbar^{-1/2} \alpha_t)\Psi_u^\hbar(t)}^2 + \hbar^{1/2} \ps{W^*(\hbar^{-1/2} \alpha_t)\Psi_u^\hbar(t)}{\Phi(\alpha_t)W^*(\hbar^{-1/2} \alpha_t)\Psi_u^\hbar(t)} 
\\
&\quad    + \norm{\alpha_t}^2_{\mathfrak{h}} \\
        & \leq \boldsymbol{\beta}^c(\Psi_u^\hbar(t),\alpha_t) + 2 \norm{\alpha_t}_\mathfrak{h} \boldsymbol{\beta}^c(\Psi_u^\hbar(t),\alpha_t)^{1/2}         
+ \norm{\alpha_t}^2_\mathfrak{h} \\
        & \lesssim_{\kappa,\sigma,t} \hbar + \hbar^{1/2}\norm{u}_{X^\sigma} + \norm{u}_{X^\sigma}^2 \ . \\
    \end{split}
\end{equation}
Here in the second line we used the definition of $\boldsymbol{\beta}^c$, and in the third line we used that $\norm{\alpha_t}_{\mathfrak{h}} \leq \norm{\alpha_t}_{\mathfrak{h}^\sigma} \leq \norm{u(t)}_{X^\sigma}$, together with \eqref{eq : second bounds in a priori for NM}. We conclude again with the assumption on $\mu$. This proves \eqref{eq : fubini with f J in proof lemma}. Next, note that in view of the inequality 
\begin{equation}
\left \lvert \tr{\mathcal{O}} \right \lvert \leq \tr{\left \lvert \mathcal{O} \right \lvert} ,
\end{equation}
it holds that $\textnormal{Tr} : \mathfrak{S}^1(\mathfrak{H}) \rightarrow \mathbb{C}$ is a bounded linear operator. Now, note that the triangle inequality implies 
\begin{equation} \label{eq : from identity on bounded ops to trace class ops}
\begin{split}
\norm{\int_{X^\sigma}\ket{f(\mathcal{J})\Psi^\hbar_u(t)} \bra{\Psi^\hbar_u(t)} \textnormal{d}\mu(u)}_{\mathfrak{S}^1(\mathfrak{H})} & \leq \int_{X^\sigma} \norm{\ket{f(\mathcal{J})\Psi^\hbar_u(t)} \bra{\Psi^\hbar_u(t)}}_{\mathfrak{S}^1(\mathfrak{H})} \textnormal{d}\mu(u) \\
& = \int_{X^\sigma} \norm{f(\mathcal{J})\Psi^\hbar_u(t)}  \textnormal{d}\mu(u) \\
& < + \infty \ . \\
\end{split} 
\end{equation}
Here in the second line we used the fact that $\ket{\Psi^\hbar_u(t)} \bra{\Psi^\hbar_u(t)}$ is a rank-one operator. To see why the third line holds, one has to apply the same reasoning we already did from \eqref{eq : ftc to see boundedness of f of op} on, and use Assumption \ref{eq : assumption for Fubini} with the fact that $\mu$ is a probability measure. As a consequence, this implies that $\int_{X^\sigma}\ket{f(\mathcal{J})\Psi^\hbar_u(t)} \bra{\Psi^\hbar_u(t)} \textnormal{d}\mu(u)$ is trace class, so is $f(\mathcal{J}) \rho_\hbar(t)$, and thus the identity \eqref{eq : fubini with f J in proof lemma} actually holds on $\mathfrak{S}^1(\mathfrak{H})$. Another consequence of \eqref{eq : from identity on bounded ops to trace class ops} is that $u \mapsto \ket{f(\mathcal{J})\Psi^\hbar_u(t)} \bra{\Psi^\hbar_u(t)}$ is $\mathfrak{S}^1(\mathfrak{H})$-valued and $\mu$-integrable, so that Theorem \ref{theorem : Yosida bochner fubini} applies, allowing us to write :
\begin{equation}
\begin{split}
\tr{f(\mathcal{J}) \rho_\hbar(t)} & = \tr{\int_{X^\sigma}\ket{f(\mathcal{J})\Psi^\hbar_u(t)} \bra{\Psi^\hbar_u(t)} \textnormal{d}\mu(u)} \\
& = \int_{X^\sigma}\ps{f(\mathcal{J})\Psi^\hbar_u(t)}{\Psi^\hbar_u(t)} \textnormal{d}\mu(u) \ , \\
\end{split}
\end{equation}
as desired. \\
\end{proof}

We now give the 
\begin{proof} [Proof of Lemma \eqref{lemma : fubini with bochner integral for the kernels of RDM}]
In order to prove \eqref{eq : operator equality RDM general state}, let us first start by noticing that we have the following equality of kernels :
    \begin{equation} \label{eq : identity kernels rho}
        \left( \int_{X^\sigma} \gamma_{\Psi^\hbar_u(t)} \textnormal{d}\mu(u) \right) (k,\lambda ; \ell, \nu) = \int_{X^\sigma} \gamma_{\Psi^\hbar_u(t)}(k,\lambda ; \ell, \nu) \textnormal{d}\mu(u) \quad \textnormal{a.e.} \ . \\
    \end{equation}
Indeed, if we define the linear map
\begin{equation}
    \begin{split}
        \mathcal{K} & : \mathfrak{S}^1(\mathfrak{h}) \longrightarrow \ \ L^2(\Omega \times \Omega) \\
        & \ \ \ \ \ \ \ \mathcal{O}  \ \ \ \  \longmapsto  \ \ \mathcal{O}(.;.) \\
    \end{split}
\end{equation}
where $\mathcal{O}(.;.)$ is the integral kernel of the operator $\mathcal{O}$, then $\mathcal{K}$ is a bounded linear operator for its induced operator norm, since 
\begin{equation}
    \begin{split}
        \norm{\mathcal{K}}_{\textnormal{op}} & = \sup_{\norm{\mathcal{O}}_{\mathfrak{S}^1(\mathfrak{h})} = 1} \norm{\mathcal{K}\mathcal{O}}_{L^2(\Omega \times \Omega)} \\
        & = \sup_{\norm{\mathcal{O}}_{\mathfrak{S}^1(\mathfrak{h})} = 1} \norm{\mathcal{O}}_{\mathfrak{S}^{2}(\mathfrak{h})} \\
        & \leq 1 \ . \\
    \end{split}
\end{equation}
Here, we used that for any Hilbert-Schmidt operator $\mathcal{O}$ with integral kernel $\mathcal{O}(.;.)$ on $\Omega \times \Omega$, there holds 
\begin{equation} \label{eq : Hs norm with kernel for HS op}
    \norm{\mathcal{O}}^2_{\mathfrak{S}^2(\mathfrak{h})} = \iint_{\Omega\times \Omega} |\mathcal{O}(x;y)|^2 \textnormal{d}x \textnormal{d}y \ . \\
\end{equation}
Moreover, since $\gamma_{\Psi_u^\hbar(t)}$ is a non-negative trace-class operator, there holds (see for example \cite[Chapter 5]{namMQMII2020}) :
\begin{equation} \label{eq : computation trace norm of gamma psi}
    \begin{split}
        \norm{\gamma_{\Psi_u^\hbar(t)}}_{\mathfrak{S}^1(\mathfrak{h})} & = \tr{\gamma_{\Psi_u^\hbar(t)}} \\
        & = \hbar\ps{\Psi_u^\hbar(t)}{\mathcal{N}\Psi_u^\hbar(t)} \ , \\
    \end{split}
\end{equation}
which is finite and $\mu$- integrable in view of \eqref{eq : relating norm N with beta c}. Hence by Theorem \eqref{theorem : Yosida bochner fubini}, we have
\begin{equation}
    \mathcal{K}\int_{X^\sigma} \gamma_{\Psi^\hbar_u(t)} \textnormal{d}\mu(u) = \int_{X^\sigma}  \mathcal{K} \gamma_{\Psi^\hbar_u(t)} \textnormal{d}\mu(u) \ , \\
\end{equation}
which is exactly \eqref{eq : identity kernels rho}. Consequently, if we prove that there holds
\begin{equation} \label{eq : idenity kernels rho}
    \Gamma_{\rho_\hbar(t)}(k,\lambda ; \ell, \nu) = \int_{X^\sigma} \gamma_{\Psi^\hbar_u(t)}(k,\lambda ; \ell, \nu) \textnormal{d}\mu(u) \ , \\
\end{equation}
the proof will be complete, since trace-class (thus Hilbert-Schmidt) operators on a $L^2$ space are completely characterized by their integral kernel. 
Let us define the operator $T_{\hbar,t}$ with integral kernel
\begin{equation}
    T_{\hbar,t}(k,\lambda ; \ell, \nu) = \int_{X^\sigma} \gamma_{\Psi^\hbar_u(t)}(k,\lambda ; \ell, \nu) \textnormal{d}\mu(u) \ , \\
\end{equation}
and let us show that there holds
\begin{equation}
    \begin{split}
        \ps{f}{\left( \Gamma_{\rho_\hbar(t)} - T_{\hbar,t} \right)g}_{\mathfrak{h}} & = 0 \ . \\
    \end{split}
\end{equation}
for all $f,g \in \mathfrak{h}$. On the one hand, we have 
\begin{equation}
    \ps{f}{\Gamma_{\rho_\hbar(t)}g}_{\mathfrak{h}} = \hbar\tr{a^*(g)a(f)\rho_\hbar(t)} \\
\end{equation}
by definition of $\Gamma_{\rho_\hbar(t)}$. On the other hand, we compute :
\begin{equation} \label{eq : ps Tht f g}
    \begin{split}
        \ps{f}{ T_{\hbar,t}g}_{\mathfrak{h}} & = \sum_{\lambda, \nu =1,2} \iint_{\mathbb{R}^3 \times \mathbb{R}^3} \overline{f(k,\lambda)}T_{\hbar,t}(k,\lambda ; \ell, \nu)g(\ell, \nu) \textnormal{d}\ell \textnormal{d}k \\
        & = \int_Y \left( \int_X \overline{f(x)} \gamma_{\Psi^\hbar_u(t)}(x;y) g(y) \textnormal{d}\mu(u) \right)\textnormal{d}\eta(x,y) \ , \\
    \end{split}
\end{equation}
where we defined $X = X^\sigma$, $Y = \left(\mathbb{R}^3 \times \left\{1,2 \right\}\right)^2$, and $\eta$ is the Lebesgue measure on $Y$. By Tonelli, we know that there holds
\begin{equation}
    \begin{split}
        \int_{X \times Y} \left \lvert f(x) \gamma_{\Psi^\hbar_u(t)}(x;y) g(y)  \right \lvert \textnormal{d} (\mu \otimes \eta)(u, (x,y)) & = \int_Y \left( \int_X \left \lvert f(x) \gamma_{\Psi^\hbar_u(t)}(x;y) g(y)  \right \lvert \textnormal{d}\mu(u) \right)\textnormal{d}\eta(x,y) \ . \\
    \end{split}
\end{equation}
Hence, if we show that the right-hand side is finite, we will be able to interchange the integrals in \eqref{eq : ps Tht f g} by Fubini. We thus control :
\begin{equation} \label{eq : Tonelli inequalities}
    \begin{split}
        \int_Y \left( \int_X \left \lvert f(x) \gamma_{\Psi^\hbar_u(t)}(x;y) g(y)  \right \lvert \textnormal{d}\mu(u) \right)\textnormal{d}\eta(x,y) & \leq  \norm{f}_{\mathfrak{h}} \norm{g}_{\mathfrak{h}} \int_X \left( \int_Y \left \lvert \gamma_{\Psi^\hbar_u(t)}(x;y)   \right \lvert^2 \textnormal{d}\eta(x,y) \right)^{\frac{1}{2}} \textnormal{d} \mu(u) \\
        & \leq  \norm{f}_{\mathfrak{h}} \norm{g}_{\mathfrak{h}} \int_X \norm{\gamma_{\Psi_u^\hbar(t)}}_{\mathfrak{S}^2(\mathfrak{h})} \textnormal{d} \mu(u) \\
        & \leq \norm{f}_{\mathfrak{h}} \norm{g}_{\mathfrak{h}} \int_X \norm{\gamma_{\Psi_u^\hbar(t)}}_{\mathfrak{S}^1(\mathfrak{h})} \textnormal{d}\mu(u) \\
        & < + \infty \ . \\
    \end{split}
\end{equation}
In the first line we used Tonelli and Cauchy-Schwarz with respect to $\eta$. In the second line we used \eqref{eq : Hs norm with kernel for HS op}, in the third line we used that the trace norm controls the Hilbert-Schmidt norm, and in the last line we used \eqref{eq : computation trace norm of gamma psi} together with \eqref{eq : relating norm N with beta c} and the initial assumption on $\mu$. Hence, going back to \eqref{eq : Tonelli inequalities}, we deduce by \eqref{eq : assumption for the kernel lemma} and Fubini that we can interchange the integral on $Y$ and the integral on $X$. Thus, \eqref{eq : ps Tht f g} becomes 
\begin{equation}
    \begin{split}
        \ps{f}{T_{\hbar,t}g}_{\mathfrak{h}} & = \int_{X^\sigma} \ps{f}{\gamma_{\Psi_u^\hbar(t)}g}_{\mathfrak{h}} \textnormal{d}\mu(u) \\
        & = \int_{X^\sigma} \hbar \ps{\Psi_u^\hbar(t)}{a^*(g)a(f)\Psi_u^\hbar(t)}_{\mathfrak{h}} \textnormal{d}\mu(u) \\
    \end{split}
\end{equation}
by definition of $\gamma_{\Psi_u^\hbar(t)}$. We are thus left with proving :
\begin{equation} \label{eq : identity for the trace a star a}
    \tr{a^*(g)a(f) \rho_\hbar(t)} = \int_{X^\sigma} \ps{\Psi_u^\hbar(t)}{a^*(g)a(f)\Psi_u^\hbar(t)} \textnormal{d}\mu(u) \ . \\
\end{equation}
To do so, note first that there holds 
\begin{equation} \label{eq : identity af rho}
    a(f) \rho_\hbar(t) = \int_{X^\sigma} \ket{a(f) \Psi_u^\hbar(t)} \bra{\Psi_u^\hbar(t)} \textnormal{d}\mu(u) \\
\end{equation}
on $\mathcal{B}\left(\mathfrak{H}\right)$. Proving this identity is done in a very similar way as for \eqref{eq : fubini with f J in proof lemma}, namely we define
\begin{equation}
    \begin{split}
        \mathcal{L}_{a(f)} & : \mathcal{X} \longrightarrow \ \ \mathcal{Y} \\
        & \ \ \ \mathcal{O}  \ \longmapsto  \ a(f) \mathcal{O} \ , \\
    \end{split}
\end{equation}
where $\mathcal{X} = \mathcal{B}\left(\mathfrak{H},\mathcal{D}\left(\mathcal{N}^{1/2}\right)\right)$ is endowed with the operator norm 
\begin{equation}
    \norm{\mathcal{O}}_{\mathcal{X}}:= \sup_{\norm{\psi}_{\mathfrak{H}} = 1} \norm{(\mathcal{N}+1)^{1/2}\mathcal{O}\psi}_{\mathfrak{H}} \ , \\
\end{equation}
and $\mathcal{Y} = \mathcal{B}\left(\mathfrak{H}\right)$ with the usual norm. Then, we show similarly that it is indeed a bounded linear operator, since we have 
\begin{equation}
    \begin{split}
    \norm{\mathcal{L}_{a(f)}}_{\mathcal{X} \rightarrow \mathcal{Y}}
    & \leq \sup_{\norm{\mathcal{O}}_{\mathcal{X}} = 1}\ \sup_{\norm{\psi}_{\mathfrak{H}} = 1}  \norm{f} \norm{\mathcal{N}^{1/2}\mathcal{O}\psi}_{\mathfrak{H}} \ . \\
    \end{split}
\end{equation}
Moreover, $\ket{\Psi_u^\hbar(t)} \bra{\Psi_u^\hbar(t)}$ is $\mathcal{X}$-valued and $\mu$-integrable, since by Cauchy-Schwarz, 
\begin{equation}
    \begin{split}
        \norm{\ket{\Psi_u^\hbar(t)} \bra{\Psi_u^\hbar(t)}}_{\mathcal{X}}  &\leq \norm{(\mathcal{N}+1)^{1/2}\Psi^\hbar_u(t)}_{\mathfrak{F}} \\
        & \lesssim_{\kappa,\sigma,t,\hbar} \norm{u}^{1/2}_{X^\sigma} + \norm{u}_{X^\sigma} + 1 \ . \\
    \end{split}
\end{equation}
Hence, a direct application of \eqref{theorem : Yosida bochner fubini} yields \eqref{eq : identity af rho}. Now, observe that $a(f)\rho_\hbar(t) \in \mathfrak{S}^1(\mathfrak{H})$. Indeed, using \eqref{eq : identity af rho} and the triangle inequality, we control
\begin{equation}
    \begin{split}
        \norm{a(f)\rho_\hbar(t)}_{\mathfrak{S}^1(\mathfrak{H})} & \leq \int_{X^\sigma} \norm{\ket{a(f) \Psi_u^\hbar(t)} \bra{\Psi_u^\hbar(t)} }_{\mathfrak{S}^1(\mathfrak{H})} \textnormal{d}\mu(u) \\
        & \leq \norm{f}_{\mathfrak{h}} \int_{X^\sigma} \norm{\mathcal{N}^{1/2}\Psi^\hbar_u(t)}_{\mathfrak{H}}\textnormal{d}\mu(u) \ , \\
    \end{split}
\end{equation}   
which is finite in view of \eqref{eq : relating norm N with beta c}. Hence by cyclicity of the trace,
\begin{equation}
    \tr{a^*(g)a(f)\rho_\hbar(t)} = \tr{a(f)\rho_\hbar(t)a^*(g)} \ . \\
\end{equation}
Furthermore, writing 
\begin{equation}
    \begin{split}
        \left(a(f) \rho_\hbar(t) a^*(g) \right)^* & = a(g) \int_{X^\sigma} \ket{\Psi^\hbar_u(t)} \bra{a(f)\Psi^\hbar_u(t)} \textnormal{d}\mu(u) \ , \\
    \end{split}
\end{equation}
we can apply the same reasoning as above for $a(g)$, since 
\begin{equation}
    \begin{split}
        \norm{\ket{\Psi_u^\hbar(t)} \bra{a(f)\Psi_u^\hbar(t)}}_{\mathcal{X}}  &\leq \norm{f}\norm{\mathcal{N}^{1/2}\Psi^\hbar_u(t)}^2 \\
        & \lesssim_{\kappa,\sigma,t,\hbar} \norm{u}_{X^\sigma} + \norm{u}^2_{X^\sigma} + 1 \ . \\
    \end{split}
\end{equation}
Hence, inserting $a(g)$ in the integral above and taking the adjoint yields
\begin{equation}
    a(f) \rho_\hbar(t) a^*(g) =  \int_{X^\sigma} \ket{a(f)\Psi^\hbar_u(t)} \bra{a(g)\Psi^\hbar_u(t)} \textnormal{d}\mu(u) \ . \\
\end{equation}
Now, proving \eqref{eq : identity for the trace a star a} requires to apply again \eqref{theorem : Yosida bochner fubini}, with the bounded linear operator $\textnormal{Tr} : \mathfrak{S}^1(\mathfrak{h}) \rightarrow \mathbb{C}$. We have :
\begin{equation}
    \begin{split}
        \textnormal{Tr} \left \lvert \ket{a(f)\Psi^\hbar_u(t)} \bra{a(g)\Psi^\hbar_u(t)}  \right \lvert & = \norm{a(f)\Psi^\hbar_u(t)} \norm{a(g)\Psi^\hbar_u(t)} _\hbar \\
        & \leq \norm{f} \norm{g} \norm{\mathcal{N}^{1/2}\Psi^\hbar_u(t)}^2 \\
        & \lesssim_{\kappa,\sigma,t,\hbar} \norm{f} \norm{g} \left( \norm{u}_{X^\sigma} + \norm{u}^2_{X^\sigma} + 1 \right) \ . \\
    \end{split}
\end{equation}
Hence $\ket{a(f)\Psi^\hbar_u(t)} \bra{a(g)\Psi^\hbar_u(t)}$ is $\mathfrak{S}^1(\mathfrak{h})$-valued and $\mu$-integrable, and the proof is complete. \\
\end{proof}

We can now turn to the 
\begin{proof}[Proof of Corollary \eqref{corollary : more general states}]
Let us start by proving \eqref{eq : more general states convergence}. Using Lemma \ref{lemma : Fubini for bochner with rho} and the fact that $\ket{\Psi^\hbar_u(t)} \bra{\Psi^\hbar_u(t)}$ has unit trace, we write :
\begin{equation}
\begin{split}
	\left \lvert \tr{f(\mathcal{J}) \rho_\hbar(t)} - \int_{X^\sigma} f(p(t)) \textnormal{d} \mu(u) \right \lvert & = \left \lvert \int_{X^\sigma} \tr{\left(f(\mathcal{J}) - f(J(t)) \right)\ket{\Psi^\hbar_u(t)} \bra{\Psi^\hbar_u(t)}} \textnormal{d} \mu(u)  \right \lvert \\
	& = \left \lvert \int_{X^\sigma} \ps{\Psi^\hbar_u(t)}{\left(f(\mathcal{J}) - f(J(t)) \right)\Psi^\hbar_u(t)} \textnormal{d} \mu(u) \right \lvert \\
	& = \left \lvert \int_{X^\sigma} \int_{\mathbb{R}} (f(s) - f(J(t))) \textnormal{d} \nu(s) \textnormal{d} \mu(u) \right \lvert \\
    & \leq \norm{\nabla f}_{L^\infty} \int_{X^\sigma} \norm{(\mathcal{J} - J(t))\Psi^\hbar_u(t)} \textnormal{d}\mu(u) \\
	& \leq  \norm{\nabla f}_{L^\infty} Ce^{ct^2} \hbar^{1/2} \ . \\
\end{split}
\end{equation}
Here we used similar arguments as in \eqref{eq : spectral measure and FTC to conclude corollary}, making use of the spectral measure $\nu = \nu_{\mathcal{J}, \Psi_u^\hbar(t)}$ and the fundamental theorem of calculus. \\

Let us now turn to the proof of \eqref{eq : more general states reduced density matrices convergence}. Using Lemma \ref{lemma : fubini with bochner integral for the kernels of RDM} and the triangle inequality, we write :
\begin{equation}
    \begin{split}
        \norm{ \Gamma_{\rho_\hbar(t)} - \int_{X^\sigma} \ket{\alpha_t}\bra{\alpha_t} \textnormal{d}\mu(u) }_{\mathfrak{S}^1(\mathfrak{h})}  & = \norm{ \int_{X^\sigma} \left( \gamma_{\Psi^\hbar_u(t)} - \ket{\alpha_t}\bra{\alpha_t} \right)\textnormal{d}\mu(u) }_{\mathfrak{S}^1(\mathfrak{h})} \\
        & \leq \int_{X^\sigma} \norm{\gamma_{\Psi^\hbar_u(t)} - \ket{\alpha_t}\bra{\alpha_t} }_{\mathfrak{S}^1(\mathfrak{h})}  \textnormal{d}\mu(u) \ . \\
    \end{split}
\end{equation}
We then conclude thanks to Corollary \ref{eq : cvge density matrices} and the fact that $\mu$ is a probability measure. The proof is complete. \\
\end{proof}

\newpage
\appendix

\section*{Appendix}

\section{Growth of the third functional} \label{appendix : Growth of the third functional}
Here we give the proof of Lemma \ref{lemma : growth of beta c} which, modulo a density argument, goes along the lines of \cite[Lemma 3.3]{falconi2023derivation}. For the sake of convenience and completeness, we reproduce it here with the needed modifications.
\begin{proof} [Proof of Lemma \ref{lemma : growth of beta c}]
We begin by proving this Lemma in the case $\sigma = 1$. The reason why we restrict ourselves to this case will be explained below. We will recover the general setting via a density argument. Let us then consider $u_0 \in X^1$, and $t \mapsto u(t) = (q(t), p(t), \alpha_t) \in \mathcal{C}(\mathbb{R}; X^1) \cap \mathcal{C}^1(\mathbb{R}; X^0)$ the associated solution to \eqref{eq : newton motion eqs PF}. Let us define the bounded operator $\mathcal{N}_\delta = \mathcal{N}e^{-\delta \mathcal{N}}$. We then set 
\begin{equation}
	\begin{split}
		\boldsymbol{\beta}^c_\delta(\Psi_t, \alpha_t) & = \hbar \ps{\Psi_t}{W(\hbar^{-1/2} \alpha_t) \mathcal{N}_\delta W^*(\hbar^{-1/2} \alpha_t) \Psi_t} \\
		& = \hbar \ps{\boldsymbol{\xi}_{t}}{\mathcal{N}_\delta\boldsymbol{\xi}_{t}}
	\end{split}
\end{equation}
Where we defined $\boldsymbol{\xi}_{t} = W^*(\hbar^{-1/2} \alpha_t) \Psi_t$ the fluctuation vector. First of all, let us note that one has 
\begin{equation}
	\boldsymbol{\beta}^c_\delta(\Psi_t, \alpha_t) \underset{\delta \rightarrow 0}{\longrightarrow} \boldsymbol{\beta}^c(\Psi_t, \alpha_t)
\end{equation}
This can be seen similarly as in \cite[Equation (3.16)]{falconi2023derivation}, by writing :
\begin{equation}
	\begin{split}
		\left \lvert \boldsymbol{\beta}^c_\delta(\Psi_t, \alpha_t) - \boldsymbol{\beta}^c(\Psi_t, \alpha_t) \right \lvert & \leq \hbar \norm{\left( \mathcal{N}_\delta - \mathcal{N} \right)\boldsymbol{\xi}_{t}}
	\end{split}
\end{equation}
and using the spectral theorem to write
\begin{equation}
	\norm{\left( \mathcal{N}_\delta - \mathcal{N} \right)\boldsymbol{\xi}_{t}}^2 = \int_0^\infty \lambda^2(e^{-\delta \lambda} - 1)^2 \textnormal{d}\mu_{\mathcal{N},\boldsymbol{\xi}_{t}}(\lambda)
\end{equation}
where $\textnormal{d}\mu_{\mathcal{N},\boldsymbol{\xi}_{t}}$ is the spectral measure associated to $\mathcal{N}$ and $\boldsymbol{\xi}_{t}$. Dominated convergence allows us to conclude. Now, a standard argument shows that the dynamics of $\boldsymbol{\xi}_{t}$ are given by
\begin{equation} \label{eq : dynamics chit}
    i \hbar \partial_t \boldsymbol{\xi}_{t} = \mathcal{G}(t) \boldsymbol{\xi}_{t} \ ,
\end{equation}
where the time dependant generator $\mathcal{G}(t)$ is given by
\begin{equation}
    \mathcal{G}(t) = i \hbar \dot{W}^*(\hbar^{-1/2} \alpha_t) W(\hbar^{-1/2} \alpha_t) + W^*(\hbar^{-1/2} \alpha_t) \mathbb{H}_{\hbar}W(\hbar^{-1/2} \alpha_t) \ . \\
\end{equation}
From \cite[Lemma 3.1]{ginibreveloclassicallimit}, we know that the map $t \mapsto W(\hbar^{-1/2} \alpha_t)$ is differentiable since $t \mapsto \alpha_t \in \mathcal{C}^1(\mathbb{R}; \mathfrak{h})$, with 
\begin{equation}
\dot{W}^*(\hbar^{-1/2} \alpha_t) = \left(\hbar^{-1/2} \left( a(\partial_t \alpha_t) - a^*(\partial_t \alpha_t) + i\hbar^{-1} \Imm \ps{\partial_t\alpha_t}{\alpha_t} \right) \right) W^*(\hbar^{-1/2} \alpha_t) \ . \\
\end{equation}
Hence, one has the following expression for the generator $\mathcal{G}(t)$ : 
\begin{equation}
		\begin{split}
\mathcal{G}(t) & = \hbar^{1/2} \left( a(|\cdot| \alpha_t - i\partial_t \alpha_t) + a^*(|\cdot| \alpha_t - i\partial_t \alpha_t) \right) + \Imm \ps{\partial_t\alpha_t}{\alpha_t} \\
        & \quad \quad + \sum_{j=1}^N \left(-i \hbar \nabla_j - \hbar^{1/2}\boldsymbol{\hat{\mathbb{A}}}(x_j) - \boldsymbol{\mathbb{A}}_{\alpha_t}(x_j) \right)^2 + \sum_{1 \leq j < k \leq N} V(x_j - x_k) + \hbar H_{\textnormal{f}} \\
        & \quad \quad + \sum_{\lambda = 1,2} \int_{\mathbb{R}^3} |k||\alpha_t(k,\lambda)|^2\textnormal{d}k \ . \\
        \end{split}
\end{equation}
Now, using the above expression, one finds :
\begin{equation}
	\begin{split}
		 \hbar^{-1/2}\com{\mathcal{G}(t)}{\mathcal{N}_\delta} & = \left( \com{a(|\cdot| \alpha_t - i\partial_t \alpha_t)}{\mathcal{N}_\delta} -  \sum_{j=1}^N\left(-i \hbar \nabla_j - \hbar^{1/2}\boldsymbol{\hat{\mathbb{A}}}(x_j) - \boldsymbol{\mathbb{A}}_{\alpha_t}(x_j) \right) \com{\boldsymbol{\hat{\mathbb{A}}}(x_j)}{\mathcal{N}_\delta} \right) - \textnormal{h.c.} \ , \\
	\end{split}
\end{equation}
which itself leads to :
\begin{equation} \label{eq : derivative beta c regularized}
	\begin{split}
		 \dt \boldsymbol{\beta}_\delta^c(\Psi_t,\alpha_t) & = i \ps{\boldsymbol{\xi}_{t}}{\com{\mathcal{G}(t)}{\mathcal{N}_\delta}\boldsymbol{\xi}_{t}} \\
		 & = 2 \hbar^{1/2} \Imm \sum_{j=1}^N \ps{\boldsymbol{\xi}_{t}}{\left(-i \hbar \nabla_j - \hbar^{1/2}\boldsymbol{\hat{\mathbb{A}}}(x_j) - \boldsymbol{\mathbb{A}}_{\alpha_t}(x_j) \right) \com{\boldsymbol{\hat{\mathbb{A}}}(x_j)}{\mathcal{N}_\delta}\boldsymbol{\xi}_{t}} \\
		 & \quad - 2 \hbar^{1/2} \Imm \ps{\boldsymbol{\xi}_{t}}{\com{a(|\cdot| \alpha_t - i\partial_t \alpha_t)}{\mathcal{N}_\delta}\boldsymbol{\xi}_{t}} \\
		 & \underset{\delta \rightarrow0}{\longrightarrow} 2 \hbar^{1/2} \Imm \sum_{j=1}^N \ps{\boldsymbol{\xi}_{t}}{\left(-i \hbar \nabla_j - \hbar^{1/2} \boldsymbol{\hat{\mathbb{A}}}(x_j) - \boldsymbol{\mathbb{A}}_{\alpha_t}(x_j) \right) \left(a(\mathbf{G}_{x_j}) - a^*(\mathbf{G}_{x_j}) \right) \boldsymbol{\xi}_{t}} \\
		 & \quad \quad - 2 \hbar^{1/2} \Imm \ps{\boldsymbol{\xi}_{t}}{a(|\cdot|\alpha_t - i\partial_t \alpha_t)\boldsymbol{\xi}_{t}} \\
		 & = 2 \hbar^{1/2} \Imm \sum_{j=1}^N \ps{\boldsymbol{\xi}_{t}}{\left(-i \hbar \nabla_j - \hbar^{1/2} \boldsymbol{\hat{\mathbb{A}}}(x_j) - \boldsymbol{\mathbb{A}}_{\alpha_t}(x_j) \right) \left(a(\mathbf{G}_{x_j}) - a^*(\mathbf{G}_{x_j}) \right) \boldsymbol{\xi}_{t}}  \\
		 & \quad - 4 \hbar^{1/2} \Imm \sum_{j=1}^N \ps{\boldsymbol{\xi}_{t}}{a(\mathbf{G}_{q_j(t)}) (p_j(t) - \boldsymbol{\mathbb{A}}_{\alpha_t}(q_j(t)))\boldsymbol{\xi}_{t}} \ . \\
	\end{split}
\end{equation}
Here, we first used the fact that
\begin{equation}
	\begin{split}
		\underset{\delta \rightarrow 0}{\lim} \ps{\boldsymbol{\xi}_{t}}{\left( \com{a(f)}{\mathcal{N}_\delta} - a(f) \right)\boldsymbol{\xi}_{t}} & = 0 \\
		\underset{\delta \rightarrow 0}{\lim} \norm{\left( \com{\boldsymbol{\hat{\mathbb{A}}}(x_j)}{\mathcal{N}_\delta} - \left(a(\mathbf{G}_{x_j}) - a^*(\mathbf{G}_{x_j}) \right) \right)\boldsymbol{\xi}_{t}} & = 0 \ , \\
	\end{split}
\end{equation}
see \cite[Equations (3.21) \& (3.22)]{falconi2023derivation}, and then we used \eqref{eq : newton motion eqs PF}. Now, by the fundamental theorem of calculus and dominated convergence we can write 
\begin{equation}
	\begin{split}
		\boldsymbol{\beta}^c(\Psi_t,\alpha_t) & = \boldsymbol{\beta}^c(\Psi_0,\alpha_0) + \int_0^t \underset{\delta \rightarrow 0}{\lim} \frac{\textnormal{d}}{\textnormal{d}s} \boldsymbol{\beta}_\delta^c(\Psi_s,\alpha_s) \textnormal{d}s \ . \\
	\end{split}
\end{equation}
Note that we can actually apply dominated convergence since it is possible to bound the quantities in the second line of \eqref{eq : derivative beta c regularized} uniformly in $\delta$, see \cite[Equation (3.13)]{falconi2023derivation}. Consequently, a computation using the shifting property \eqref{eq : shifting pty weyl} yields :
\begin{equation} \label{eq : integral formuma beta c 1}
	\begin{split}
		\boldsymbol{\beta}^c(\Psi_t,\alpha_t) - & \boldsymbol{\beta}^c(\Psi_0,\alpha_0) \\
        & = \int_0^t 2 \Imm \sum_{j=1}^N \sum_{\lambda=1,2} \int_{\mathbb{R}^3} \ps{\Psi_s}{ \left(-i \hbar \nabla_j - \hbar^{1/2} \boldsymbol{\hat{\mathbb{A}}}(x_j) \right)\overline{\mathbf{G}_{x_j}(k,\lambda)} \left(\hbar^{1/2}a_{k,\lambda} - \alpha_s(k,\lambda) \right) \Psi_s}\textnormal{d}k \textnormal{d}s \\
		& \quad -\int_0^t 2 \Imm \sum_{j=1}^N \ps{\Psi_s}{\left(-i \hbar \nabla_j - \hbar^{1/2} \boldsymbol{\hat{\mathbb{A}}}(x_j) \right)\mathbf{G}_{x_j}(k,\lambda) \left(\hbar^{1/2}a^*_{k,\lambda} - \overline{\alpha_s(k,\lambda)} \right)\Psi_s} \textnormal{d}k \textnormal{d}s \\
		& \quad - \int_0^t 4 \Imm \sum_{j=1}^N \sum_{\lambda=1,2} \int_{\mathbb{R}^3} \Tilde{p}_j(s) \ps{\Psi_s}{\overline{\mathbf{G}_{q_j(s)}(k,\lambda)}\left(\hbar^{1/2}a_{k,\lambda} -  \alpha_s(k,\lambda) \right)\Psi_s}\textnormal{d}k \textnormal{d}s \ . \\
	\end{split}
\end{equation}
Another calculation using the CCR gives :
\begin{equation}
	\com{-i \hbar \nabla_j - \hbar^{1/2} \boldsymbol{\hat{\mathbb{A}}}(x_j)}{\hbar^{1/2}a^*_{k,\lambda} - \overline{\alpha_t(k,\lambda)}} = - \hbar \overline{\mathbf{G}_{x_j}(k,\lambda)} \ , \\
\end{equation}
which combined with the above allows to recover Lemma \ref{lemma : growth of beta c}. \\

In order to prove Lemma \ref{lemma : growth of beta c} for a general $\sigma \in [1/2,1]$, note that $\mathfrak{h}^1$ is dense in $\mathfrak{h}^\sigma$. Hence, let us consider $u_0 = (q_0, p_0, \alpha_0) \in X^\sigma$, and let $u^{(n)}_0 = \left(q_0, p_0, \alpha_0^{(n)} \right) \in X^1$ such that 
\begin{equation}
    \norm{u_0 - u^{(n)}_0}_{X^\sigma} = \norm{\alpha_0 - \alpha_0^{(n)}}_{\mathfrak{h}^\sigma} \underset{n \rightarrow + \infty}{\longrightarrow} 0 \ . \\
\end{equation}
Let $t \mapsto u(t) = (q(t), p(t), \alpha_t) \in \mathcal{C}(\mathbb{R}; X^\sigma) \cap \mathcal{C}^1(\mathbb{R}; X^{\sigma-1})$ and $t \mapsto u^{(n)}(t) = \left(q^{(n)}(t), p^{(n)}(t), \alpha^{(n)}_t\right) \in \mathcal{C}(\mathbb{R}; X^1) \cap \mathcal{C}^1(\mathbb{R}; X^{0})$ be the solutions to \eqref{eq : newton motion eqs PF} with initial data $u_0$ and $u_0^{(n)}$ respectively. Then, from \eqref{eq : second bounds in a priori for NM}, we have 
\begin{equation}
    \begin{split}
        \norm{\alpha_t - \alpha^{(n)}_t}_{\mathfrak{h}^\sigma} & \leq \norm{u(t) - u^{(n)}(t)}_{X^\sigma} \\
        & \leq Ce^{ct} \norm{u_0 - u^{(n)}_0}_{X^\sigma}\\
        & \underset{n \rightarrow + \infty}{\longrightarrow} 0 \ . \\
    \end{split}
\end{equation}
Let us denote $\textnormal{I} = \eqref{eq : beta c term I}$ and $\textnormal{I} = \eqref{eq : beta c term I}$. From the beginning of the proof, we know that there holds
\begin{equation} \label{eq : change of beta c in time}
             \boldsymbol{\beta}^c\left(\Psi_t, \alpha^{(n)}_t\right) - \boldsymbol{\beta}^c\left(\Psi_0,\alpha^{(n)}_0\right) = \textnormal{I}_n + \textnormal{II}_n \ , 
    \end{equation}
    where $\textnormal{I}_n$ and $\textnormal{II}_n$ denote respectively $\textnormal{I}$ and $\textnormal{II}$ where $\alpha_t$ is replaced by $\alpha^{(n)}_t$. We thus conclude the proof by showing that for all $t \geq 0$,
    \begin{equation} \label{eq : points to prove to conclude beta c}
        \begin{split}
            \left \lvert \boldsymbol{\beta}^c\left(\Psi_t, \alpha^{(n)}_t\right) - \boldsymbol{\beta}^c\left(\Psi_t, \alpha_t\right)  \right \lvert & \underset{n \rightarrow + \infty}{\longrightarrow} 0 \\
            \left \lvert \textnormal{I}_n -  \textnormal{I} \right \lvert \ , \ \left \lvert \textnormal{II}_n -  \textnormal{II} \right \lvert & \underset{n \rightarrow + \infty}{\longrightarrow} 0 \ . \\
        \end{split}
    \end{equation}
    A straightforward computation using \eqref{eq : rewriting for beta c} yields
    \begin{equation}
         \boldsymbol{\beta}^c\left(\Psi_t, \alpha^{(n)}_t\right) - \boldsymbol{\beta}^c\left(\Psi_t, \alpha_t\right)  = 2 \Ree \  \hbar^{1/2} \sum_{\lambda = 1,2} \int_{\mathbb{R}^3} \ps{\Psi_t}{a^*_{k,\lambda}\Psi_t} \left( \alpha_t(k,\lambda) - \alpha_t^{(n)}(k,\lambda) \right) \textnormal{d}k + \norm{ \alpha^{(n)}_t}^2_{\mathfrak{h}^0} - \norm{ \alpha_t}^2_{\mathfrak{h}^0} \ . \\
    \end{equation}
    Hence, using the reversed triangle inequality, the Cauchy-Schwarz inequality and the fact that $\norm{\cdot}_{\mathfrak{h}^0} \leq \norm{\cdot}_{\mathfrak{h}^\sigma}$, we have 
    \begin{equation}
        \begin{split}
            \left \lvert \boldsymbol{\beta}^c\left(\Psi_t, \alpha^{(n)}_t\right) - \boldsymbol{\beta}^c\left(\Psi_t, \alpha_t\right)  \right \lvert & \lesssim \norm{\mathcal{N}^{1/2}\Psi_t}\norm{\alpha_t - \alpha^{(n)}_t}_{\mathfrak{h}^\sigma} + \norm{\alpha_t - \alpha^{(n)}_t}_{\mathfrak{h}^\sigma}^2 \\
             & \underset{n \rightarrow + \infty}{\longrightarrow} 0 \ . \\
        \end{split}
    \end{equation}
    The second line in \eqref{eq : points to prove to conclude beta c} follows from similar arguments. \\
\end{proof}

\section{Relating the variance in kinetic momentum with the variance in momentum} \label{appendix : equivalence beta b}
Here we prove that the second bound in \eqref{eq : cvgce in variance thm} can be obtained without estimating the growth of the LHS quantity, that is $\norm{ \left(-i \hbar \nabla - p(t)  \right)\Psi_t}$, but rather by introducing the functional $\boldsymbol{\beta}^b$ which is defined by means of the quantum \textit{kinetic momentum} $-i\hbar\nabla -  \hbar^{1/2} \boldsymbol{\hat{\mathbb{A}}}(x)$ and its classical counterpart.  As already pointed out, this is crucial to obtain important cancellations.  We define
\begin{equation} \label{eq : def beta b tilda}
            \boldsymbol{\Tilde{\beta}}^b \left(\Psi, p\right) =  \norm{ \left(-i \hbar \nabla - p  \right)\Psi}^2 \ .\\
\end{equation}
In analogy to Definition \ref{def : Definition of the functional},   we define the functional 
\begin{equation}
    \boldsymbol{\Tilde{\beta}}(\Psi,u) = \boldsymbol{\beta}^a(\Psi,q) + \boldsymbol{\Tilde{\beta}}^b(\Psi,p) + \boldsymbol{\beta}^c(\Psi,\alpha) \ , \\
\end{equation}
with the same domain assumptions as in Definition \eqref{def : Definition of the functional}. Then,   the following lemma holds.
\begin{lemma}  \label{eq : lemma equivalence beta b beta b tilda}
     Let $\Psi \in \mathcal{D} \left( \left( - \Delta \right)^{1/2} \right) \cap \mathcal{D}(\mathbf{x}) \cap \mathcal{D}\left(\mathcal{N}^{1/2}\right)$. Let $\sigma \in [1/2,1]$, and $u = (q,p,\alpha) \in X^\sigma$. Then, there exists a constant $C > 0$ such that
     \begin{equation} \label{eq : appendix beta b}
          \boldsymbol{\beta}^b(\Psi,u) \leq C ( \boldsymbol{\Tilde{\beta}}(\Psi,u) + \hbar ) \\
     \end{equation}
     and 
     \begin{equation} \label{eq : appendix beta b tilda}
          \boldsymbol{\Tilde{\beta}} ^b(\Psi,p) \leq C \left( \boldsymbol{\beta}(\Psi,u) + \hbar \right) \ . \\
     \end{equation}
\end{lemma}

\begin{proof}
Recall the notation $\Tilde{p}_j = p_j - \boldsymbol{\mathbb{A}}_{\alpha}(q_j)$. Using twice the triangle inequality, we have 
\begin{equation}
\begin{split}
&\left \lvert \norm{\left(-i\hbar \nabla_j - p_j \right)\Psi} - \norm{\left(-i\hbar \nabla_j - \hbar^{1/2} \boldsymbol{\hat{\mathbb{A}}}(x_j) -\Tilde{p}_j  \right)\Psi} \right \lvert
\\
&\quad 
 \leq \norm{\left(\hbar^{1/2} \boldsymbol{\hat{\mathbb{A}}}(x_j) - \boldsymbol{\mathbb{A}}_{\alpha}(x_j) \right)\Psi} +  \norm{\left(\boldsymbol{\mathbb{A}}_{\alpha}(x_j) - \boldsymbol{\mathbb{A}}_{\alpha}(q_j) \right)\Psi} \ . \\
\end{split}
\end{equation}
Hence applying the inequality $(a+b)^2 \leq 2(a^2 + b^2)$ yields
\begin{equation}
\begin{split}
	\boldsymbol{\beta}^b(\Psi, p) & \lesssim \boldsymbol{\Tilde{\beta}}^b(\Psi, p) + \sum_{j=1}^N \norm{\left(\hbar^{1/2} \boldsymbol{\hat{\mathbb{A}}}(x_j) - \boldsymbol{\mathbb{A}}_{\alpha}(x_j) \right)\Psi}^2 + \sum_{j=1}^N \norm{\left(\boldsymbol{\mathbb{A}}_{\alpha}(x_j) - \boldsymbol{\mathbb{A}}_{\alpha}(q_j) \right)\Psi}^2 \ . \\
\end{split}
\end{equation}
Now, in view of arguments already used before (see for example \eqref{eq : bound for T j ell m 1 in beta b term 1} and \eqref{eq : control for Ibetab2}), we can bound
\begin{equation} \label{eq : ineqs beta c beta a appendix}
\begin{split}
	\sum_{j=1}^N \norm{\left(\hbar^{1/2} \boldsymbol{\hat{\mathbb{A}}}(x_j) - \boldsymbol{\mathbb{A}}_{\alpha}(x_j) \right)\Psi}^2 & \lesssim  \boldsymbol{\beta}^c(\Psi,\alpha) + \hbar  \\
	\sum_{j=1}^N \norm{\left(\boldsymbol{\mathbb{A}}_{\alpha}(x_j) - \boldsymbol{\mathbb{A}}_{\alpha}(q_j) \right)\Psi}^2 & \lesssim \boldsymbol{\beta}^a(\Psi, q) \ , \\
\end{split}
\end{equation}
from which we deduce \eqref{eq : appendix beta b}. \eqref{eq : appendix beta b tilda} is obtained similarly. \\
\end{proof}

\section*{Acknowledgements}

Funded by the Deutsche Forschungsgemeinschaft (DFG, German Research Foundation) -- Project number 580841385.  The authors gratefully acknowledge this support.

\printbibliography

\end{document}

---------- DELETED IN THE MAIN SCRIPT / APPENDIX D ----------

\section{Check of Assumption \eqref{assumptions : init 1} for a Gaussian particle} \label{appendix : gaussian particle}

The standard Gaussian wave function or wave packet for given position and momentum $(q,p) \in \mathbb{R}^3 \times \mathbb{R}^3$ and with variance $\hbar$ is defined by :
\begin{equation} \label{eq : gaussian wave packet 3d}
\psi^\hbar_{q,p}(x) = (\pi \hbar)^{-3/4} \exp \left( - \frac{|x - q|^2}{2 \hbar} \right) \exp \left( \frac{i}{\hbar}p \cdot \left(x - q \right)\right) \ , \ \forall x \in \mathbb{R}^3 \ . \\
\end{equation}
We place ourselves in the setting where the initial data has the following structure :
\begin{equation} \label{eq : product stucture initial data}
    \Psi = \psi^\hbar_{q,p} \otimes W(\hbar^{-1/2}\alpha) \Omega \ . \\
\end{equation}
Note then that in view of this product form, we only have to show that the variance in position and in shifted momentum against $\psi^\hbar_{q,p}$ are small. Let us compute the variance in position. Since the second exponential in \eqref{eq : gaussian wave packet 3d} is of modulus $1$, we compute, with the change of variables $x \leftarrow \hbar^{-1/2} \left( x - q \right)$ :
\begin{equation}
\begin{split}
	\norm{(x - q)\psi^\hbar_{q,p}}^2_{L^2(\mathbb{R}^3)} = \pi^{-3/2} \hbar \int_{\mathbb{R}^3}|x|^2e^{-|x|^2} \textnormal{d}x = \frac{3}{2} \hbar \ .
\end{split}
\end{equation}
Here we used the standard result that $\int_{\mathbb{R}^3}|x|^2e^{-|x|^2} \textnormal{d}x = \frac{3}{2}\pi^{3/2} $. Next, we compute the variance in shifted momentum. Note first that there holds
\begin{equation}
\boldsymbol{\beta}^c(\Psi, \alpha) = \hbar \norm{\Psi_0}^2_{L^2(\mathbb{R}^{3N})} \ps{\Omega}{\mathcal{N}\Omega}_{\mathfrak{F}} = 0 \ .
\end{equation}
Hence, using Lemma \eqref{eq : lemma equivalence beta b beta b tilda} with $N=1$, $t=0$ and the specific choice of \eqref{eq : gaussian wave packet 3d}, we see that 
\begin{equation}
\norm{\left(-i \hbar \nabla - \hbar^{1/2}\boldsymbol{\hat{\mathbb{A}}}(x) - \left( p -\boldsymbol{\mathbb{A}}_{\alpha}(q) \right)\right) \psi^\hbar_{\mathbf{X}_0,p}}^2_{L^2(\mathbb{R}^{3})} \lesssim \norm{\left(-i \hbar \nabla - p \right)\psi^\hbar_{q,p}}^2_{L^2(\mathbb{R}^{3})} + \hbar
\end{equation}
Let us then compute the norm in the RHS. In view \eqref{eq : gaussian wave packet 3d} and using Parseval, we can write
\begin{equation}
\begin{split}
\norm{\left(-i \hbar \nabla - p \right)\psi^\hbar_{q,p}}^2_{L^2(\mathbb{R}^{3})} & = \norm{\left( \hbar k - p \right) \mathcal{F}\left(\psi^\hbar_{q,p}\right)}^2_{L^2(\mathbb{R}^{3})} \\
& = (\pi \hbar)^{-3/2} \norm{\left( \hbar k - p \right) \mathcal{F}\left(G_{1 / 2\hbar}\right)\left( \frac{p}{\hbar} -k \right)}^2_{L^2(\mathbb{R}^{3})} \\
\end{split}
\end{equation}
Here we defined the Gaussian function $G_\alpha(\xi) = e^{-\alpha |\xi|^2}$. Now, recall that for $\Ree(\alpha) > 0$, there holds
\begin{equation}
\mathcal{F}\left( G_\alpha \right)(x) = (2\alpha)^{-3/2} G_{1/4 \alpha}(x)
\end{equation}
Thus we find :
\begin{equation}
\begin{split}
	\norm{\left(-i \hbar \nabla - p \right)\psi^\hbar_{q,p}}^2_{L^2(\mathbb{R}^{3})} & = (\pi \hbar)^{-3/2} \hbar^3 \int_{\mathbb{R}^3} |\hbar k - p|^2 e^{- |\hbar k - p|^2/\hbar} \textnormal{d}k \\
	& = \frac{3}{2}\hbar
\end{split}
\end{equation}
In the last line, we used again the standard gaussian formula from above. This allows us to conclude. \\

---------- OLD DUHAMEL FORMULA FOR THE A PRIORI ESTIMATES ON THE NM SOLUTIONS ----------

Here instead of writing a Duhamel formula for each component of $u(t)$, we write the Duhamel for $u(t)$ as in \cite{ammari2022derivationclassicalelectrodynamicscharges} : \begin{equation} \label{eq : duhamel NM}
    u(t) = \mathcal{D}_0(t)u_0 + \int_0^t \mathcal{D}_0(t-s) F(u(s)) \textnormal{d}s
\end{equation}
where $\mathcal{D}_0(.)$ is the \textit{free field flow}, which action on $X^\sigma$ is given by 
\begin{equation}
    \mathcal{D}_0(t)(p,q,\alpha) = (p,q,e^{-it|.|}\alpha) \ , \\
\end{equation}
and $F$ is the \textit{nonlinearity} of the Newton--Maxwell system, given by :
\begin{equation}
    \begin{split}
        F(u)_{q_j} & = 2 \Tilde{p}_j \\
        F(u)_{p_j} & = 2 \sum_{\ell= 1}^d \Tilde{p}_j^\ell \ \nabla \boldsymbol{\mathbb{A}}_{\alpha}^\ell (q_j) - \sum_{k \neq j} \left( \nabla V \right) \left(q_j - q_k \right) \\
        F(u)_\alpha(k,\lambda) &= 2i \sum_{j=1}^N \mathbf{G}_{q_j}(k,\lambda) \cdot \Tilde{p}_j \\
    \end{split}
\end{equation}
where we used again the notation $\Tilde{p}_j$ for the kinetic momentum. A direct application of Lemma \eqref{lemma : estimates on the classical elec potential} gives the following bounds :
\begin{equation} \label{eq : estimates on the nonlinearity}
    \begin{split}
        \left \lvert F(u)_{q_j} \right \lvert & \lesssim_{\kappa} |p_j| + \norm{\alpha}_{\dot{\mathfrak{H}}^{1/2}} \\
         \left \lvert F(u)_{p_j} \right \lvert & \lesssim_{\kappa} \norm{\alpha}_{\dot{\mathfrak{h}}^{1/2}} \sum_{\ell=1}^d \left( |p^\ell_j| + \norm{\alpha}_{\dot{\mathfrak{h}}^{1/2}} \right) + 1\\
         \norm{F(u)_\alpha}_{\mathfrak{h}^\sigma} & \lesssim_{\kappa, \sigma}  \sum_{j=1}^N \left( |p_j| + \norm{\alpha}_{\dot{\mathfrak{h}}^{1/2}} \right) \ . \\
    \end{split}
\end{equation}
Now, as an immediate consequence of the conservation of the energy (see Lemma \ref{lemma : conservation of the energy}) and \eqref{eq : pointwise bounds for V and nabla V}, we have
\begin{equation}
    \begin{split}
        \norm{\alpha_t}_{\dot{\mathfrak{h}}^{1/2}} & \lesssim_\kappa \mathcal{H}(u_0)^{1/2} + 1 \\
        |p_j(t)| & \lesssim_\kappa \mathcal{H}(u_0)^{1/2} + 1 \ . \\
    \end{split}
\end{equation}
This together with \eqref{eq : estimates on the nonlinearity}, \eqref{eq : duhamel NM} and the triangle inequality yield :
\begin{equation}
    \begin{split}
        \norm{u(t)}_{X^\sigma} & \leq \norm{u_0}_{X^\sigma} + \int_0^t \left( \sum_{j=1}^N\left( \left \lvert F(u(s))_{q_j} \right \lvert^2 + \left \lvert F(u(s))_{p_j} \right \lvert^2 \right) + \norm{F(u(s))_\alpha}^2_{\mathfrak{h}^\sigma}  \right)^{1/2}\textnormal{d}s \\
        & \leq \norm{u_0}_{X^\sigma} + Ct \left( \mathcal{H}(u_0) + 1 \right) \\
    \end{split}
\end{equation}
Next, using $\norm{\alpha}_{\dot{\mathfrak{h}}^{1/2}} \leq \norm{\alpha}_{\mathfrak{h}^\sigma}$, we have 

from which follows \eqref{eq : first bounds in a priori for NM}. \\

---------- CONSERVATION OF THE ENERGY ----------

\section{Conservation of the energy for the Newton--Maxwell system}
Here we state and prove the conservation of the energy for the Newton--Maxwell system, by using the Duhamel formulation of the equations, with the subtelty that we write a Duhamel formula for $|\Tilde{p}_j(t)|^2$ rather than $p_j(t)$, which allows for a suitable cancellation. 

\begin{lemma}[Conservation of the energy] \label{lemma : conservation of the energy}
    Let $\sigma \in [1/2,1]$. Let $u_0 = (q_0, p_0, \alpha_0) \in X^\sigma$ and consider $u(t) = (q(t), p(t), \alpha_t)$ the associated solution to \eqref{eq : newton motion eqs PF}. Then, the corresponding energy functional $\mathcal{H}$ defined in \eqref{eq : energy NM} is conserved along the flow : for all $t \geq 0$, there holds
    \begin{equation} \label{eq : Conservation of the energy}
        \mathcal{H}(u(t)) = \mathcal{H}(u_0) \ . \\
    \end{equation}
\end{lemma}
\begin{proof}
First, let us note that, using the symmetry of $\nabla V$ and the chain rule, there holds
\begin{equation}
    \dt \sum_{1 \leq j < k \leq N} V(q_j(t) - q_k(t)) = 2 \sum_{j=1}^N \sum_{k \neq j} (\nabla V)(q_j(t) - q_k(t)) \cdot \Tilde{p}_j(t) \ , \\
\end{equation}
so that, from the fundamental theorem of calculus, 
\begin{equation}
     2 \int_0^t \sum_{j=1} \sum_{k \neq j} \Tilde{p}_j(s) \cdot (\nabla V)(q_j(s) - q_k(s)) \textnormal{d}s = \sum_{1 \leq j < k \leq N} V(q_j(t) - q_k(t)) - \sum_{1 \leq j < k \leq N} V(q_{j,0} - q_{k,0}) \ . \\
\end{equation}
Let us now recall the Duhamel formulas for $|\Tilde{p}_j(t)|^2 = \left \lvert p_j(t) - \kappa \ast \mathbf{A}_{\alpha_t}(q_j(t)) \right \lvert^2$ and $\alpha_t$, which are computed in the proof of Lemma \eqref{lemma : A priori estimates for the Newton--Maxwell solutions} :
    \begin{equation}
    \begin{split}
        |\Tilde{p}_j(t)|^2 & = |\Tilde{p}_{j,0}|^2 + 2 \int_0^t \Tilde{p}_j(s) \cdot  \boldsymbol{\mathbb{E}}_{\alpha_s}(q_j(s)) \textnormal{d}s - 2 \int_0^t \sum_{k \neq j} \Tilde{p}_j(s) \cdot (\nabla V)(q_j(s) - q_k(s)) \textnormal{d}s \\
        \alpha_t(k,\lambda) & = e^{-it|k|}\alpha_0(k,\lambda) + 2i \int_0^t e^{-i(t-s)|k|} \sum_{j=1}^N \mathbf{G}_{q_j(s)}(k,\lambda) \cdot \Tilde{p}_j(s)  \textnormal{d}s \ . \\ 
    \end{split}
\end{equation}
It will be useful to recall the Fourier expansion of the regularized classical $\mathbf{A}$ and $\mathbf{E}$ fields :
\begin{equation}
    \begin{split}
        \boldsymbol{\mathbb{A}}_{\alpha}(q) & = 2 \Ree \sum_{\lambda= 1,2} \int_{\mathbb{R}^3} \frac{\mathcal{F}[\kappa](k)}{\sqrt{2|k|}} \boldsymbol{\varepsilon}_\lambda(k) e^{i k \cdot q} \alpha(k,\lambda) \textnormal{d}k \\
        \boldsymbol{\mathbb{E}}_{\alpha}(q) & = 2 \Ree \ i \sum_{\lambda= 1,2} \int_{\mathbb{R}^3} \sqrt{\frac{|k|}{2}} \mathcal{F}[\kappa](k)\boldsymbol{\varepsilon}_\lambda(k) e^{i k \cdot q} \alpha(k,\lambda) \textnormal{d}k \ . \\
    \end{split}
\end{equation}
The energy for the Newton--Maxwell system, using the kinetic momentum $\Tilde{p}_j$, writes :
\begin{equation}
    \mathcal{H}(u) = \sum_{j=1}^N |\Tilde{p}_j|^2 + \sum_{1 < j \leq k < N}V(q_j - q_k) + \norm{\alpha}^2_{\dot{\mathfrak{h}}^{1/2}} \ . \\
\end{equation}
We thus insert the Duhamel formula for $\alpha_t$ in $\boldsymbol{\mathbb{E}}_{\alpha_t}(q_j(t))$. We have :
\begin{equation} \label{eq : expression for E field with duhamel alpha}
    \begin{split}
        \boldsymbol{\mathbb{E}}_{\alpha_t}(q_j(t)) & = 2 \Ree \ i \sum_{\lambda= 1,2} \int_{\mathbb{R}^3} \sqrt{\frac{|k|}{2}} \mathcal{F}[\kappa](k)\boldsymbol{\varepsilon}_\lambda(k) e^{i k \cdot q_j(t)} e^{-i |k|t} \alpha_0(k,\lambda) \textnormal{d}k \\
        & \quad - 2\Ree \sum_{\lambda= 1,2} \sum_{k=1}^N \int_{\mathbb{R}^3} \int_0^t |\mathcal{F}[\kappa](k)|^2 e^{i k \cdot (q_j(t) - q_k(s))} e^{-i |k|(t-s)} \boldsymbol{\varepsilon}_\lambda(k) \boldsymbol{\varepsilon}_\lambda(k) \cdot \Tilde{p}_k(s) \textnormal{d}s \textnormal{d}k \\
        & := \mathcal{T}_{j  t}^{(1)} + \mathcal{T}_{j t}^{(2)} \ . \\
    \end{split}
\end{equation}
Let us now compute $\norm{\alpha_t}^2_{\dot{\mathfrak{h}}^{1/2}}$. Using the Duhamel formula for $\alpha_t$ and expanding $|\alpha_t(k,\lambda)|^2$, we have :
\begin{equation}
    \begin{split}
        \norm{\alpha_t}^2_{\dot{\mathfrak{h}}^{1/2}} & = \norm{\alpha_0}^2_{\dot{\mathfrak{h}}^{1/2}} - 4 \Ree \ i \sum_{\lambda= 1,2} \sum_{j=1}^N \int_{\mathbb{R}^3} \int_0^t \sqrt{\frac{|k|}{2}} \mathcal{F}[\kappa](k)\boldsymbol{\varepsilon}_\lambda(k) \cdot \Tilde{p}_j(s) e^{i k \cdot q_j(s)} e^{-i |k|s} \alpha_0(k,\lambda) \textnormal{d}s \textnormal{d}k \\
        & \quad \quad + 2 \sum_{\lambda = 1,2} \sum_{j,r = 1}^N \int_{\mathbb{R}^3} \int_0^t\int_0^t |\mathcal{F}[\kappa](k)|^2 e^{i k \cdot (q_j(s) - q_r(\tau))} e^{-i |k|(\tau-s)} \boldsymbol{\varepsilon}_\lambda(k) \cdot \Tilde{p}_j(s) \boldsymbol{\varepsilon}_\lambda(k) \cdot \Tilde{p}_r(\tau) \textnormal{d}s \textnormal{d}\tau \textnormal{d}k \\
        & = \norm{\alpha_0}^2_{\dot{\mathfrak{h}}^{1/2}} - 2 \sum_{j=1}^N\int_0^t \Tilde{p}_j(s) \cdot \mathcal{T}^{(1)}_{js} \textnormal{d}s -  2 \sum_{j=1}^N\int_0^t \Tilde{p}_j(s) \cdot \mathcal{T}^{(2)}_{js} \textnormal{d}s \\
        & = \norm{\alpha_0}^2_{\dot{\mathfrak{h}}^{1/2}} - \sum_{j=1}^N \int_0^t \Tilde{p}_j(s) \cdot \boldsymbol{\mathbb{E}}_{\alpha_s}(q_j(s)) \textnormal{d}s \ , \\
    \end{split}
\end{equation}
where we used \eqref{eq : expression for E field with duhamel alpha} in the second and third lines. Gathering the above calculations, we see that all the non constant terms in $t$ cancel, leading to
\begin{equation}
    \begin{split}
        \mathcal{H}(u(t))  = \sum_{j=1}^N|p_{j,0}|^2 + \sum_{1 \leq j < k \leq N} V(q_{j,0} - q_{k,0}) + \norm{\alpha_0}^2_{\dot{\mathfrak{h}}^{1/2}} 
         = \mathcal{H}(u_0) \ , \\
    \end{split}
\end{equation}
as claimed. \\
\end{proof}

---------- OLD PROOF FOR CONVERGENCE OF TERMS IN BETA B ----------

 Using that $\Tilde{p}^m_j(t)$ and $F^{j \ell m}_t$ are scalars, expanding the inner product using the definitions of $\mathbb{G}^\ell_{jt}$ and $\hat{F}^{j  \ell m}$, we have, for all $s \in [0,t]$,
\begin{equation}
    \begin{split}
         \left \lvert \mathcal{E}^{(b,1)}_{n,s} - \mathcal{E}^{(b,1)}_s \right \lvert  \lesssim \sum_{j=1}^N \sum_{\ell, m=1}^3 C_{j \ell m}(s) \bigg (  & \left \lvert \ps{\Psi_s^{(n)}}{P_j^\ell  \Psi_s^{(n)}} - \ps{\Psi_s}{P_j^\ell  \Psi_s}
        \right \lvert + \left \lvert  \norm{\Psi_s^{(n)}}^2 - \norm{\Psi_s}^2  \right \lvert \\
        + &  \left \lvert \ps{\Psi_s^{(n)}}{P_j^\ell \partial_{jm} \boldsymbol{\hat{\mathbb{A}}}^\ell (x)  \Psi_s^{(n)}} - \ps{\Psi_s}{P_j^\ell \partial_{jm}  \boldsymbol{\hat{\mathbb{A}}}^\ell (x)  \Psi_s}
        \right \lvert \\
        + &  \left \lvert \ps{\Psi_s^{(n)}}{P_j^\ell \partial_{j\ell} \boldsymbol{\hat{\mathbb{A}}}^m (x)  \Psi_s^{(n)}} - \ps{\Psi_s}{P_j^\ell \partial_{j\ell}\boldsymbol{\hat{\mathbb{A}}}^m (x)  \Psi_s}
        \right \lvert \\
        + & \left \lvert \ps{\Psi_s^{(n)}}{\partial_{jm}  \boldsymbol{\hat{\mathbb{A}}}^\ell(x)  \Psi_s^{(n)}} - \ps{\Psi_s}{\partial_{jm}  \boldsymbol{\hat{\mathbb{A}}}^\ell (x)  \Psi_s}
        \right \lvert  + (\ell \longleftrightarrow m) \bigg) \\
    \end{split}
\end{equation}
where $\left( m \longleftrightarrow \ell \right)$ is a shorthand notation for the expression preceding it with the indices $\ell$ and $m$ exchanged. The constant $C_{j \ell m}(s)$ depends on $s$ via $\Tilde{p}^m_j(s)$, $\partial_{jm} \boldsymbol{\mathbb{A}}^\ell_{\alpha_s} (q_j(s))$, and the same quantities with $\ell$ and $m$ exchanged. Using \eqref{lemma : estimates on the classical elec potential} and \eqref{eq : first bounds in a priori for NM}, these quantities are uniformly bounded in time by a constant $C > 0$. Now, note that $\norm{\Psi_s^{(n)}}^2 - \norm{\Psi_s}^2$ goes uniformly to zero since $e^{-i\frac{s}{\hbar}\mathbb{H}_{\hbar}}$ is unitary. Next, using \eqref{eq : inequality to show cvgce beta b n beta b}, \eqref{eq : estimates in lemma to show cvgce} and the fact that $\com{\mathbb{H}_{\hbar}}{e^{-i\frac{s}{\hbar}\mathbb{H}_{\hbar}}} = 0$, we have the following bounds :
\begin{equation} \label{eq : unif bound for ps P ell for cvgce}
    \begin{split}
        & \sup_{s \in [0,t]} \left \lvert \ps{\Psi_s^{(n)}}{P^\ell  \Psi_s^{(n)}} - \ps{\Psi_s}{P^\ell  \Psi_s}
        \right \lvert \ , \ \sup_{s \in [0,t]} \left \lvert \ps{\Psi_s^{(n)}}{\partial_m  \boldsymbol{\hat{\mathbb{A}}}^\ell(x)  \Psi_s^{(n)}} - \ps{\Psi_s}{\partial_m  \boldsymbol{\hat{\mathbb{A}}}^\ell (x)  \Psi_s}
        \right \lvert \\
        & \quad \quad \quad \quad \quad \quad \quad \quad \quad \quad \lesssim_{N, \hbar, \kappa} \norm{\left( \mathbb{H}_{\hbar} + 1\right)^{1/2}\left(\Psi^{(n)} - \Psi_0 \right)} \norm{\Psi^{(n)}} + \norm{\left( \mathbb{H}_{\hbar} + 1\right)^{1/2}\Psi} \norm{\Psi^{(n)} - \Psi_0}  \ , \\
    \end{split}
\end{equation}
and 
\begin{equation}
    \begin{split}
                & \sup_{s \in [0,t]} \left \lvert \ps{\Psi_s^{(n)}}{P^\ell \partial_m  \boldsymbol{\hat{\mathbb{A}}}^\ell(x)  \Psi_s^{(n)}} - \ps{\Psi_s}{P^\ell \partial_m \boldsymbol{\hat{\mathbb{A}}}^\ell(x) \Psi_s}
        \right \lvert \\
        & \quad \quad \quad \quad \quad \quad \quad \quad \quad \quad \lesssim_{N, \hbar, \kappa} \norm{\left( \mathbb{H}_{\hbar} + 1\right)^{1/2}\left(\Psi^{(n)} - \Psi_0 \right)} \left( \norm{\left( \mathbb{H}_{\hbar} + 1\right)^{1/2}\Psi^{(n)}} + \norm{\left( \mathbb{H}_{\hbar} + 1\right)^{1/2}\Psi} \right) \ . \\
    \end{split}
\end{equation}
where we dropped the $j$ subscript for clarity. This allows us to conclude that 
\begin{equation}
    \sup_{s \in [0,t]}  \left \lvert \mathcal{E}^{(b,1)}_{n,s} - \mathcal{E}^{(b,1)}_s \right \lvert \underset{n \rightarrow + \infty}{\longrightarrow} 0 \ , \\
\end{equation}
as desired. \\